\documentclass[aps,pra,superscriptaddress,eps2pdf,twocolumn]{revtex4-1}

\usepackage{epsfig}
\usepackage{amsfonts}
\usepackage{amssymb}
\usepackage{mathrsfs}
\usepackage{amsthm}
\usepackage{amsmath}
\usepackage{times}
\usepackage{color}
\usepackage{soul,enumerate}
\usepackage[colorlinks=false,hidelinks]{hyperref}

\usepackage{algorithm,algpseudocode}
\usepackage{algorithmicx}

\definecolor{hp}{rgb}{.6,0.7,.5}
\definecolor{earl}{rgb}{.7,0.1,.9}
\definecolor{jrscol}{rgb}{0.5,0.07,0.55}
   
      \usepackage{graphics,tikz}
\usetikzlibrary{calc}
\usepackage[caption=false]{subfig}

\usepackage{ifpdf}
\ifpdf
\usepackage{epstopdf}   
\fi

\newcommand{\hakop}[1]{\textcolor{black}{#1}} 
 
\newcommand{\br}[1]{{\color{black} #1}}
\newcommand{\jrs}[1]{\textcolor{black}{#1}}
 
\newcommand{\prx}[1]{{\color{black} #1}}

\let\br\relax
\let\hakop\relax

\newtheorem{theorem}{Theorem}
\newtheorem{defin}{Definition}

\newtheorem{lem}[theorem]{Lemma}
\newtheorem{corollary}{Corollary}

\newtheorem{proposition}[theorem]{Proposition}

\newcommand{\thistheoremname}{}
\newtheorem*{genericthm}{\thistheoremname}
\newenvironment{namedthm}[1]
  {\renewcommand{\thistheoremname}{#1}%
   \begin{genericthm}}
  {\end{genericthm}}

\newcommand{\abs}[1]{\ensuremath{\left\vert#1\right\vert}}
\newcommand{\ket}[1]{\ensuremath{\vert#1\rangle}}
\newcommand{\bra}[1]{\ensuremath{\langle #1\vert}}
\newcommand{\bk}[2]{\ensuremath{\langle #1\vert #2\rangle}}
\newcommand{\kb}[2]{\ensuremath{\vert #1 \rangle\! \langle #2 \vert}}
\newcommand{\kbself}[1]{\kb{#1}{#1}}
\newcommand{\op}[1]{\kbself{#1}}
\newcommand{\chan}{\ensuremath{\mathcal{E}}}
\newcommand{\braket}[2]{\bk{#1}{#2}}
\newcommand{\braketself}[1]{\bk{#1}{#1}}
\newcommand{\bkself}[1]{\bk{#1}{#1}}
\newcommand{\expected}[1]{\mathbb{E}[#1]}
\newcommand{\VarOmega}{\mathrm{Var}[\braketself{\Omega}]}

\let\bar\overline

\renewcommand{\vec}[1]{\ensuremath{\mathbf{#1}}}

\newcommand{\re}{\ensuremath{\mathrm{Re}}}

\newcommand{\onenorm}[1]{\ensuremath{\| #1 \|_1}}
\newcommand{\norm}[1]{\ensuremath{\| #1 \|}}

\newcommand{\algorithmicbreak}{\textbf{break}}
\newcommand{\BREAK}{\State \algorithmicbreak}
\newcommand{\Psim}{P_{\mathrm{sim}}}
\newcommand{\Pex}{P_{\mathrm{ex}}}
\newcommand{\poly}{\mathrm{poly}}

\def\id{\mbox{\small 1} \!\! \mbox{1}}

\def\id{\mbox{\small 1} \!\! \mbox{1}}

\def\id{{\mathchoice {\rm 1\mskip-4mu l} {\rm 1\mskip-4mu l} {\rm 1\mskip-4.5mu l} {\rm 1\mskip-5mu l}}}

\newcommand{\Tr}{\ensuremath{\mathrm{Tr}}}

 \usepackage{colortbl}
 \makeatletter
\def\thickhline{%
  \noalign{\ifnum0=`}\fi\hrule \@height \thickarrayrulewidth \futurelet
   \reserved@a\@xthickhline}
\def\@xthickhline{\ifx\reserved@a\thickhline
               \vskip\doublerulesep
               \vskip-\thickarrayrulewidth
             \fi
      \ifnum0=`{\fi}}
\makeatother

\newlength{\thickarrayrulewidth}
\newcolumntype{P}[1]{>{\centering\arraybackslash}p{#1}}

\begin{document}
\title{Quantifying quantum speedups: improved classical simulation from tighter magic monotones}
\date{\today}

\author{James R. Seddon}
 \email{jseddonquantum@gmail.com}
\affiliation{Department of Physics and Astronomy, University College London, London, United Kingdom}

\author{Bartosz Regula}
 \email{bartosz.regula@gmail.com}
\affiliation{School of Physical and Mathematical Sciences, Nanyang Technological University, 637371, Singapore}

\author{Hakop Pashayan}
\affiliation{Institute for Quantum Computing and Department of Combinatorics and Optimization, University of Waterloo, ON, N2L 3G1 Canada}
\affiliation{Perimeter Institute for Theoretical Physics, Waterloo, ON, N2L 2Y5 Canada}
\affiliation{Centre for Engineered Quantum Systems, School of Physics, The University of Sydney, Sydney, NSW 2006, Australia}
\author{Yingkai Ouyang}
\affiliation{Department of Physics \& Astronomy, University of Sheffield, Sheffield, S3 7RH, United Kingdom}

\author{Earl T.\ Campbell}
 \email{earltcampbell@gmail.com}
\affiliation{Department of Physics \& Astronomy, University of Sheffield, Sheffield, S3 7RH, United Kingdom}
\affiliation{AWS Center for Quantum Computing, Pasadena, California 91125, USA}

\begin{abstract}

\prx{Consumption of magic states promotes the stabilizer model of computation to universal quantum computation. Here, we propose three different classical algorithms for simulating such universal quantum circuits, and characterize them by establishing precise connections with a family of magic monotones. Our first simulator introduces a new class of quasiprobability distributions and connects its runtime to a generalized notion of negativity. We prove that this algorithm has significantly improved exponential scaling compared to all prior quasiprobability simulators for qubits. Our second simulator is a new variant of the stabilizer-rank simulation algorithm, extended to work with mixed states and with significantly improved runtime bounds. Our third simulator trades precision for speed by discarding negative quasiprobabilities. We connect each algorithm's performance to a corresponding magic monotone and, by comprehensively characterizing the monotones, we obtain a precise understanding of the simulation runtime and error bounds. Our analysis reveals a deep connection between all three seemingly unrelated simulation techniques and their associated monotones. For tensor products of single-qubit states, we prove that our monotones are all equal to each other, multiplicative and efficiently computable, allowing us to make clear-cut comparisons of the simulators' performance scaling. Furthermore, our monotones establish several asymptotic and non-asymptotic bounds on state interconversion and distillation rates. Beyond the theory of magic states, our classical simulators can be adapted to other resource theories under certain axioms, which we demonstrate through an explicit application to the theory of quantum coherence.}

\end{abstract}

\maketitle


\section{Introduction}

\prx{%
Classical simulation of quantum systems has a long and fruitful history. 
Insurmountable obstructions to the classical simulation of quantum systems gave birth to the field of quantum computation~\cite{feynman1982simulating} and the search for quantum computational resources. Despite the computational limitations of classical simulation, surprisingly powerful classical simulators have since been discovered including simulators of stabilizer circuits~\cite{gottesman1998theory,aaronson04improved}, fermionic
linear optics/matchgates~\cite{valiant2002quantum, terhal2002classical, jozsa2008matchgates, Brod2016efficient} and others~\cite{Shi2006tensor, schwarz2013simulating, oszmaniec2018classical, bremner2017achieving,  nest2009simulating, deraedt2019massively}. 
Improvement and characterization of classical simulation algorithms helps benchmark the computational speedups that quantum computers can provide and also provides tools useful in their own right~\cite{google2019supremacy, OakRidge17, temme2017error}.

Stabilizer circuits are initialised in so-called stabilizer states and evolved by stabilizer operations, such that the system stays in a stabilizer state throughout the whole computation.  These circuits are important in fault-tolerant quantum computation and can be efficiently classically simulated by virtue of the Gottesman-Knill theorem~\cite{gottesman1998theory}.  An elegant extension of stabilizer circuits enables them to perform universal computation by allowing the input states to include so-called \emph{magic states}~\cite{bravyi2005universal,campbell2017roads}. Aaronson and Gottesman~\cite{aaronson04improved} showed how to classically simulate such circuits with a runtime that scales exponentially with the number of input magic state qubits, yet still scales efficiently with respect to the number of stabilizer state qubits.  Consequently, we can perform an efficient classical simulation for any class of circuits that is \textit{nearly-stabilizer} in the sense that they use only logarithmically many input magic state qubits.  Subsequent developments showed that the difficulty of simulating a quantum circuit depends not only on the number of magic state inputs, but also on the type of magic that these states possess.

In the pursuit of faster classical simulation of nearly-stabilizer circuits, two leading approaches have emerged: quasiprobability~\cite{Stahlke14, pashayan15, Howard17robustness,OakRidge17,Seddon19} and stabilizer rank--based \cite{Bravyi16stabRank,bravyi2016improved,bravyi2018simulation,qassim2019clifford} simulators.  These simulators all have their runtime determined by a function called a \textit{magic monotone} that quantifies how far the magic states deviate from the set of stabilizer states.  With these modern simulators, even a very large number of magic state inputs is classically tractable, provided the magic states are close enough to stabilizer states, as quantified by the relevant magic monotone. However, different simulators come with their own magic monotone and therefore different runtime scalings. So far, no overarching study has precisely compared the runtimes and monotones for different stabilizer simulators. The difficulty of comparison is exacerbated since some monotones are not easily calculated.  We will next review these simulation methods, before stating our main results that further sharpen the performance of modern simulators and reveal a cohesive picture of a previously fragmented landscape of simulators.

\subsection{Review of prior art}

Quasiprobability simulators work by representing the target quantum state by an operator probabilistically chosen from a discrete set known as a \emph{frame}~\cite{pashayan15,ferrie2008frame}. Examples of relevant frames include the set of density operators corresponding to pure stabilizer states~\cite{Howard17robustness}, the set of Pauli operators~\cite{Rall2019}, and the set of phase point operators~\cite{pashayan15} used in the construction of the discrete Wigner function~\cite{Leonhardt96, gross2006hudson}. Importantly, given a choice of a classically simulable frame, any input state which is a convex combination of frame elements admits an efficient classical simulation algorithm~\cite{veitch_2012-1}. In Ref.~\cite{pashayan15}, Pashayan \textit{et al.}\ showed that when the input state is a non-convex linear combination of frame elements, the only source of inefficiency in the runtime of quasiprobabilistic algorithms is given by the \emph{negativity} of the state --- a frame-dependent quantity which measures the degree of departure from convex mixtures of frame elements.

Quantum systems consisting of odd-dimensional subsystems (qudits)~\cite{campbell2012magic,anwar2014fast,campbell2014enhanced} admit an especially natural choice of frame. Here, the frame can be fixed to a set of phase-point operators for which the convex combinations of frame elements are the states with a positive discrete Wigner function~\cite{Leonhardt96,gross2006hudson}. All qudit stabilizer states have a positive Wigner function which leads to efficient, classical simulation of qudit stabilizer circuits~\cite{veitch_2012-1}. The negativity under this choice of frame was shown in Ref.~\cite{pashayan15} to correspond to the \emph{mana} $\mathcal{M}(\cdot)$ --- a magic monotone introduced in Ref.~\cite{veitch2014resource}. Notably, the mana has the convenient property that it is multiplicative~\footnote{Or equivalently additive ($\mathcal{M}(\rho \otimes \sigma)=\mathcal{M}(\rho)+ \mathcal{M}( \sigma)$) after taking the logarithm.} i.e. $\mathcal{M}(\rho \otimes \sigma)=\mathcal{M}(\rho) \mathcal{M}( \sigma)$. Computations of the mana in large dimensions are generally extremely difficult, but, due to multiplicativity, they are significantly simplified for products of states on smaller systems. Multiplicativity of operationally meaningful monotones allows for an easy evaluation of related quantities, such as a simulator's runtime or bounds on asymptotic rates of state conversion.

Curiously, for the fundamentally important case of qubits, phase-point operator frames do not possess many of the aforementioned desirable properties. A straightforward application of techniques that work for qudits yields a Wigner function that can be negative for some pure stabilizer states. Although alternative ways of defining a well-behaved Wigner function for qubits are possible, they always~\cite{mansfield2018quantum} suffer from drawbacks such as: the free operations and states being restricted to a subclass of the usual free operations~\cite{delfosse2015wigner,Raussendorf17}; or the monotones being super-multiplicative and the set of positively represented states not being closed under tensor product~\cite{raussendorf2019phase}.
Quasiprobability simulators based on qubit phase-point operator frames inherit these limitations, prompting alternative approaches.

In Ref.~\cite{Howard17robustness} Howard and Campbell presented a quasiprobability simulator for qubits using a frame composed of projectors onto pure stabilizer states. They showed that this gives rise to a classical simulation algorithm with a runtime linked to a magic monotone called the \emph{robustness of magic}. 
It is a qubit-based simulator that permits and utilizes the simulation of noisy inputs and operations, and possesses many desirable traits. However, presently, quasiprobability simulators are slower than stabilizer rank simulators; additionally, the robustness of magic is non-multiplicative and extremely difficult to compute, even in the asymptotic regime for products of relevant single-qubit states~\cite{bravyi2005universal, heinrich2018robustness}.

A seemingly independent line of work on classical simulation was introduced in Ref.~\cite{Bravyi16stabRank} with the stabilizer rank--based simulators~\cite{Bravyi16stabRank,Bravyi16stabRank,bravyi2018simulation,kocia2018stationary,huang2019approximate,kocia2020improved}. These simulators achieve a stronger notion of simulation~\cite{Pashayan2020fromestimation} by approximately sampling from the output distribution of the quantum circuit. However, they can only simulate pure states and operations and have not previously been generalized to noisy quantum circuits.
Stabilizer rank simulators represent the initial quantum state \jrs{vector} as a \jrs{superposition} of stabilizer states, and the only source of inefficiency in runtime is introduced by the exponential number of terms required to represent states in this way --- the minimal number of such terms being precisely the stabilizer rank. The original algorithm had a runtime quadratic in the stabilizer rank, but this was later improved by the development of fast norm estimation~\cite{bravyi2016improved} that provides a runtime linear in stabilizer rank}. This has resulted in a sizable runtime advantage for stabilizer rank simulators, and it is currently unclear if a similar improvement in quasiprobabilistic methods is possible. To circumvent the difficulty in computing the stabilizer rank as well as its non-multiplicative behavior,
Ref.~\cite{bravyi2016improved} also introduced the notion of approximate stabilizer rank, which was later related to a monotone called the \emph{stabilizer extent}~\cite{bravyi2018simulation}. While the extent is in general not multiplicative~\cite{heimendahl2020stabilizer}, it is multiplicative on any tensor product of one-, two-, and three-qubit states~\cite{bravyi2018simulation}.
Unfortunately, these concepts only apply to pure states, and no mixed-state simulation method based on the stabilizer rank has been devised thus far.

\renewcommand{\arraystretch}{1.5}

 \begin{table*}[t]
 \prx{%
  \begin{minipage}[t]{0.28\linewidth} 
\hspace{-3cm}
\includegraphics[width=4cm]{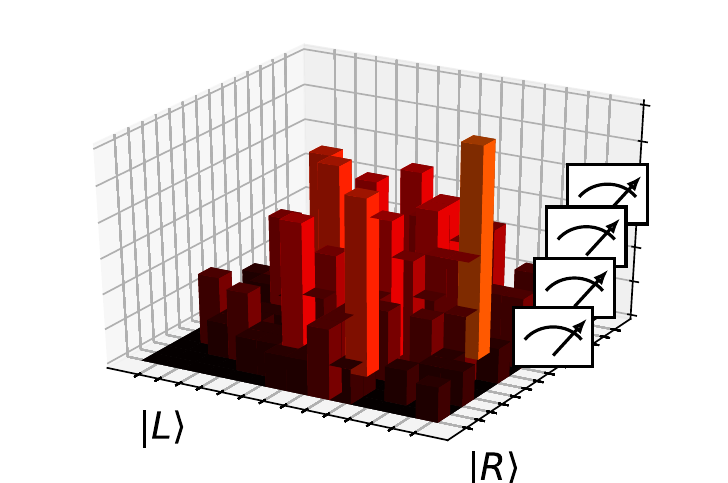} 

\hspace{-3cm}
\centering\noindent{\scriptsize Dyadic frame simulator}
\end{minipage}%
\begin{minipage}[t]{0.28\linewidth} 
\hspace{-1.5cm}
\includegraphics[width=4cm]{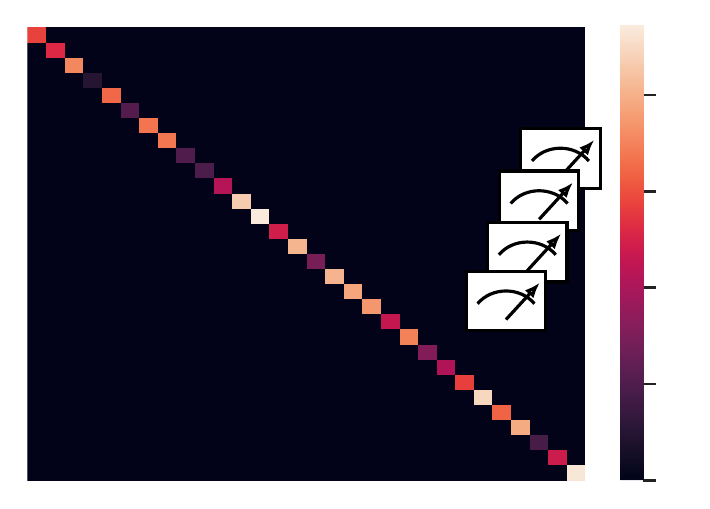} 

\hspace{-1.5cm}
\centering\noindent{\scriptsize Density-operator stabilizer rank simulator}
\end{minipage}%
\begin{minipage}[t]{0.2\linewidth} 
\includegraphics[width=4cm]{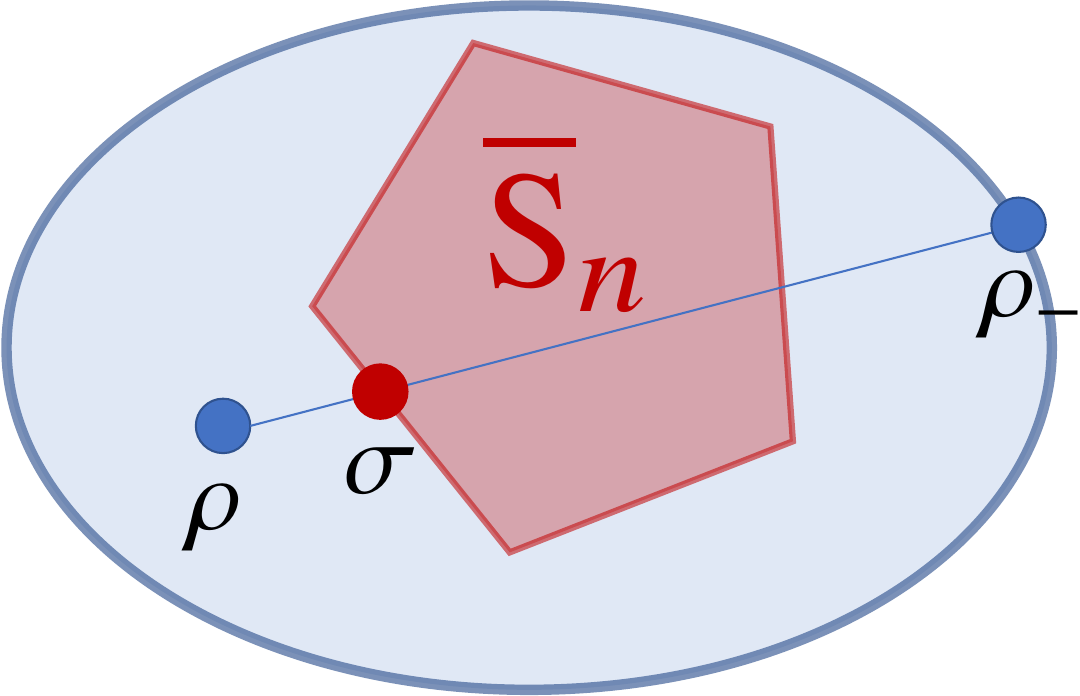} 

\centering\noindent{\scriptsize \quad Constrained path simulator}
\end{minipage}\\[.5\baselineskip]

\begin{tabular}{p{3.7cm}|p{3.85cm}|P{2cm}|P{2.55cm}|P{3.8cm}}
\thickhline
\textbf{Classical Simulator} & \centering \textbf{Output} & \textbf{Associated Monotone}  & \textbf{Method} & \textbf{Runtime/(${\rm poly}(n)\log (p_{\mathrm{fail}}^{-1})$)} \\
\thickhline
{\bf Dyadic frame simulator}

(Sec.~\ref{sec:simulate1}) 
&  Estimate of a single Born rule probability (or Pauli observable) within additive error $\epsilon$, for arbitrary $\epsilon>0$. & $\quad\ \Lambda$\newline  Dyadic negativity & Sample dyads from quasi-probability distribution. &  $\mathcal{O}(\Lambda(\rho)^2 \epsilon^{-2})$ \\
\hline
 {\bf Density-operator stabilizer rank simulator }
 
 (Sec.~\ref{sec:simulate2}) \newline    & A bit string sampled from a distribution $\delta$-close in the 1-norm to the quantum distribution, for arbitrary $\delta>0$.
 \newline &  $\quad\ \Xi$ \newline Mixed-state extent
  & Samples pure states from an ensemble then sparsifies their stabilizer rank representation. & $ \mathcal{O}(\Xi(\rho) \delta^{-3} ) \,\,\text{when}\,\,\delta \geq \delta_c, $  \newline
  $\mathcal{O}(\Xi(\rho) \delta^{-4} ) \,\,\text{when}\,\, \delta < \delta_c$\\
\hline 
 {\bf Constrained path simulator}

 (Sec.~\ref{sec:simulate3}) \newline   
  &Estimate of a single Born rule probability (or Pauli observable), with estimation error $\Delta = \mathcal O((\Lambda^+(\rho)-1))$. 
 &$\quad\ \Lambda^+$
  \newline
  Generalized robustness & Approximates a mixed magic state $\rho$ with a single mixed stabilizer state $\sigma$. & $     \mathcal O(1)$\\
\hline
\end{tabular}
\caption{A summary of the properties of our three classical simulation algorithms and their connections with magic monotones.
 Here, $p_{\mathrm{fail}}$ denotes the failure probability of the associated algorithm.}
 \label{table:simulators}
 }
\end{table*}

\prx{%
\subsection{Summary of results}

In this paper, we present three new classical simulation algorithms  (overviewed in Table~\ref{table:simulators})} which we call the \emph{dyadic frame}, the \emph{\jrs{density-operator} stabilizer rank} and the \emph{constrained path} simulators. The algorithms allow classical simulation of general noisy stabilizer circuits with mixed magic-state inputs, providing a significant extension of the capabilities of previous approaches, and revealing connections between stabilizer rank and quasiprobability-based simulation. The dyadic frame and constrained path simulators produce additive precision estimates of Born rule probabilities \jrs{and Pauli observables}, while the \jrs{density-operator} stabilizer rank simulator approximately samples from the quantum circuit's measurement outcome distribution. 
Our first two simulators trade off quantum computational resources for additional runtime of classical simulation. The constrained path simulator, on the other hand, is always efficient in runtime, instead reducing in accuracy as the simulated quantum circuits increase in magic.

Our dyadic frame simulator is a new state-of-the-art quasiprobability simulator for qubits. Instead of sampling from stabilizer states or phase-point operators, we sample from objects we call stabilizer dyads.  We show the corresponding resource monotone is smaller than the robustness of magic, leading to faster simulation runtimes. This can lead to a significantly improved exponent in the exponential scaling of the simulator's performance: for instance, for $n$ copies of a $T$ state the dyadic simulator has a runtime $\mathcal{O}(4^{0.228443 n})$, whereas for simulators based on the robustness of magic~\cite{Howard17robustness,Seddon19} the runtime is lower bounded by $\Omega(4^{0.271553 n})$.

Our stabilizer rank simulator is a new state-of-the-art simulator for sampling from qubit-based quantum circuits with three key technical contributions. First, our work generalizes the stabilizer rank simulator of Refs.~\cite{bravyi2016improved,bravyi2018simulation} from pure states to general mixed states. This allows our classical simulator to operate in and be directly comparable to more experimentally relevant regimes, where the input magic states are noisy. The natural generalization to mixed states produces a simulator with a probabilistic runtime. Second, we show this runtime can be made deterministic for an important subset of magic states. Third, we substantially improve runtimes by exploiting tighter proof techniques available in the density operator picture. Remarkably, this density operator technique is applicable and advantageous even when simulating pure states.

We show that each of our simulators --- the dyadic frame, the \jrs{density-operator} stabilizer rank and the constrained path simulator --- is associated with a particular magic monotone, which we call the \emph{dyadic negativity}, the \emph{mixed-state extent}, and the \emph{generalized robustness}, respectively.  Specifically, we show that the runtime (in the case of the dyadic frame and \jrs{density-operator} stabilizer rank simulators) or the precision (in the case of the constrained path simulator) of the algorithms directly relates to the corresponding magic monotone. This identifies the exponential growth of magic as the only source of inefficiency in these simulators.
Crucially, we completely characterize these monotones for single-qubit states and tensor products thereof, where we prove the unexpectedly strong result that these monotones are all equal and act multiplicatively. The multiplicativity of the monotones is the first result of this type for general qubit magic monotones, and the equality between all three monotones reveals a deep and precisely quantified connection between the runtimes of stabilizer rank and quasiprobability simulators. To the best of our knowledge, no previous work has established a quantitative connection between these, \textit{a priori} very different, classes of simulators. All of the monotones reduce to the stabilizer extent for pure states, and so they can all be considered as generalizations of the extent to mixed states.
In addition to serving as an important contribution to magic theory and tightly characterizing the resource consumption of our simulators, we use the monotones to introduce computable bounds on the asymptotic and non-asymptotic rates for magic state distillation. For some example distillation tasks, we compare our bounds to other recent results~\cite{fang_2019-1} and find they are much tighter across a wide parameter regime.

\jrs{Classical simulation of quantum systems has been studied within various contexts other than magic theory  ~\cite{bartlett_2002,jozsa2008matchgates,veitch_2012-1,mari_2012-1,pashayan15}, but to our knowledge none of these approaches have been adapted to the umbrella of quantum resource theories~\cite{chitambar_2019}.} We provide a comprehensive recipe to apply our methodology to general quantum resources. We thus establish connections between a family of resource monotones and simulation tasks, shedding light on classical simulation algorithms in broader settings. 
For instance, in the resource theory of quantum coherence~\cite{baumgratz_2014,streltsov_2017}, the $\ell_1$-norm of coherence is a fundamental quantifier of this resource but lacks an operational meaning. Our results fill this gap by showing that the $\ell_1$-norm of coherence quantifies the runtime of classical simulation within this theory.

This article is structured as follows. In Sec.~\ref{sec:magicmonotones} we introduce the setting of magic theory and our family of monotones. Sec.~\ref{sec:magicmonotones} also discusses how the monotones connect with our simulation algorithms, providing a statement of our main theorems. In Sec.~\ref{sec:single} we present a complete picture of how our monotones compare for single-qubit states by showing that they are all equal. The equality between monotones is then extended in Sec~\ref{sec:Multi} to tensor product states, where we show that the monotones are strongly multiplicative. In Sec.~\ref{SecCompareROM} we compare our new monotones with the robustness of magic and show that they can be exponentially smaller in magnitude.  In Sec.~\ref{sec:distill}, we discuss how the monotones can be used to bound the performance of magic state distillation protocols. Sec.~\ref{sec:simulate} contains a complete discussion of our simulation algorithms --- we focus on providing an intuitive picture through illustrative examples and sketches of the main proof ideas, with the full technical details deferred to the Appendix. We conclude in Sec.~\ref{sec:otherTheories} with a discussion of our underlying assumptions and the extension of our results to resource theories beyond magic.

\section{Preliminaries}\label{sec:magicmonotones}
\prx{%
\subsection{The stabilizer formalism}

Here, we briefly review the stabilizer formalism.  The single-qubit Pauli group $\mathcal{P}_1$ contains the identity matrix $\id$, the Pauli spin matrices $X$, $Y$, and $Z$, as well as their products with $\pm i \id$.  We say a pure, single-qubit state is a stabilizer state if there exists a Pauli operator $P \in \mathcal{P}_1$, such that $P \ket{\psi} = \ket{\psi}$ \jrs{and $P\neq \id$}. There are six such states:
\begin{align}
    Z \ket{0} & = \ket{0} , \\ \nonumber
    (-Z) \ket{1} & = \ket{1} , \\ \nonumber
    \pm X \ket{\pm} & = \ket{\hakop{\pm}} := (\ket{0}\pm \ket{1}) / \sqrt{2} ,\\ \nonumber
   \pm Y \ket{\pm_i} & = \ket{\hakop{\pm}_i} := (\ket{0}\pm i\ket{1}) / \sqrt{2} .
\end{align}
The $n$-qubit Pauli group $\mathcal{P}_n$ is the group generated by tensor-products of $n$ Pauli operators.  We say a pure, $n$-qubit state is a stabilizer state if there exists an Abelian subgroup of the Pauli group $\mathcal{S} \subset \mathcal{P}_n$ containing $2^n$ elements such that $S \ket{\psi} = \ket{\psi}$ for all $S \in \mathcal{S}$.  The group $\mathcal{S}$ is called the stabilizer group of the state $\ket{\psi}$, and that $\mathcal{S}$ can be described using $\mathcal{O}(n^2)$ bits underpins the efficient classical simulation results of the Gottesman-Knill theorem.  We use $\mathrm{S}_n$ to denote the set of pure $n$-qubit stabilizer states.  \br{The set of mixed stabilizer states $\bar{\mathrm{S}}_{n}$ is then formed by all states which can be decomposed as a mixture of pure stabilizers, that is, $\bar{\mathrm{S}}_{n} = \mathrm{conv} \{ \kb{\phi}{\phi} : \ket{\phi} \in \mathrm{S}_n \}$}.

An $n$-qubit unitary $C$ is Clifford if for every Pauli $P \in \mathcal{P}_n$, it follows that $C P C^\dagger \in \mathcal{P}_n$.  
We see that stabilizer states are mapped to stabilizer states under Clifford unitaries, and furthermore this update can be tracked efficiently.  In addition, measurements of Pauli operators on stabilizer states can also be efficiently simulated by appropriately updating the stabilizer group.

  We will refer to the stabilizer operations as any sequence of the following: preparation of stabilizer states, Clifford unitaries, Pauli measurements and adaptive feedforward depending on previous measurement outcomes or random coin tosses.  From the perspective of complexity theory, a small caveat is required that adaptive feedforward decisions are computed using only a small (constant size) classical computer.  

A quantum channel $\mathcal{E}$ is said to be stabilizer preserving if it maps every mixed stabilizer state $\rho \in \bar{\mathrm{S}}_{n}$ to another mixed stabilizer state, so $\mathcal{E}(\rho) \in \bar{\mathrm{S}}_{n}$. \br{Although meaningful when acting on the whole system in consideration, such maps can exhibit undesirable properties when acting on a part of a larger system~\cite{Seddon19}. We will thus consider a relevant class of free operations defined as follows.}
\begin{defin} \label{def:CPTPSP}
We define the set of free operations $\mathcal{O}_n$ as the set of channels $\mathcal{E}$ that are: (\textit{i}) completely positive; (\textit{ii}) trace preserving, so that $\mathrm{Tr}[\mathcal{E}(\rho)] = \mathrm{Tr}[\rho]$; (\textit{iii})  completely stabilizer preserving, in the sense that 
\begin{equation}
    [\chan \otimes \id] (\rho) \in \bar{\mathrm{S}}_{2n}  \quad \forall \rho \in \bar{\mathrm{S}}_{2n}.
\end{equation}
This set can equivalently be defined via the Choi-Jamio{\l}kowski isomorphism as was shown in~\cite[Thm. 3.1]{Seddon19}.
\end{defin}
While it is clear that the stabilizer operations are contained in $\mathcal{O}_n$, it is not known whether all elements of $\mathcal{O}_n$ can be realized by the standard stabilizer operations without post-selection.  The Gottesman-Knill theorem has long been known to show that stabilizer operations can be efficiently classically simulated, but only recently was it shown that the more general class $\mathcal{O}_n$ also admits efficient simulation algorithms~\cite{Seddon19}. 
Furthermore, it is known that certain stabilizer-preserving but non-trace-preserving maps, such as post-selection on the outcome of a Pauli measurement, can also be efficiently simulated. For technical reasons we do not consider these as elements of the convex set of free operations $\mathcal{O}_n$ in our resource theory, but we exploit their simulability in Section \ref{sec:simulate}.

While the stabilizer operations \br{(or the free operations $\mathcal{O}_n$)} are not universal for quantum computation, they can be promoted to universality given an unlimited supply of a suitable non-stabilizer operation.  For instance, adding the $T$ gate (also called the $\pi/8$ phase gate)
\begin{equation}
    T = \left( \begin{array}{cc}
    e^{i \pi /8} & 0 \\
    0 & e^{- i \pi /8}
    \end{array} \right) ,
\end{equation}
promotes the stabilizer operations to full quantum universality~\cite{boykin1999universal}.  Alternatively, one can add a supply of non-stabilizer states such as the so-called magic states:
\begin{align}
    \kb{H}{H} & =(1/2)\left( \id + (X+Z)/\sqrt{2} \right) \\ 
     \kb{T}{T} & =(1/2)\left( \id + (X+Y)/\sqrt{2} \right) \\ 
     \kb{F}{F} & =(1/2)\left( \id + (X+Y+Z)/\sqrt{3} \right) ,
\end{align}
which we use throughout.  Given a single copy of the Hadamard eigenstate $\ket{H}$ or the Clifford equivalent $T$-state $\ket{T}$, we can perform a deterministic $T$ gate using state injection \cite{bravyi2005universal}.  Therefore, full university can be achieved given stabilizer operations and a supply of magic states.  This is an important paradigm as it is the route most commonly used in the design of fault-tolerant quantum computers. 

However, stabilizer operations \jrs{with} access to a restricted number of magic states \jrs{do} not lead to universal quantum computation.  Rather, the computational power depends on the type and quantity of magic states provided. It is precisely this question of computational power that we quantify by studying the complexity of simulating computations with a limited resource of magic states. 
}

\subsection{Magic monotone definitions}
We now introduce several magic monotones of interest, borrowing some results from the general resource theory literature. 
Although in our discussion we specialize to the theory of magic states, the basic considerations below can also be applied to more general resources in which the set of free states is defined by convex combinations of free pure states, which includes important examples such as coherence and entanglement. We will elaborate on this in Sec.~\ref{sec:otherTheories}. 

For pure states, we define the following.
\begin{defin}[\cite{bravyi2018simulation}]
The \textbf{pure-state extent} $\xi$ is the quantity	
\begin{equation}
	\xi( \Psi ) := \mathrm{min} \{  \| c \|_1^2 :  \ket{\Psi} = \sum_j c_{j} \ket{\phi_j} ;    \ket{\phi_j} \in \mathrm{S}_n \} . \label{eq:stabextentDef}
\end{equation} 
\end{defin}
In magic theory, $\xi$ is the stabilizer extent~\cite{bravyi2018simulation}.  A related quantity appears in other resource theories such as entanglement, where it admits an analytical formula as the squared sum of the Schmidt coefficients of a state~\cite{rudolph_2001}, or in coherence theory, where it is the square of the $\ell_1$-norm of coherence~\cite{baumgratz_2014}.  It is well known~\cite{regula2017convex,bravyi2018simulation} that this can be recast as a dual optimization problem
\begin{equation}
	\xi( \Psi ) := \mathrm{max} \left\{  | \bk{\omega}{\Psi}|^2 :  |\bk{\omega}{\phi}| \leq 1 \; \forall \ket{\phi} \in \mathrm{S}_n \right\} .\label{opt:extent-dual}
\end{equation}
Here, we define an $\omega$-witness to be any feasible solution to the optimization problem in Eq.~\eqref{opt:extent-dual}.

We now consider four monotones, of which three can be regarded as mixed-state extensions of $\xi$.  First, one can extend the extent to mixed states using a convex roof extension~\cite{bennett_1996}:
\begin{defin}
\label{RoofExtension}
The \textbf{mixed-state extent} $\Xi$ is the quantity
\begin{equation} \nonumber
	\Xi( \rho ) := \mathrm{min} \left\{   \sum_j p_j 	\xi( \Psi_j ) : \rho=\sum_j p_j \kb{\Psi_j }{ \Psi_j }   \right\} , \label{eq:XiDef}
\end{equation}
where every $\ket{\Psi_j}$ is a pure state and $p_j$ are non-negative coefficients such that $\sum_j p_j = 1$. Furthermore, if the minimum can be achieved with a decomposition where all $\xi( \Psi_j )$ are equal, then we say the state admits an \textbf{equimagical} decomposition.
\end{defin}
We also consider quasiprobability distributions over free states as follows.
\begin{defin}[\cite{Howard17robustness}] \label{DefRoM}
The \textbf{robustness} $\mathcal{R}$ is the quantity
	\begin{equation} \nonumber
		\mathcal{R}( \rho ) := \mathrm{min} \left\{ \|q\|_1 : \rho = \sum_j q_j \kb{\phi_j}{\phi_j} ;\;   \ket{\phi_j}  \in \mathrm{S}_n \right\},
	\end{equation}
	where $q_j$ are real coefficients.
\end{defin}
 In magic theory, $\mathcal{R}$ is called the robustness of magic~\cite{Howard17robustness,heinrich2018robustness}, inspired by the (standard) robustness of entanglement~\cite{vidal_1999}. \br{This quantity is precisely the negativity with respect to the frame defined by the set of pure-state stabilizer projectors. In particular,} the robustness uses decompositions where the rank-one ket--bra terms are Hermitian.  Relaxing this, we can define
\begin{defin}\label{def:dyadic}
The \textbf{dyadic negativity} $\Lambda$ is the quantity	
	\begin{equation} \nonumber
\Lambda(\rho) := \mathrm{min} \left\{ \|\alpha\|_1 : \rho \!=\! \sum_j \alpha_j \kb{L_j}{R_j} ;\;   \ket{L_j} , \ket{R_j} \in \mathrm{S}_n \right\},
	\end{equation}
	where the coefficients $\alpha_j$ are complex numbers. 
\end{defin}
The name reflects the fact that each $\kb{L_j}{R_j}$ comprises of a pair of vectors, and so is a dyad. Within the resource theory of entanglement, a related quantity called the projective tensor norm was considered~\cite{rudolph_2001,rudolph_2005}, and in the resource theory of coherence the dyadic negativity corresponds to the $\ell_1$-norm of coherence~\cite{baumgratz_2014}. Viewing this quantity as the primal solution of a convex optimization problem, it is useful to state the equivalent dual formulation~\cite{regula2017convex} in terms of witness operators.  We define the set of $W$-witnesses, denoted $\mathcal{W}$, to be the Hermitian operators such that 
\begin{equation}
	\label{WitnessClass}
	 \mathcal{W} := \{  W :   | \bra{L} W \ket{R}| \leq 1 \, \forall\, \ket{L} , \ket{R} \in \mathrm{S}_n \},
\end{equation}	
which by strong duality leads to
\begin{equation}
\label{DualFormulation2}
	\Lambda(\rho) = \mathrm{max} \{ \mathrm{Tr}[W \rho ] :   W \in \mathcal{W} \} .
\end{equation}
 This brings us to our last monotone of interest.
\begin{defin}
The \textbf{generalized robustness} $\Lambda^+$ is the quantity	
\begin{equation}\label{eq:Lambda_plus}
	\Lambda^+(\rho) = \mathrm{max} \{ \mathrm{Tr}[W \rho ] :  W \in \mathcal{W} ; W \geq 0 \} ,
\end{equation}
where $\mathcal{W}$ is the set of $W$-witnesses. 
\end{defin}
\noindent 
A corresponding quantity to $\Lambda^+$ was first defined in entanglement theory~\cite{vidal_1999,steiner_2003} and appears in many resource theories.

Notice that this is similar to the dual formulation given in Eq.~\eqref{DualFormulation2} except we further restrict to witnesses that are also positive semidefinite operators.  
We define a $W^+$-witness to be any feasible solution to the optimization problem in Eq.~\eqref{eq:Lambda_plus}.
Since $W^+$-witnesses are positive semidefinite, the condition $| \bra{L} W \ket{R}|\leq 1 $ simplifies to $ \bra{\psi} W \ket{\psi} \leq 1$ for all $\ket{\psi} \in \mathrm{S}_n$.

Furthermore, the primal form of this monotone is
\begin{align}
	\Lambda^+(\rho) &= \mathrm{min} \{ \lambda :  \rho \leq \lambda \sigma,\; \sigma \in \bar{\mathrm{S}}_{n} \}
	\label{eq:LambdaPlusDefn2}\\
	&= \mathrm{min} \left\{ \lambda \geq 1 :  \frac{\rho + (\lambda-1) \rho' }{\lambda} \in \bar{\mathrm{S}}_{n} \right\}
	\label{eq:LambdaPlusDefn3},
\end{align}
where the optimization in the second line is over all density matrices $\rho'$.
This form motivates the name of generalized robustness: rearranging Def.~\ref{DefRoM}, the robustness $\mathcal{R}$ can be similarly expressed as
\begin{equation}
\frac{\mathcal{R}(\rho)+1}{2} = \mathrm{min} \left\{ \lambda \geq 1 :  \frac{\rho + (\lambda-1) \sigma}{\lambda} \in \bar{\mathrm{S}}_{n},\; \sigma \in \bar{\mathrm{S}}_{n} \right\}
\label{eq:RobDef2}
\end{equation} 
where now the states in the optimization are restricted to free states in $\bar{\mathrm{S}}_{n}$.

We stress that both $\Lambda$ and $\Lambda^+$ are computable, in the sense that their evaluation corresponds to convex optimization problems --- a second-order cone program for $\Lambda$, and a semidefinite program for $\Lambda^+$ --- which can be evaluated using numerical solvers~\cite{boyd2004convex}. In practice, we were able to compute $\Lambda^+$ up to $n=4$ and $\Lambda$ up to $n=3$, but one can certainly hope to make further progress in computing the quantities for states obeying some symmetry, just as in the case of $\mathcal{R}$~\cite{heinrich2018robustness}. The evaluation of convex-roof--based quantities such as $\Xi$ is notoriously hard in general~\cite{uhlmann_2010}, although one could again use symmetry to facilitate it in special cases~\cite{vollbrecht_2001}. Our results in Sec.~\ref{sec:single}-\ref{sec:Multi} further simplify the computation of all of the monotones for the practically important case of tensor products of single-qubit states.

The monotones have been considered from the perspective of general resource theories~\cite{regula2017convex}, and in particular they have been shown to satisfy a number of useful properties:
\begin{enumerate}
	\item \textit{faithfulness}: $\mathcal{M}(\rho)=1$ if and only if $\rho \in \bar{\mathrm{S}}_{n}$;
	\item \textit{monotonicity}: $\mathcal{M}(\rho) \geq \mathcal{M}(O(\rho))$ for any free operation $O \in \mathcal{O}_n$;
	\item \textit{strong monotonicity} (monotonicity on average under selective free measurements):
	\begin{equation}\mathcal{M}(\rho) \geq \sum_i p_i \mathcal{M} \left(\frac{K_i \rho K^\dagger_i}{p_i}\right) ,\end{equation}
	where $\{K_i\}_i$ are the Kraus operators of a quantum channel such that each $K_i$ is stabilizer preserving, i.e. $K_i \ket{\phi} \propto \ket{\phi'} \in \mathrm{S}_n \; \forall\, \ket{\phi} \in \mathrm{S}_n$, and $p_i = \Tr(K_i \rho K^\dagger_i)$;
	\item \textit{convexity}: $\mathcal{M}( \sum_j p_j \rho_j) \leq \sum_j p_j \mathcal{M}(  \rho_j) $;
	\item \textit{submultiplicativity}: $\mathcal{M}( \otimes_j \rho_j) \leq \prod_j \mathcal{M}(  \rho_j)$.
\end{enumerate}
We remark that, although $\mathcal{R}$ and $\Lambda^+$ are monotones in any convex resource theory, the fact that $\Lambda$ and $\Xi$ obey monotonicity under all completely stabilizer-preserving operations $\mathcal{O}_n$ is a consequence of two properties: the strong monotonicity of the measures~\cite{regula2017convex} coupled with the fact that any operation $O \in \mathcal{O}_n$ can be expressed in terms of Kraus operators $\{K_i\}_i$ which preserve the set of stabilizer states~\cite{Seddon19}.
If we instead work with logarithmic monotones, $\mathcal{M}_{\mathrm{log}}(\rho) = \mathrm{log} [ \mathcal{M}(\rho) ]$ then multiplicativity becomes additivity, faithfulness instead has a $\mathcal{M}_{\mathrm{log}}(\rho)=0$ condition, and due to concavity of the logarithm $\mathcal{M}_{\mathrm{log}}$ is no longer a convex function but still obeys strong monotonicity~\cite{plenio_2005}.  Here we find it convenient to work without the logarithm in most cases.

Next, we present some general relations between these monotones that are reminiscent of known results in general resource theories~\cite{regula2017convex} . 
\begin{lem}[\cite{regula2017convex}]
	\label{RegulaLemma1}
	For any pure state 
	\begin{equation}
\Lambda^+(\kb{\Psi}{\Psi}) = \Lambda(\kb{\Psi}{\Psi}) = \Xi[ \kb{\Psi}{\Psi} ]	= \xi ( \Psi ).
	\end{equation}
\end{lem}
\noindent Therefore, our monotones can be interpreted as mixed-state extensions of $\xi$. We also observe the following.
\begin{theorem}
	\label{Sandwich1}
	For any state $\rho$ we have
	\begin{equation}
		\Lambda^+( \rho )	\leq \Lambda( \rho ) \leq \Xi ( \rho ).
	\end{equation}
\end{theorem}
For completeness, we provide alternative proofs of these results in  Appendix~\ref{RegulaProof}.  Since $\Lambda^+$ is often easier to evaluate than $\Lambda$ and $\Lambda^+ \le \Lambda$, in practical settings, one can approximate $\Lambda$ by evaluating $\Lambda^+$.

\subsection{Connecting monotones with simulation}

To further motivate our investigation of the magic monotones that follows in the subsequent sections, we summarize our main results and show how the properties of the monotones will be vital to the understanding of several classes of classical simulation algorithms. Our first simulation algorithm is a quasiprobability-based approach, which introduces several novel modifications to standard Monte Carlo techniques, notably the use of dyadic frames.

\setlength{\tabcolsep}{4pt}

\begin{namedthm}{Theorem~\ref{thm:simulationJS} (informal)}
Consider an $n$-qubit initial state with known decomposition into dyads $\rho = \sum_j \alpha_j \ket{L_j}\bra{R_j}$ where $\|\alpha\|_1=\Lambda(\rho)$. Let $\chan$ be a sequence of $T$ stabilizer-preserving operations, each acting on a few qubits. Then, given a stabilizer projector $\Pi$, we can estimate the Born rule probability $\mu = \mathrm{Tr}(  \Pi \chan [\rho ] )$ with probability $1-p_{\mathrm{fail}}$ and additive error $\epsilon$ within a runtime
	\begin{equation}
  \frac{\Lambda(\rho)^2}{\epsilon^2}  \log(p_{\mathrm{fail}}^{-1})
T{\rm poly}(n).
	\end{equation}	
\end{namedthm}

Hence, the dyadic negativity $\Lambda$ exactly characterizes our algorithm's runtime. 
To understand how the performance scales when more copies of the input state $\rho$ are provided, it is then necessary to understand the multiplicativity of $\Lambda$. We solve this question completely with the following.

  \begin{namedthm}{Theorem~\ref{Multiplicative2}}
	Let $\sigma_j$ be single-qubit states. Then
	\begin{align}
		\Lambda(\otimes_j \sigma_j )  =  \Xi(\otimes_j \sigma_j ) = \Lambda^+(\otimes_j \sigma_j ) =   \prod_j \Lambda^+( \sigma_j).
	\end{align}	
\end{namedthm}

This not only reveals a connection between three monotones introduced previously --- allowing, for instance, for the evaluation of the generally hard-to-compute quantifier $\Xi$ --- but also shows them to be strictly multiplicative for qubit states.  Consequently, when we plot these quantities on a log scale, we get a straight line, as shown with the example in Fig.~\ref{fig:Scaling}. Although a common occurrence in the structurally simpler theory of qudit magic states~\cite{veitch2014resource,wang2018efficiently}, multiplicativity has not been shown before for any mixed-state monotone in qubit magic theory.

Thm. \ref{Multiplicative2} lets us avoid the main problem which hinders an understanding of the performance of previous quasiprobability simulation algorithms such as the Howard-Campbell simulator based on the robustness $\mathcal{R}$, namely the inability to efficiently compute $\mathcal{R}(\rho^{\otimes n})$ for large $n$~\cite{Howard17robustness,heinrich2018robustness}. In addition, we can use the multiplicativity result to show an exponential separation between our monotones and the robustness of magic.

\begin{namedthm}{Theorem~\ref{ExpGapProp}}
Given any single-qubit non-stabilizer state $\rho$, there exists positive real constants $\alpha$ and $\beta$ where  $\alpha > \beta$ and so that
\begin{align} 
    2^{\alpha n} & \leq \mathcal{R}(\rho^{\otimes n}) \\ 
    2^{\beta n} & = \Lambda(\rho^{\otimes n}) = \Lambda^+(\rho^{\otimes n})= \Xi(\rho^{\otimes n}).
\end{align}
\end{namedthm}
This establishes the simulation algorithm of Thm.~\ref{thm:simulationJS} as polynomially faster than previous quasiprobability simulators, as illustrated by the example in Fig.~\ref{fig:Scaling}.

Our second simulation algorithm is based on the stabilizer rank, which allows it to be used for both Born rule probability estimation and for approximately sampling from the output distribution of a quantum circuit. Importantly, existing stabilizer rank simulation algorithms only applied to pure states~\cite{Bravyi16stabRank,bravyi2018simulation}. We extend this to mixed states through the monotone $\Xi$ as follows.

\begin{namedthm}{Theorem~\ref{thm:bit-string-sampling} (informal)}
    Let $\rho$ be a state with known mixed-state extent decomposition. Then there is a classical algorithm that approximately samples from the probability distribution associated with a sequence of Pauli measurements on $\rho$.  Our samples come from a distribution that is $\delta$-close in $\ell_1$-norm to the actual distribution, and each sample has an expected runtime
\begin{align}\mathbb{E}(T) = \mathcal{O}(\Xi(\rho) / \delta^3 ) \end{align}
as long as $\delta$ is not too small. Furthermore, if $\rho$ is a product of single-qubit states, there is no variance in the runtime.
\end{namedthm}

There are two notable technical advances here: one is a factor $1/\delta$ improvement in runtime over previous simulators of this type~\cite{bravyi2018simulation}, even when applied to pure states; the other improvement is the rather surprising result that sampling can often be performed without any variance in the runtime.

\br{The last of our simulation algorithms is the constrained path simulator, which enjoys an efficient runtime, but instead sacrifices the accuracy of the simulation depending on how resourceful the input state is. The precision has an inverse polynomial dependence on the generalized robustness $\Lambda^+$, again directly connecting a magic monotone with classical simulation.}

We thus see that the tightness of our simulators' runtimes and our ability to sharply characterize them is inherited from the properties and characterization of the monotones introduced earlier. \br{We overview the connections between the monotones and our simulation algorithms in Table~\ref{table:simulators}.} The detailed derivation of the Theorems, as well as additional results --- including connecting the monotones with magic distillation rates --- all follow in the remainder of the paper. 


\begin{figure}[t!]
\prx{%
    \centering
    \includegraphics[width=0.8\linewidth]{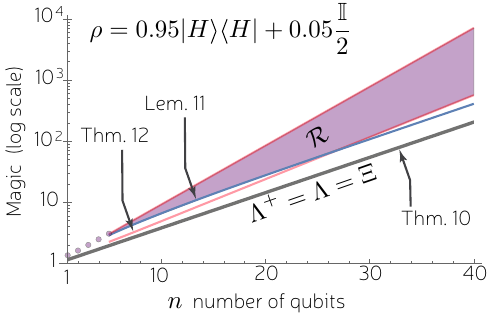}
    \caption{The scaling of magic monotones for many copies of a noisy single-qubit magic state $\rho$, highlighting several of our results.  Our new monotones, $\Lambda^+(\rho^{\otimes n})$, $\Lambda(\rho^{\otimes n})$ and $\Xi(\rho^{\otimes n} )$ are proved to be equal with multiplicative scaling leading to a straight line (grey) on this logarithmic scale (due to Thm.~\ref{Multiplicative2}).  We contrast this with a previously studied monotone, the robustness of magic $\mathcal{R}$, for which we can numerically compute the value up to $n=5$ (shown as purple data points).  For $n>5$, the shaded purple region shows the possible values of $\mathcal{R}$ as enforced by upper bounds (due to sub-multiplicativity) and two lower bounds (Lemma~\ref{RegulaLemma2} and Thm.~\ref{ExpGapProp}).  The robustness of magic has a wide range of possible values, but even the lower bound grows exponentially faster than the value of our new monotones, entailing that classical simulation algorithms based on the new monotones offer an improvement in the exponential scaling of their runtime.}
    \label{fig:Scaling}
}
\end{figure}

\section{Single-qubit magic states}
\label{sec:single}

In this section, we present a complete description of our magic monotones for single-qubit states. Recalling from Thm.~\ref{Sandwich1} that the monotones in general obey the relation $\Xi(\rho) \geq \Lambda(\rho) \geq \Lambda^+(\rho)$, the key question is then whether the inequalities can be tight, thus unifying the different approaches to the quantification of magic. We answer this in the affirmative.
 \begin{theorem}
 	\label{Sandwich2}
For any single-qubit state  $\rho$,  we have
     	\begin{equation}
    	\Lambda^+( \rho ) = \Lambda( \rho ) = \Xi ( \rho ),
    \end{equation}
and furthermore $\rho$ admits an equimagical decomposition (recall Def.~\ref{RoofExtension}).       
 \end{theorem}
We will see in the following section that this equivalence persists for tensor products of single-qubit states. However, equality does not extend to general $n$-qubit states for $n\geq 2$, as numerically we find that $\Lambda^+(\rho) < \Lambda(\rho)$ for most random two-qubit density matrices $\rho$. The proof of Thm.~\ref{Sandwich2} rests on a trio of lemmata.  First, we have:
\begin{lem}[\textbf{The monotone equality lemma}]
    \label{EqualityLem}
    For any $\omega$-witness $\ket{\omega}$, we define the set $B_\omega$ to be the convex hull of all pure states $\Psi$ for which $|\langle \omega \vert \Psi \rangle |^2=\xi(\Psi)$.  It follows that for all $\rho \in B_{\omega}$ we have 
    	\begin{equation}
    	\Lambda^+( \rho ) = \Lambda( \rho ) = \Xi ( \rho ) = \langle \omega \vert \rho \vert \omega \rangle.
    \end{equation}
\end{lem}
\begin{proof}[Proof of Lem.~\ref{EqualityLem}]  If $\rho \in B_{\omega}$, we can find a convex decomposition
\begin{equation}
    \rho = \sum_j p_j \vert \Psi_j \rangle \langle \Psi_j \vert ,
\end{equation}
where $|\langle \omega \vert \Psi_j \rangle |^2=\xi(\Psi_j)$ for all $j$.  We can use this decomposition to obtain an upper bound on the mixed-state extent as follows
\begin{align} \label{XiOmega}
    \Xi[\rho] & \leq \sum_j p_j \xi(\Psi_j) \\ \nonumber 
    & = \sum_j p_j \langle \omega \vert \Psi_j \rangle \langle \Psi_j \vert \omega \rangle \\ \nonumber
    & = \langle \omega \vert \rho \vert \omega \rangle .
\end{align}
On the other hand, $W=\vert \omega \rangle \langle \omega \vert \in \mathcal{W}^+$ and so can be used to lower bound the generalized robustness to show
\begin{equation} \label{LambdaPlusOmega}
    \langle \omega \vert \rho \vert \omega \rangle \leq \Lambda^{+}(\rho).
\end{equation}
Combining Eq.~\eqref{XiOmega} and Eq.~\eqref{LambdaPlusOmega} with Thm.~\ref{Sandwich1}, we have 
\begin{equation}
     \langle \omega \vert \rho \vert \omega \rangle \leq \Lambda^+( \rho )	\leq \Lambda( \rho ) \leq \Xi ( \rho ) \leq \langle \omega \vert \rho \vert \omega \rangle.
\end{equation}
Therefore, these inequalities all collapse to equalities.\end{proof}

Making use of Lem.~\ref{EqualityLem} requires us to first understand the structure of optimal $\omega$-witnesses, which we shall discuss soon.  However, first it is useful to define some different subsets of the Bloch sphere.
\begin{defin} \label{OctantDefs}
The positive octant is the set
\begin{equation}
    P := \{ \rho :  \langle X \rangle, \langle Y \rangle, \langle Z \rangle \geq 0 \}.
\end{equation}
We further subdivide the positive octant as follows:
\begin{align}
     P_X & := \{ \rho : \rho \in P, \langle X \rangle \leq \langle Y \rangle,  \langle X \rangle \leq \langle Z \rangle \}  , \\ \nonumber
     P_Y & := \{ \rho : \rho \in P, \langle Y \rangle \leq \langle X \rangle,  \langle Y \rangle \leq \langle Z \rangle   \}  , \\ \nonumber
    P_Z & := \{ \rho : \rho \in P, \langle Z \rangle \leq \langle X \rangle,  \langle Z \rangle \leq \langle Y \rangle \}  . 
\end{align}
where we use the shorthand $\langle M \rangle := \mathrm{Tr}[\rho M]$.  See Fig.~\ref{CanonialRegion} for an illustration of $P_Y$.
\end{defin}
 The sets $P_X$, $P_Y$ and $P_Z$ further divide the positive octant into thirds and it is easy to verify that $P = P_X \cup P_Y \cup P_Z$.  These sets are not quite disjoint because of the following proposition.
\begin{proposition} \label{intersections}
From Def.~\ref{OctantDefs}, we have the following
\begin{align} \nonumber
   P_Z \cap P_Y & = \{ \rho : \rho \in P, \langle Z \rangle = \langle Y \rangle \leq  \langle X \rangle  \}   \\  \nonumber
   P_X \cap P_Z & = \{ \rho : \rho \in P, \langle X \rangle = \langle Z \rangle \leq  \langle Y \rangle  \}   \\ \nonumber
  P_X \cap P_Y & = \{ \rho : \rho \in P, \langle X \rangle = \langle Y \rangle \leq  \langle Z \rangle  \}     .
\end{align}
\end{proposition}
This is straightforward to prove.  For example, in $P_Z$ the smallest expectation value is for $Z$ and for $P_Y$ the smallest expectation value is for $Y$.  Therefore, in the intersection these two expectation values must be equal. We note that any state is Clifford equivalent to a state in the positive octant $P$. Furthermore, the Clifford
\begin{equation}
F:=\frac{1}{\sqrt{2}} \left(
\begin{array}{cc}
1 & -i \\
1 & i \\
\end{array}
\right) ,
\end{equation}
satisfies $FXF^\dagger = Y$, $FYF^\dagger  = Z$ and  $FZF^\dagger  = X $. Therefore, the sets $P_X$, $P_Y$ and $P_Z$ are Clifford equivalent and therefore every state is Clifford equivalent to some $\rho \in P_Y$.

Now we are ready to characterize optimal $\omega$-witnesses.
\begin{lem}
\label{WitnessLemma}
    Let $\ket{\Psi}$ be any pure, single-qubit non-stabilizer state in the set $P_Y$. Then the $\omega$-witness $\ket{\omega}$ that achieves $|\bk{\Psi}{\omega}|^2=\xi(\Psi)$ has an operator representation of the form
    \begin{equation}
    \label{OptWitness}
        \kb{\omega}{\omega} = \frac{\id + q
        H + \sqrt{1-q^2}Y }{1+q/\sqrt{2}  },
    \end{equation}
    where $ \sqrt{2/3} \leq q \leq 1$ and $H=(X+Z)/\sqrt{2}$. Furthermore, if $\ket{\Psi}$ is in the set $P_Y \cap P_X$ or $P_Y \cap P_Z$ then $q=\sqrt{2/3}$ and the $\omega$-witness takes the form
    \begin{equation}
    \label{FaceWitness}
        \kb{\omega}{\omega} = \frac{\id + (X+Y+Z)/\sqrt{3} }{1+1/\sqrt{3}  }.
    \end{equation}
\end{lem}
The actual value of the variable $q$ is easy to numerically compute, but is analytically complicated and not instructive to present.  Rather, in  Fig.~\ref{CanonialRegion}, we illustrate the region $P_Y$ and highlight where $q=\sqrt{2/3}$ and $q>\sqrt{2/3}$. 

\begin{figure}
	\includegraphics[width=200pt]{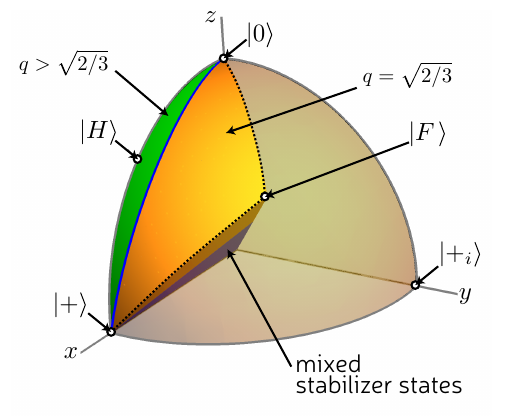}
	\caption{The region $P_Y$, which is a third of the positive octant.  The dotted lines show the pure states at boundaries $P_Y \cap P_X$ and $P_Y \cap P_Z$. For these boundary states, we know (by Lem.~\ref{WitnessLemma}) that the optimal $\omega$-witness is given by Eq.~\eqref{OptWitness} with the parameter set to $q=\sqrt{2/3}$. For other pure states in $P_Y$, the $\omega$-witness still has the form given by Eq.~\eqref{OptWitness} but the parameter $q$ may be greater than $\sqrt{2/3}$. However, interestingly, the majority of pure states in $P_Y$ have an optimal $\omega$-witness with $q=\sqrt{2/3}$ and these are shown in yellow in this plot.   On the geodesic through $\ket{0}$, $\ket{H}$ and $\ket{+}$, we have that $q=1$. Between this geodesic and the yellow region, $q$ varies continuously from $1$ to $\sqrt{2/3}$ and this intermediate region is shown in green.}
	\label{CanonialRegion}
\end{figure}

\begin{proof}[Proof of Lem.~\ref{WitnessLemma}] We begin by observing that for any $\ket{\Psi}$ there exists a decomposition into stabilizer states such that $\ket{\Psi}=\sum_j c_j \ket{\phi_j}$ and $\xi(\Psi)=(\sum_j |c_j|)^2$.  Given an optimal $\omega$-witness we have 
\begin{equation}
    \xi(\Psi)= |\langle \omega \vert \Psi \rangle|^2= \Big| \sum_j c_j \langle \omega \vert \phi_j \rangle   \Big|^2 .   
\end{equation}
Therefore,
\begin{equation}
     \Big(\sum_j |c_j|\Big)^2 = \Big| \sum_j c_j \langle \omega \vert \phi_j \rangle   \Big|^2 .   
\end{equation}
Given that $| \langle \omega \vert \phi_j \rangle | \leq 1$, the above equality can only hold if $| \langle \omega \vert \phi_j \rangle |= 1$ for every $j$ with $|c_j| > 0$.  In particular, if $\Psi$ is a non-stabilizer state it must have at least two non-zero $c_j$ terms, and there must exist at least two stabilizer states such that $| \langle \omega \vert \phi_j \rangle| = 1$.  We return to use this fact shortly.

Using the set of Pauli matrices as a basis
\begin{equation}
    \kb{\omega}{\omega} = \lambda(\id + q_x X + q_y Y + q_z Z),
\end{equation}
where the coefficients $q_x,q_y$ and $q_z$ are real.
Since $\kb{\omega}{\omega}$ is a rank-1 operator we know
\begin{equation}
\label{radius}
    q_x^2+q_y^2+q_z^2=1,
\end{equation}
and since $\kb{\omega}{\omega}$ is a positive operator we have $\lambda > 0$. Given a valid $\omega$-witness, we can always obtain another valid $\omega$-witness by permuting any of $\{ q_x, q_y , q_z \}$ or changing the signs.  Therefore, the optimal $\omega$-witness for a state in the set $P_Y$ has $q_x \geq q_y$ and $q_z \geq q_y$, since this ordering maximizes $|\bk{\Psi}{\omega}|^2$.  This means that the two stabilizer states with the largest overlap with $\vert \Psi \rangle$ are  $\vert + \rangle$ and $\vert 0 \rangle$.  We showed earlier there must be at least two stabilizer states for which $ |\langle \omega \vert \phi_j \rangle| = 1$, so we conclude $ |\langle \omega \vert + \rangle |= 1$ and $| \langle \omega \vert 0 \rangle |= 1$.  It follows that $q_x=q_z$ and we define $q:=\sqrt{2} q_x =\sqrt{2} q_z$.  Condition Eq.~\eqref{radius} implies that $q_y=\sqrt{1-q^2}$ so we have shown that the optimal $\omega$-witness has the form 
\begin{equation}
    \kb{\omega}{\omega} = \lambda \left( \id + q \frac{X + Z}{\sqrt{2}} + \sqrt{1-q^2} Y \right).
\end{equation}
Furthermore, $| \langle \omega \vert + \rangle| = 1$ implies that
\begin{equation}
 \lambda  = (1+q/\sqrt{2})^{-1} .
\end{equation}
Lastly, we note that the condition $ |\langle \omega \vert +_i \rangle| \leq 1$ entails that $q \geq \sqrt{2/3}$.  Therefore, we know the form of the $\omega$-witness in the set $P_Y$ and proved that $ \sqrt{2/3} \leq q \leq 1$.  Next, consider the special case when the state is at an intersection, such as $P_Y \cap P_Z$.  Then, the optimal $\omega$-witness has the above form determined for the region $P_Y$.  However, the region $P_Z$ only differs by an $F$ rotation, so the optimal $\omega$-witness must have a similar form but with the Pauli operators permuted, so that 
\begin{equation}
\label{RotWitness}
    \kb{\omega}{\omega} = \frac{ \left( \id + p \frac{Y + X}{\sqrt{2}} + \sqrt{1-p^2} Z \right) }{2(1+p/\sqrt{2})}.
\end{equation}
The only way Eq.~\eqref{OptWitness} and Eq.~\eqref{RotWitness} can both be true, is if $q=p=\sqrt{2/3}$. A similar argument holds for $P_Y \cap P_X$ and this proves Lem.~\ref{WitnessLemma}.\end{proof}

Our third lemma shows that every mixed state is contained in an appropriate convex set.
\begin{lem}
\label{1qubitLemma}
    For any single-qubit non-stabilizer state $\rho$, there exists a $\omega$-witness $\vert \omega \rangle$ such that $\rho \in B_{\omega}$ (as defined in Lem.~\ref{EqualityLem}). 
\end{lem}
This implies that for single-qubit states we can leverage Lem.~\ref{EqualityLem} and Lem.~\ref{WitnessLemma} to prove Thm.~\ref{Sandwich2}. 

\begin{proof}[Proof of Lem.~\ref{1qubitLemma}]
We consider individual slices of the Bloch sphere such that $ \mathrm{Tr}[\rho \sigma_F]=f$ where $\sigma_F=(X+Y+Z)/\sqrt{3}$ and $f$ is equal to the inner product between the Bloch vectors representing $\rho$ and $\sigma_F$. A particular $f$ value specifies a slice through the Bloch sphere.  Let us denote $\mathbb{S}_f$ as the set of all states inside this slice. For every non-stabilizer state in the positive octant we have $1/\sqrt{3} \leq f$, and for all normalized states we have $f \leq 1$. Within this slice there are three special, pure states, which are
\begin{align} 
\label{SpecialStatesExp}
   \vert \Psi_{f}^X \rangle \langle \Psi_{f}^X \vert & := \frac{1}{2}(\id + a Y + a Z + \sqrt{1-2a^2}X),     \\ \nonumber
   \vert \Psi_{f}^Y \rangle \langle \Psi_{f}^Y \vert & := \frac{1}{2}(\id + a X + a Z + \sqrt{1-2a^2}Y),     \\ \nonumber
   \vert \Psi_{f}^Z \rangle \langle \Psi_{f}^Z \vert & := \frac{1}{2}(\id + a Y + a X + \sqrt{1-2a^2}Z),
   \end{align}
where $a$ obeys
\begin{equation}
   \sqrt{3}f=2a+\sqrt{1-2a^2}.
\end{equation}
For $1/\sqrt{3} \leq  f \leq 1$, there is a unique $a$ such that $a \in [0,1/\sqrt{3}]$ and $\sqrt{1-2a^2} \in [1/\sqrt{3},1]$. Crucially, these states are the unique pure states
of the following set intersections.
\begin{align} 
\label{SpecialStates}
\vert \Psi_{f}^X \rangle \langle \Psi_{f}^X \vert & \in \mathbb{S}_f \cap P_Y \cap P_Z , \\ \nonumber
  \vert \Psi_{f}^Y \rangle \langle \Psi_{f}^Y \vert & \in  \mathbb{S}_f \cap P_X \cap P_Z , \\ \nonumber
     \vert \Psi_{f}^Z \rangle \langle \Psi_{f}^Z \vert & \in \mathbb{S}_f \cap P_X \cap P_Y. 
\end{align}    
Referring back to Prop.~\ref{intersections}, it is clear that these states must have the form given in Eq.~\eqref{SpecialStatesExp}. 

Notice that these special states are Clifford rotations of each other.  By Lem.~\ref{WitnessLemma} these three special states all have the same optimal $\omega$-witness given by Eq.~\eqref{FaceWitness}.  Since they share their optimal $\omega$-witness, Lem.~\ref{EqualityLem} applies to all convex combinations of states $\{\Psi_{f}^X, \Psi_{f}^Y, \Psi_{f}^Z \}$ as illustrated in Fig.~\ref{SetsSlice}.  Note that $\{\Psi_{f}^X, \Psi_{f}^Y, \Psi_{f}^Z \}$ all have the same value for the extent, since
\begin{equation}
    \xi(\Psi_{f}^X)=\xi(\Psi_{f}^Y)=\xi(\Psi_f^Z)=\frac{1+f}{1+1/\sqrt{3}}. \label{eq:equimagical}
\end{equation}
Therefore, for these states, a mixture of states with the same amount of magic achieves the optimal convex roof extension. That is, each of these states admit an equimagical decompositions.

Next, we consider mixed states outside the convex hull of $\{\Psi_{f}^X, \Psi_{f}^Y, \Psi_{f}^Z \}$ and inside $P_Y \cap \mathbb{S}_f$ as illustrated in Fig.~\ref{SetsSlice}.  We define a set of linearly independent, Hermitian operators
\begin{align}
\label{WonkyBasis}
    \sigma_A  = \frac{X+Z-2Y}{\sqrt{6}} & ,  & 
    \sigma_B  = \frac{X-Z}{\sqrt{2}} ,  
\end{align}
and $\sigma_F$ as defined earlier. The set $\{ \sigma_A, \sigma_B, \sigma_F \}$ is unitarily equivalent to $\{X,Y, Z \}$, so every state can be decomposed as
\begin{align}
\label{WonkyDecomp}
    \rho = \left( \id + r_A \sigma_A + r_B \sigma_B + r_F \sigma_F \right) / 2,
\end{align}
where inside the slice $\mathbb{S}_f$ we have $r_F=f$. The variables $r_A$ and $r_B$ are used for the co-ordinate system in Fig.~\ref{SetsSlice}.

Given a mixed state $\rho$, we can define a pair of pure states $\Phi^+_\rho$ and $ \Phi^-_\rho$, such that
\begin{align}
     \bra{  \Phi^\pm_\rho  }   \sigma_A  \ket{ \Phi^\pm_\rho } & = \mathrm{Tr}[\sigma_A \rho] = r_A \\ \nonumber
       \bra{  \Phi^\pm_\rho  }   \sigma_F  \ket{ \Phi^\pm_\rho } & = \mathrm{Tr}[\sigma_F \rho] = f ,
\end{align}
and the states are pure, so that
\begin{equation}  \label{PhiPMstates}
      \bra{  \Phi^\pm_\rho  }   \sigma_A  \ket{ \Phi^\pm_\rho }^2 +  \bra{  \Phi^\pm_\rho  }   \sigma_B  \ket{ \Phi^\pm_\rho }^2    +  \bra{  \Phi^\pm_\rho  }   \sigma_F  \ket{ \Phi^\pm_\rho }^2 = 1 .
\end{equation}
There are two possible solutions for $\bra{  \Phi^\pm_\rho  }   \sigma_B  \ket{ \Phi^\pm_\rho }$, which leads to
\begin{align}  \label{PhiDefined}
    \kb{\Phi_{\rho}^\pm}{\Phi_{\rho}^\pm} & = \frac{1}{2}\left( \id + r_A \sigma_A \pm \sqrt{1-r_A^2-f^2} \sigma_B + f \sigma_F \right). 
\end{align}
By construction, $\rho$ is a convex combination of $\Phi^+_\rho$ and $ \Phi^-_\rho$.  The geometry is illustrated in Fig.~\ref{SetsSlice}, where the pair of purified states are shown as green dots with $\rho$ located on the line between them.  To deploy Lem.~\ref{EqualityLem}, it remains to prove that $\Phi_{\rho}^\pm$ share an optimal $\omega$-witness.  

The states $\Phi_{\rho}^\pm$ are both in the region $P_Y$, which can be seen from the geometry on Fig.~\ref{SetsSlice} though we also give an algebraic proof in App.~\ref{AppGeometry}.  Due to $ \Phi_{\rho}^\pm \in P_Y$, we can use Lem.~\ref{WitnessLemma} to determine the form of their optimal $\omega$-witnesses.  In Lem.~\ref{WitnessLemma}, the witness $\omega(q)$ had a free parameter $q$ that we had to maximize over. Since $\bra{\omega(q)}\sigma_{B} \ket{\omega(q)}=0$ for any $q$ value, we have
\begin{equation}
\label{equalityofwitness}
    \bra{\omega(q)}\rho \ket{\omega(q)}= |\bk{\Phi_{\rho}^+}{\omega(q)}|^2=|\bk{\Phi_{\rho}^-}{\omega(q)}|^2.
\end{equation}
Performing the maximization over $q$, the optimal $q$ value is the same for $\Phi_{\rho}^+$ and $\Phi_{\rho}^-$ due to Eq.~\eqref{equalityofwitness}.  Therefore, $\Phi_{\rho}^\pm$ share exactly the same optimal $\omega$-witness. This completes the proof of Lem.~\ref{1qubitLemma} and thus also of Thm.~\ref{Sandwich2}.\end{proof}

\begin{figure}
	\includegraphics[width=200pt]{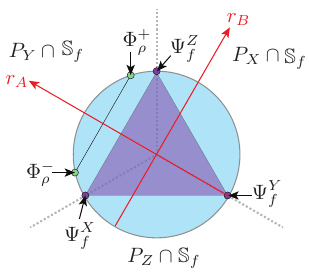}
	\caption{A slice $\mathbb{S}_f$ through the positive octant. States are parameterized by the coordinates $\{ r_A, r_B, r_F \}$ as defined in Eq.~\eqref{WonkyDecomp}.  For the slice $\mathbb{S}_f$ we have $r_F=f$ for some constant $f$.  The axes for the $\{ r_A, r_B \}$ coordinates are shown in red. The slice is divided into thirds corresponding the sets defined in Def.~\ref{OctantDefs} with the intersections of these sets shown with dashed lines.  Where these intersections meet the pure states we label the pure states $\{ \Psi_f^X, \Psi_f^Y, \Psi_f^Z  \}$ defined in Eq.~\eqref{SpecialStates}, and the purple triangle denotes the convex hull of the set $\{ \Psi_f^X, \Psi_f^Y, \Psi_f^Z  \}$.  States outside this convex set are considered as a mixture of two pure states $\ket{\Phi_\rho^{\pm}}$, defined in Eq.~\eqref{PhiDefined} and shown with green dots.}
	\label{SetsSlice}
\end{figure}


\section{Multiplicativity}
\label{sec:Multi}

We now study the behavior of the monotones $\Xi$, $\Lambda$, and $\Lambda^+$ for tensor products of states. It was found by Bravyi \textit{et al.}~\cite{bravyi2018simulation} that $\omega$-witnesses of small dimension are closed under tensor products, formalized as follows.
  \begin{theorem}[\cite{bravyi2018simulation}]
 	\label{StabRank}
  Let $\ket{\omega_j }$ be vectors from a \mbox{1-,} \mbox{2-,} or 3-qubit Hilbert space such that each $\omega_j$ is an $\omega$-witness. Then $\smash{\ket{\Omega}:=\otimes_j \ket{\omega_j}}$ is an $\omega$-witness.
 \end{theorem}
This is a rewording of Corollary 1 and Corollary 3 of Ref.~\cite{bravyi2018simulation}.    From the above result, Ref.~\cite{bravyi2018simulation} further showed that the extent is multiplicative for such tensor products:
 \begin{theorem}[\cite{bravyi2018simulation}]
 	\label{ExtentMulti} 
 	Let $\ket{\psi_j }$ be 1-, 2-, or 3-qubit states. Then
 	\begin{equation}
 		 \xi( \otimes_j \ket{\psi_j} ) = \prod_j   \xi( \ket{\psi_j} ) .
	\end{equation}
\end{theorem}
 Here, we give a related multiplicativity result for several mixed state monotones.
  \begin{theorem}
	\label{Multiplicative2}
	Let $\sigma_j$ be single-qubit states. Then
	\begin{align}
		\Lambda(\otimes_j \sigma_j )  =  \Xi(\otimes_j \sigma_j ) = \Lambda^+(\otimes_j \sigma_j ) =   \prod_j \Lambda^+( \sigma_j),
	\end{align}	
	and furthermore $\otimes_j \sigma_j$ admits an equimagical decomposition  (recall Def.~\ref{RoofExtension}).
\end{theorem}

Prior to this work, there were no known strict multiplicativity results for resource monotones for mixed states in qubit magic theory.  For instance, Howard and Campbell~\cite{Howard17robustness} found that the robustness of magic can be strictly sub-multiplicative, $\mathcal{R}(\rho \otimes \rho) < \mathcal{R}(\rho)^2$ for all non-stabilizer $\rho$ considered, and we discuss this later in this section.   There does exist a multiplicative lower-bound on the robustness of magic, proved using the so-called stab-norm~\cite{Howard17robustness}. However, the lower bounds and upper bounds appear to always be loose and so we have no strict multiplicativity results. Additionally, Raussendorf \textit{et al.}~\cite{raussendorf2019phase} introduced a qubit-based phase-space robustness $\mathcal{R}_{\mathrm{ps}}$ that can behave strictly super-multiplicatively, so that $\mathcal{R}_{\mathrm{ps}}(\rho \otimes \rho) > \mathcal{R}_{\mathrm{ps}}(\rho)^2$ for some $\rho$.  It is natural to wonder if Thm.~\ref{ExtentMulti} or Thm.~\ref{Multiplicative2} could extend to a tensor product of states with arbitrary dimension.  However, in the final stages of completing this work, it was proved that this cannot hold in full generality~\cite{heimendahl2020stabilizer}. \br{It remains an open question whether the monotones satisfy multiplicativity for states composed of a low number of qubits, mirroring the multiplicativity of the extent; indeed, numerical results suggest that $\Lambda^+$ is also multiplicative for mixed two-qubit states.}

\begin{proof}[Proof of Theorem~\ref{Multiplicative2}]
From the definition of $\Xi$ we see that it is manifestly sub-multiplicative.  Combining this observation with Thm.~\ref{Sandwich2} we have that
\begin{equation}
	\label{XiUpper}
	\Xi(\otimes_j \sigma_j) \leq \prod_j \Xi( \sigma_j) = \prod_j \Lambda^+( \sigma_j) 
\end{equation}
holds for all products of single-qubit states.   Strengthening this to strict equality requires us to find a matching lower bound. The proof of Thm.~\ref{Sandwich2} established that for every single-qubit state the optimal $W^+$-witness has the form $\kb{\omega_j}{\omega_j}$ where $\omega_j$ is an $\omega$-witness.  By Thm.~\ref{StabRank}, $\ket{\Omega}=\otimes \ket{\omega_j}$ is also an $\omega$-witness, and consequently $\kb{\Omega}{\Omega}=\otimes \kb{\omega_j}{\omega_j}$ is a $W^+$-witness that can be used to lower bound $\Lambda^+$ as follows
\begin{equation}
		\label{LambdaPlus}
\prod_j \Lambda^+( \sigma_j)  =  \bra{\Omega} \otimes_j \sigma_j \ket{\Omega} \leq	\Lambda^+(\otimes_j \sigma_j) .
\end{equation}
Combining Eq.~\eqref{XiUpper}, Eq.~\eqref{LambdaPlus} and Thm.~\ref{Sandwich1} we obtain
\begin{align}
	\prod_j \Lambda^+( \sigma_j)  & \leq	\Lambda^+(\otimes_j \sigma_j)  \\ \nonumber & \leq \Lambda(\otimes_j \sigma_j) \leq \Xi(\otimes_j \sigma_j) \leq \prod_j \Lambda^+( \sigma_j).
\end{align}
Since the left- and rightmost quantities are the same, all these inequalities must collapse to equalities.

It remains to show that these product states admit equimagical decompositions. This is easily verified by taking an equimagical decomposition for each single qubit state (existence ensured by Thm.~\ref{Sandwich2}) and using this to construct the natural decomposition for the product state.  It then follows immediately from Thm.~\ref{ExtentMulti} that each pure term has equal extent and by the above argument that this is optimal w.r.t to the $\Xi$ monotone.\end{proof}

\section{Comparison with robustness} \label{SecCompareROM}
Here we discuss how our new monotones scale compared to the robustness of magic (recall Def.~\ref{DefRoM}).  While $\Lambda^+, \Lambda$ and $\Xi$ are often equal, the robustness of magic is typically much larger, as formalized in the following result.
\begin{lem}
	\label{RegulaLemma2}
	For any density matrix $\rho$ we have
	\begin{equation} \label{RobBound1}
	\mathcal{R}(\rho) \geq   2 \Lambda^+(\rho) -1.
	\end{equation}
	Furthermore, if $\rho$ is a single-qubit state this tightens to
	\begin{equation} \label{RobBound2}
	\mathcal{R}(\rho) \geq   (1+\sqrt{2}) \Lambda^+(\rho) -\sqrt{2}.
	\end{equation}
\end{lem}
We remark that a similar result to Eq.~\eqref{RobBound1} for $\Lambda$ is claimed in Refs.~\cite{rudolph_2005,regula2017convex}, but the proof contains an error. 

However, because the robustness of magic is not multiplicative, Lem.~\ref{RegulaLemma2} does not tell us much about how the different monotones scale.  For this, we observe that the gap can scale exponentially.

\begin{theorem} \label{ExpGapProp}
Given any single-qubit non-stabilizer state $\rho$, there exists positive real constants $\alpha$ and $\beta$ where  $\alpha > \beta$ and
\begin{align} \label{RoMlowerBound}
    2^{\alpha n} & \leq \mathcal{R}(\rho^{\otimes n}) \\  \label{LambdalowerBound}
    2^{\beta n} & = \Lambda(\rho^{\otimes n}) = \Lambda^+(\rho^{\otimes n})= \Xi(\rho^{\otimes n})  ,
\end{align}
For example, for the Hadamard $\ket{H}$ state we will show that this holds with $\alpha=0.271553$ and $\beta=0.228443$.
\end{theorem}
 
\begin{proof}[Proof of Lem.~\ref{RegulaLemma2}]
The dual formulation of the robustness of magic tells us that $\mathcal{R}(\rho) \geq \mathrm{Tr}[R \rho]$ for any $R$ such that $|\bra{\phi} R \ket{\phi}| \leq 1$ for all $\ket{\phi}$ that are stabilizer states.  We call such an operator an $R$-witness. Note that an $R$-witness is not necessarily positive.  Let $W$ denote the $W^+$-witness such that $\Lambda^+(\rho)=\mathrm{Tr}[W \rho]$.  Now, we consider the operator 
\begin{equation}
    R = \frac{2}{1-s} W -  \frac{1+s}{1-s}\id ,
\end{equation}
where 
\begin{equation}
  s  = \mathrm{min}_{\phi \in \mathrm{S}_n} \bra{\phi} W \ket{\phi}    .
\end{equation}
Next, we show $R$ is indeed an $R$-witness.  For any $\ket{\phi} \in \mathrm{S}_n$,  \begin{align}
    \bra{\phi}R \ket{\phi}& = \frac{2}{1-s} \bra{\phi}W \ket{\phi}-\frac{1+s}{1-s} \\ \nonumber
    & \leq \frac{2}{1-s} -\frac{1+s}{1-s} = 1  ,
\end{align}
where we have used $\bra{\phi}W \ket{\phi} \leq 1$.  Using $\bra{\phi}W \ket{\phi} \geq s$, we similarly obtain
\begin{align}
    \bra{\phi}R \ket{\phi} & \geq \frac{2s}{1-s} -\frac{1+s}{1-s} = -1 ,
\end{align}
Therefore, $R$ is indeed an $R$-witness and we can lower bound the robustness as follows
\begin{align}
\label{RobBoundGeneral}
    \mathcal{R}(\rho) \geq \mathrm{Tr}[R \rho] & = \frac{2}{1-s}\Lambda^{+}(\rho) - \frac{1+s}{1-s} \\ \nonumber
    & = \frac{(2\Lambda^{+}(\rho)-1) - s}{1-s} ,
\end{align}
Since $\Lambda^{+}(\rho) \geq 1$, the right hand side is monotonically increasing with $s$ on the relevant range $s \in [0,1)$. This prompts the question whether we can lower bound $s$.  By definition $s \geq 0$ for any $W^{+}$-witness and so Eq.~\eqref{RobBound1} holds in general.  In the special case of single-qubit states, and assuming for brevity that $\rho \in P_Y$, we know the optimal witness has the form $W=\kb{\omega}{\omega}$ given by Lem.~\ref{WitnessLemma}. Since $\ket{0}$ has the largest possible overlap with $\ket{\omega}$, it follows that $\ket{1}$ must have the smallest possible overlap and one finds that
\begin{equation}
   \bra{1}W\ket{1}= s = (1-q/\sqrt{2})/(1+q/\sqrt{2}).    
\end{equation}
Over the allowed range $q \in [\sqrt{2/3},1]$, we have
\begin{equation}
    s \geq (1-1/\sqrt{2})/(1+1/\sqrt{2}),
\end{equation}
for every optimal single-qubit $W^+$-witness. Substituting this into Eq.~\eqref{RobBoundGeneral} gives Eq.~\eqref{RobBound2}. \end{proof}


\begin{proof}[Proof of Thm.~\ref{ExpGapProp}]  The stab-norm $\mathcal{D}$ has been shown to provide a lower bound on the robustness of magic (see the supplementary material of Ref.~\cite{Howard17robustness} and also Ref.~\cite{Campbell11}), so that for any single-qubit non-stabilizer state $\rho$ we have
\begin{equation} \label{StabNormLower}
 \mathcal{D}(\rho)^n =  \mathcal{D}(\rho^{\otimes n}) \leq \frac{\mathcal{D}(\rho)^n - \frac{1}{2^n}}{1- \frac{1}{2^n}} \leq \mathcal{R}(\rho^{\otimes n}) ,
\end{equation}
where the stab-norm of a single-qubit state is 
\begin{equation}
 \mathcal{D}(\rho) =  \frac{1}{2}(1 + |  \langle X \rangle | +|  \langle Y \rangle |+ |  \langle Z \rangle |).
\end{equation}
Defining $\alpha=\log_2( \mathcal{D}(\rho) )$, we obtain Eq.~\eqref{RoMlowerBound}. For instance $\mathcal{D}(\kb{H}{H})=1.207$ and so $\alpha=0.271553$ for Hadamard states.  

Similarly, Eq.~\eqref{LambdalowerBound} holds due to Thm.~\ref{Multiplicative2} and setting $\beta:= \log_2( \Lambda(\rho) )$.  For instance, $\beta=0.228443$ for Hadamard states.  To show $\alpha > \beta$ in general, we need to show that $\mathcal{D}(\rho) > \Lambda(\rho) $ for all non-stabilizer, single-qubit states. We note that for a single-qubit we have
\begin{equation}
 \mathcal{R}(\rho)=|\langle X \rangle |  + |\langle Y \rangle | + |\langle Z \rangle |
\end{equation}
for any non-stabilizer state. This can be shown by using Eq.~\eqref{StabNormLower} to obtain a lower bound on $\mathcal{R}(\rho)$, with the corresponding upper bound following from a simple quasiprobability decomposition into stabilizer states. Therefore, $\mathcal{R}(\rho)= 2\mathcal{D}(\rho)-1$ and
combining this with Lem.~\ref{RegulaLemma2}, we get
\begin{equation}
  \mathcal{D}(\rho) \geq \frac{1+\sqrt{2}}{2}\Lambda^{+}(\rho)-\left(  \frac{\sqrt{2}-1}{2} \right).
\end{equation}
This reveals that $\mathcal{D}(\rho)>\Lambda^{+}(\rho)$ whenever $\Lambda^{+}(\rho)>1$.
\end{proof}

We further remark that the robustness of magic is not multiplicative and the known upper bounds on $\mathcal{R}(\rho^{\otimes n})$ are loose compared to the lower bound in Eq.~\eqref{RoMlowerBound}.  For instance, Heinrich and Gross~\cite{heinrich2018robustness} showed that for the Hadamard state (or the equivalent $T$-state) $\mathcal{R}(\kb{H}{H}^{\otimes n})= \mathcal{O}(2^{0.368601 n})$ and this is the best known upper bound.

\begin{figure*}
\centering
\begin{tikzpicture}
    \node at (0,0) {\includegraphics{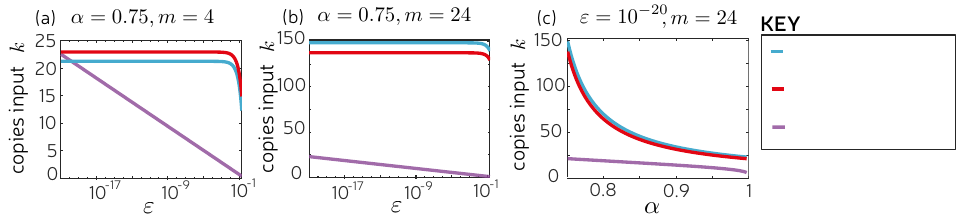}};
    \node  (X) at (6.37,0.98) [align=left]{\footnotesize Bound in Eq.~\eqref{eq:dist_bound1}};
    \node  (X) at ($(X)-(0,0.64)$) {\footnotesize Bound in Eq.~\eqref{eq:dist_bound2}};
    \node  (X) at ($(X)-(-0.19,0.64)$) {\footnotesize Bound from Ref.~\cite{fang_2019-1}};
\end{tikzpicture}
\caption{Comparison of the lower bounds for the number of copies $k$ of the state $\rho=\alpha \vert H \rangle \langle H \vert + (1-\alpha)\id/2$ necessary to distill $m$ copies of $\ket{H}$ with success probability $p=0.9$ and output infidelity $\varepsilon$. In (a), we fix $\alpha=0.75$ and demonstrate that the bounds in this paper can characterize distillation well in a range of physical error regimes even for a small number of target copies ($m=4$), providing a better bound than Ref.~\cite{fang_2019-1} down to $\varepsilon \approx 10^{-21}$. In (b), we show that the bounds substantially improve when $m$ increases. This suggests in particular that, even though the bound of Ref.~\cite{fang_2019-1} gets increasingly better as $\epsilon \to 0$ by construction, in practical regimes its performance can be exceeded by considering a larger number of copies of the distillation target. We show this in (c) by varying the input error parameter $\alpha$ with a fixed small output error of $\varepsilon = 10^{-20}$ and 24 target copies $m$. Our bounds perform better even in the regime of $\alpha$ close to 1, and their performance can be improved further by considering larger $m$. Note that our bounds apply also to the pure-state case ($\alpha=1$), while the bound of Ref.~\cite{fang_2019-1} explicitly applies only to full-rank inputs.}
\label{fig:bounds_comparison}
\end{figure*}
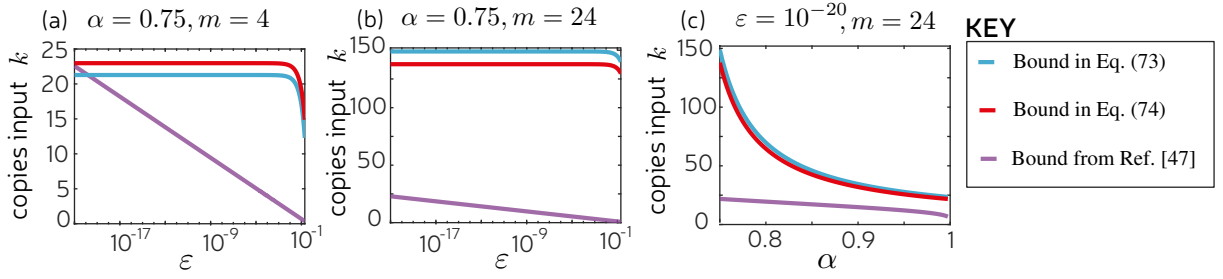

\section{Distillation and asymptotic rates}
\label{sec:distill}

We now consider the scenario of distillation --- that is, consuming many copies of an input resource state $\rho$ to prepare copies of some target state --- and show how the quantifiers we introduced characterize this task. Firstly, it is easy to see using the multiplicativity of the magic monotones $\Lambda$, $\Lambda^+$, and $\Xi$ for single-qubit systems together with their monotonicity that, whenever there exists a stabilizer operation taking $\rho^{\otimes k} \to \sigma^{\otimes m}$ for some single-qubit $\rho$ and $\sigma$, we must have
\begin{equation}\begin{aligned}
\frac{k}{m} \geq \frac{\log \Lambda(\sigma)}{\log \Lambda(\rho)} ,
\end{aligned}\end{equation}
and analogously for the other magic monotones. This already allows one to obtain insightful no-go results on the transformations between stabilizer states and gate synthesis, along the lines considered in~\cite{Howard17robustness} but without the need to perform the difficult computation of the monotones for many copies of a state.

However, in practical settings it is often desirable to go beyond such exact transformations and consider protocols which allow for imperfect conversion. Our quantifiers can yield bounds for the efficiency of more general distillation protocols and their asymptotic rates. We focus on the magic monotone $\Lambda^+$ as it is the most efficiently computable out of the three and gives us the tightest bounds. 
A useful property of $\Lambda^+$ is its monotonicity on average under general probabilistic protocols:
specifically, we have~\cite{regula2017convex}
\begin{equation}\label{eq:prob_monotonicity}
\Lambda^+(\rho) \geq \sum_i p_i \Lambda^+\!\left(\frac{O_i(\rho)}{p_i}\right),
\end{equation}
where each $O_i$ is a stabilizer-preserving quantum operation that need not preserve trace (i.e. $O_i(\sigma) \propto \omega \in \bar{\mathrm{S}}_{n} \; \forall \,\sigma \in \bar{\mathrm{S}}_{n}$), the overall quantum operation $\sum_i O_i$ preserves trace, and $p_i = \Tr (O_i(\rho))$ denotes the probability that the input state $\rho$ is transformed to the output $O_i(\rho)$.

The most general representation of a distillation protocol is then an operation which takes $k$ copies of a given input $\rho$ to $m$ copies of some desired pure output state $\psi$, up to error $\varepsilon$ in fidelity, and succeeding with probability $p$. All such protocols are limited as follows.
\begin{theorem}\label{thm:distillation_bounds}
Let $\rho$ be any $n$-qubit quantum state, and $\psi$ a pure state of at most 3 qubits. If there exists a probabilistic (that is, not necessarily trace-preserving) stabilizer operation taking $\smash{\rho^{\otimes k} \to p \tau}$, where $\tau$ is a state such that 
$\smash{ \bra{\psi^{\otimes m}} \tau \ket{\psi^{\otimes m}} \geq 1-\varepsilon}$, then it necessarily holds that
\begin{equation}\begin{aligned}\label{eq:dist_bound1}
	k \geq \frac{\log p + \log(1-\varepsilon) + m \log F(\psi)^{-1}}{\log \Lambda^+(\rho)}
\end{aligned}\end{equation}
and
\begin{equation}\begin{aligned}\label{eq:dist_bound2}
	k \geq p \left( \frac{\log (1-\varepsilon) + m \log F(\psi)^{-1}}{\log \Lambda^+(\rho)} \right)
\end{aligned}\end{equation}
where $\smash{F(\psi) = \max_{\ket{\phi} \in \mathrm{S}_n} \left|\bk{\psi}{\phi}\right|^2}$ denotes the stabilizer fidelity~\cite{bravyi2018simulation}.
\end{theorem}
The above establishes two bounds on the least number of copies of $\rho$ necessary to perform the distillation of $\psi$ up to the desired accuracy, characterizing the dependence on the resources contained in both $\rho$ (as quantified by $\Lambda^+$) and in $\psi$ (as quantified by stabilizer fidelity $F$). Note that either of the two bounds can perform better, depending on the values of the parameters (see Fig.~\ref{fig:bounds_comparison}).
\begin{proof}
By sub-multiplicativity of $\Lambda^+$ we have $\Lambda^+(\rho)^k \geq	\Lambda^+(\rho^{\otimes k})$.  By monotonicity \br{under probabilistic protocols} (see Eq.~\eqref{eq:prob_monotonicity}) we have $\Lambda^+(\rho^{\otimes k}) \geq  p \Lambda^+ (\tau)$.  Because
$\kb{\psi}{\psi}^{\otimes m}/F(\psi^{\otimes m}) $ is a \smash{$W^+$-witness} and hence a feasible solution to the dual form of $\Lambda^+$, we arrive at 
\begin{equation}\begin{aligned}
\Lambda^+(\rho)^k 	&\geq p \operatorname{Tr} \left(\tau \frac{ \kb{\psi}{\psi}^{\otimes m}}{F(\psi^{\otimes m})} \right) \\
	&\geq p \frac{1-\varepsilon}{F(\psi^{\otimes m})} ,
\end{aligned}\end{equation}
If $\psi$ is any single-qubit, two-qubit, or three-qubit pure state, then $F(\psi^{\otimes m}) = F(\psi)^m$ (see Ref.~\cite{bravyi2018simulation} or Thm.~\ref{StabRank}), and so
 \begin{equation}\begin{aligned}\label{eq:dist_bound_Lambda}
\Lambda^+(\rho)^k 	&\geq p \frac{1-\varepsilon}{F(\psi^{\otimes m})} ,
\end{aligned}\end{equation}
 Taking the logarithm, we get
\begin{equation}\begin{aligned}
	k &\geq \log_{\Lambda^+(\rho)} \left[ p (1-\varepsilon) F(\psi)^{-m} \right],
\end{aligned}\end{equation}
which is precisely Eq.~\eqref{eq:dist_bound1}. Alternatively, if we use $\log \Lambda^+$ instead of $\Lambda^+$ in the above derivation (noting that $\log \Lambda^+$ also decreases on average under stabilizer protocols due to concavity of the logarithm), we obtain the bound in Eq.~\eqref{eq:dist_bound2}.\end{proof}
 
  Another bound of this kind, which also explicitly depends on $\Lambda^+(\rho)$ and $F(\psi)$ but exhibits a different scaling with respect to $\varepsilon$, was recently obtained in~\cite{fang_2019-1}. We compare the performance of the bounds in Fig.~\ref{fig:bounds_comparison}.
  
  When $p=1$, Eq.~\eqref{eq:dist_bound_Lambda} recovers a related recent bound of~\cite{regula_2019-1}.  When $\varepsilon = 0$, we obtain a benchmark on the performance of all distillation protocols which distill the target exactly, but can fail with a certain probability:
  \begin{equation}\begin{aligned}
	\frac{k}{m p} \geq \frac{\log F(\psi)^{-1}}{\log \Lambda^+(\rho)}.
\end{aligned}\end{equation}
  This was considered for odd-dimensional qudits in~\cite{veitch2014resource,wang2018efficiently} as the  ``distillation efficiency''.

Additionally, the ultimate constraints on the convertibility between two states are often characterized in the asymptotic limit, where we are interested in the best achievable rate $R(\rho \to \psi)$ at which $k$ copies of $\rho$ can be approximately converted to $k R(\rho \to \psi)$ copies of $\psi$, with the error $\varepsilon$ of this conversion vanishing in the limit $k \to \infty$. Using Eq.~\eqref{eq:dist_bound2} with $p=1$, any such rate must satisfy
\begin{equation}\begin{aligned}\label{eq:dist_asymptotic}
	R(\rho \to \psi) \leq \frac{\log \Lambda^+(\rho)}{\log F(\psi)^{-1}},
\end{aligned}\end{equation}
which gives a semidefinite programming upper bound on the asymptotic rate of transformation between any state $\rho$ and a pure state $\psi$ of at most three qubits.

States of interest in magic state distillation include $\ket{H}$ and $\ket{F}$~\cite{bravyi2005universal}.
These states obey a Clifford symmetry in the following sense;  we say a state $\ket{\psi}$ is Clifford symmetric if there exists an Abelian subgroup $\mathcal{C}_\psi$ of the Clifford group such that: (\textit{i}) $C \ket{\psi} = \ket{\psi}$ for all $C \in \mathcal{C}_\psi$; and (\textit{ii}) $\ket{\psi}$ is the unique state with this property up to a global phase. Crucially, any such state has extent equal to the inverse of its stabilizer fidelity~\cite{bravyi2018simulation}, so $\xi(\psi) = F(\psi)^{-1}$. When we already know the value of the extent $\xi(\psi)$, we only need to evaluate $\Lambda^+(\rho)$ to determine the bounds in Thm.~\ref{thm:distillation_bounds} and in \eqref{eq:dist_asymptotic}.
For instance, for the rate of transformation from any state to a Clifford symmetric state of up to three qubits, we get $	R(\rho \to \psi) \leq \frac{\log \Lambda^+(\rho)}{\log \Lambda^+(\psi)}$.
Asymptotic distillation rates of the magic states $\ket{H}$ and $\ket{F}$ are bounded by
\begin{align}
	R(\rho\to \kb{H}{H}) & \leq \frac{\log \Lambda^+(\rho)}{\log(4 - 2 \sqrt{2})} \\
	R(\rho\to\kb{F}{F}) & \leq \frac{\log \Lambda^+(\rho)}{\log(3 - \sqrt{3})}
\end{align}
where we used the known values of $\xi(\ket{H})$ and $\xi(\ket{F})$~\cite{bravyi2018simulation,beverland2019lower}.

The above can be compared with the recent bounds obtained in~\cite{wang2018efficiently} for qudit magic state theory, as our approach similarly yields computable upper bounds on the rates of distillation, although applicable to the fundamentally important case of qubit systems.

We can alternatively show these asymptotic results by \br{using the regularized relative entropy of magic~\cite{veitch2014resource} to} bound the achievable rates of transformations  
between states using any stabilizer protocol. \br{Specifically, define} $r_\infty(\rho) = \lim_{n\to\infty} \frac{1}{n} r(\rho^{\otimes n})$ where $r(\rho) = \min_{\sigma \in \bar{\mathrm{S}}_{n}} D(\rho\|\sigma)$ \br{and $D(\rho\|\sigma) = \mathrm{Tr}(\rho \log \rho) - \mathrm{Tr}(\rho \log \sigma)$ is the quantum relative entropy}. Then the ratio $r_\infty(\rho)/r_\infty(\sigma)$ provides a general upper bound on the rate $R(\rho \to \sigma)$ of the transformation from $\rho$ to $\sigma$ using stabilizer protocols~\cite{veitch2014resource}.
This upper bound is achievable whenever the states can be reversibly interconverted~\cite{veitch2014resource} or when the set of stabilizer protocols is relaxed to the class of operations which asymptotically preserve the set of stabilizer states~\cite{brandao_2015}. Using the bounds $r(\rho) \leq \log \Lambda^+ (\rho)$ for arbitrary states~\cite{datta_2009} and $r(\psi) \geq - \log F(\psi)$ for pure states~\cite{datta_2009}, we similarly obtain Eq.~\eqref{eq:dist_asymptotic}. Notice also that $r(\psi) = \log \Lambda^+(\psi)$ for any Clifford symmetric state, and $r_\infty(\psi) = \log \Lambda^+(\psi)$ for a Clifford symmetric state of at most three qubits.

Finally, we remark that the best known magic state distillation protocols perform many orders of magnitude worse than our best bounds. It remains a considerable challenge to close this gap.

\section{Classical simulation algorithms}
\label{sec:simulate}

 \prx{%
 
 \jrs{In this section we introduce three simulation techniques, each associated to one of the magic monotones defined earlier.}
In subsection \ref{sec:simulate1}, we generalize quasiprobability-based methods \cite{pashayan15,Howard17robustness,OakRidge17} for estimating Born rule probabilities or expectation values of bounded observables \jrs{up to additive error}.  
\hakop{We use a novel choice of frame consisting of the set of stabilizer dyads and extend quasiprobabilistic techniques to accommodate this choice. By doing so, we are able to reduce the sampling overhead compared to previous qubit quasiprobability simulators~\cite{Howard17robustness,OakRidge17,Seddon19}, resulting in a runtime proportional to the  dyadic negativity squared, $\Lambda(\rho)^2$.}

\hakop{In subsection \ref{sec:simulate2}} we describe \jrs{a simulator which extends stabilizer rank methods \cite{bravyi2016improved,bravyi2018simulation},} previously only defined for pure states, to arbitrary mixed-state \jrs{input}s. \jrs{The algorithm simulates the sampling of bit strings from a quantum circuit (i.e. by measurement of a subset of qubits in the computational basis). We show that the classical distribution we sample from is $\delta$-close in $\ell_1$-norm to the quantum distribution, and that under modest assumptions each string is sampled in average time $\mathcal{O}(\Xi(\rho)/\delta^{-3})$, where $\Xi$ is the mixed-state extent. When an equimagical decomposition is known \br{(recall Def.~\ref{RoofExtension})}, this becomes the worst-case runtime. This reduces the runtime by a factor of $\delta^{-1}$ compared to the results of Ref. \cite{bravyi2018simulation}}.

Finally, in subsection \ref{sec:simulate3}, \jrs{we introduce the constrained path simulation technique, which efficiently estimates Pauli expectation values \jrs{or Born rule probabilities} up to additive error} on $\chan(\rho)$ for stabilizer channel $\chan$ and non-stabilizer state $\rho$. The technique approximates the magic state $\rho$ with the stabilizer part of a feasible solution to the generalized robustness problem, Eq. \eqref{eq:LambdaPlusDefn2}.  Whereas the dyadic frame simulator outputs estimates to arbitrarily high precision but with runtime that grows with $\Lambda(\rho)$, here the estimate is efficiently computed, but \jrs{with unavoidable additive error lower bounded as $\mathcal{O}(\Lambda^+(\rho))$, where $\Lambda^+$ is the generalized robustness.}

}
\subsection{\jrs{Dyadic frame simulator}}
\label{sec:simulate1}
\subsubsection{Quasiprobability simulators\label{sec:quasi}}
\jrs{Before describing our first algorithm, we briefly review the principles of classical simulation using quasiprobabilities. A very general notion of quasiprobability simulation was introduced by Pashayan, Wallman and Bartlett \cite{pashayan15}. A specific instance of this type of simulator is defined by fixing a frame, a finite set of operators that forms a basis for the space of Hermitian operators acting on a Hilbert space. This basis need not be orthonormal and can in general be over-complete. For concreteness we consider the algorithm introduced by Howard and Campbell \cite{Howard17robustness}, where the frame is the set of pure stabilizer state projectors. We can define the $n$-qubit stabilizer frame  as:}
\begin{equation}
    \mathcal{G}_n = \{\op{\phi}: \ket{\phi}\in \mathrm{S}_n\},
\end{equation}
\jrs{so that the convex hull of $\mathcal{G}_n$ is precisely $\bar{\mathrm{S}}_{n}$, the set of mixed stabilizer states. Indeed, $\mathcal{G}_n$ forms an over-complete basis for the Hermitian operators on $\mathbb{C}^{2^n}$. It follows that \emph{any} $n$-qubit density $\rho$ matrix has at least one decomposition of the form:}
\begin{equation}
    \rho = \sum_j q_j \op{\phi_j}, \quad \op{\phi_j} \in \mathcal{G}_n, \quad \sum_j q_j =1,  \label{eq:quasidecomp}
\end{equation}
\jrs{with $q_j$ real. Consider the simulation task of estimating the Born rule probability $\mu = \Tr[\Pi \chan(\rho)]$, where $\Pi$ is a stabilizer projector and $\chan$ is an efficiently simulable channel, but $\rho$ is a general mixed magic state. Given a known quasiprobability decomposition as per Eq. \eqref{eq:quasidecomp}, we can rewrite:}
\begin{equation}
    \mu = \sum_j q_j  \Tr[\Pi \chan(\op{\phi_j})] 
    =\sum_j \frac{\abs{q_j}}{\onenorm{q}}  E_j.
\end{equation}
\jrs{where $E_j =  \onenorm{q} \mathrm{sign}(q_j)\Tr[\Pi \chan(\op{\phi_j})]$. Now $\abs{q_j}/\onenorm{q}$ are non-negative and sum to unity, so form a proper probability distribution. The Howard and Campbell \cite{Howard17robustness} algorithm goes as follows. First fix a total number of samples $M$. Then: 
\begin{enumerate}
\item For each integer $k$ from 1 to $M$, sample index $j_k$ from the distribution $\{\abs{q_j}/\onenorm{q}\}$. \label{samplestep}
    \item Compute each $\widehat{E}_k = E_{j_k}$.
    \item Output $\widehat{\mu} = \frac{1}{M} \sum_k \widehat{E}_k$.
\end{enumerate}
It is clear that since $\Tr[\Pi \chan(\op{\phi_j})]$ amounts to evaluating a stabilizer circuit, each $\widehat{E}_k$ can be efficiently computed using the standard Gottesman-Knill tableaux method \cite{gottesman1998theory,aaronson04improved}. Moreover, one can easily check that $\mathbb{E}(\widehat{\mu}) = \sum_j (|q_j|/\onenorm{q}) E_j = \mu$, so the algorithm gives an unbiased estimator for the Born rule probability. However, due to the renormalization of the distribution, each estimate $\widehat{E}_k$ takes a value in the range $[-\onenorm{q},+\onenorm{q}]$, increasing the variance of the estimator. From Hoeffding's inequalities \cite{Hoeffding1963}, the probability that $\widehat{\mu}$ is far from the expected value $\mu$ is bounded as:}
\begin{equation}
    \Pr\{|\widehat{\mu} - \mu| \geq \epsilon\} \leq 2 \exp\left(-\frac{M \epsilon^2}{2\onenorm{q}^2 }\right).
\end{equation}
\jrs{It follows that to estimate the value within additive error at most $\epsilon$ with probability at least $1-p_{\mathrm{fail}}$, we must set the number of samples so that $M \geq 2 \onenorm{q}^2 \epsilon^{-2} \log (2 p_{\mathrm{fail}}^{-1}) $.  Recall from Definition \ref{DefRoM} that robustness of magic $\mathcal{R}(\rho)$ is defined as the minimal $\onenorm{q}$, so the worst-case runtime for the Howard and Campbell algorithm scales with (at least)  $\mathcal{R}(\rho)^2$. }

\jrs{Whereas in the simulation model described above, the frame was comprised of stabilizer projectors $\op{\phi}$, in our dyadic frame simulator} we extend the frame to include dyads $\kb{L}{R}$ where $\ket{L}$ and $\ket{R}$ may be different stabilizer states. An operator is now considered free if it is in the convex hull of the dyads $e^{i\theta}\kb{L}{R}$. \br{Importantly, a density matrix $\sigma$ can be written in this form if and only if $\sigma \in \bar{\mathrm{S}}_n$.} Non-free density matrices are then expressed as generalized quasiprobability distributions over the set of $n$-qubit dyads, where the ``quasiprobabilities'' are now complex-valued. As we shall see, the associated dyadic negativity \jrs{ quantifies the} classical simulation overhead for estimating Born rule probabilities on a non-free state.  \jrs{In the next subsection we illustrate our new algorithm by giving a simplified version where the stabilizer circuit elements are restricted to be probabilistic mixtures of Clifford gates. We subsequently generalize the algorithm to cover all completely stabilizer-preserving circuits with magic state inputs.}

\subsubsection{Dyadic frame simulator\label{sec:dyadic_simple}}

\jrs{We assume the following restricted simulation setting. The input to the algorithm will consist of (i) a known dyadic decomposition of a mixed magic state $\rho = \sum_j \alpha_j \kb{L_j}{R_j}$; (ii) a circuit description comprised of a list of $T$ quantum operations $\{O^{(1)}, \ldots, O^{(T)}\}$; and (iii) a stabilizer projector $\Pi$ representing the outcome of a Pauli measurement. We stipulate that each $O^{(t)}$ must be a convex mixture of unitary Clifford channels, $O^{(t)} = \sum_k p^{(t)}_k U_k (\cdot) U_k^\dagger $ , and we assume this decomposition is known and can be efficiently sampled from. The output of the algorithm is again an estimate for the Born rule probability $\mu = \Tr[\Pi \chan(\rho)]$, where $\chan = O^{(T)} \circ \ldots \circ O^{(1)}$. Note that the above restriction on $O^{(T)}$ means that we can write the whole circuit as an ensemble over unitary Clifford gates:}
\begin{equation}
\chan(\cdot) = \sum_{\vec{k}} p_{\vec{k}} U_{\vec{k}} (\cdot) U_{\vec{k}}^\dagger 
\end{equation}
\jrs{where $\vec{k} = (k_1, k_2, \ldots, k_T)$ is a vector that represents a Clifford trajectory through the circuit $U_{\vec{k}}= U_{k_T} \ldots U_{k_2} U_{k_1}$, and $p_{\vec{k}}$ is a product distribution and so can be efficiently sampled from. The algorithm proceeds by sampling elements from the initial distribution, computing an estimate, repeating many times and averaging. The procedure for generating one sample is as follows:} \jrs{
\begin{enumerate}
    \item Randomly select index $j$ with probability $|\alpha_j|/\onenorm{\alpha}$.
    \item Randomly select trajectory $\vec{k}$ with probability $p_{\vec{k}}$.
    \item Compute final dyad: \begin{equation}{e^{i\theta'_{j,\vec{k}}} \ket{L'_{j,\vec{k}}} \bra{R'_{j,\vec{k}}} = e^{i\theta_j} U_{\vec{k}}\ket{L_j} \bra{R_j}U_{\vec{k}}^\dagger }.\end{equation}
    \item Compute sample $\widehat{E} = \onenorm{\alpha} \re\{ e^{i\theta'_{j,\vec{k}}} \bra{R'_{j,\vec{k}}}\Pi \ket{L'_{j,\vec{k}}}\} $.
\end{enumerate}}

\jrs{In step 3, $e^{i\theta'}$ is a final global phase taking into account the initial phase  $e^{i\theta_j} = \alpha_j/\abs{\alpha_j} $ and the action of the sampled unitary circuit on $\ket{L_j} $ and $\ket{R_j} $ respectively. Whereas the Howard and Campbell algorithm dealt with projectors $\op{\phi}$, so that any global phase on $\ket{\phi}$ is unimportant, here $\ket{L_j} $ and $\ket{R_j} $  can represent different stabilizer states and the combined phase can affect both the magnitude and sign of the real-valued sample $\widehat{E}$. While the original tableaux method used in the Gottesman-Knill theorem does not track this global phase, subsequent extensions of the method show that the update can be efficiently computed, including the phase~\cite{Bravyi16stabRank,bravyi2016improved,bravyi2018simulation}. We can also efficiently compute the complex inner product $\bk{L}{R}$ for any pair of stabilizer states~\cite{Bravyi16stabRank,bravyi2016improved,bravyi2018simulation}. Thus steps 3 and 4 are efficient. Note that the two parts of the dyad $U_{\vec{k}}\ket{L_j}$ and  $\bra{R_j}U_{\vec{k}}^\dagger = (U_{\vec{k}}\ket{R_j})^\dagger$ are updated independently. }

The algorithm is completed by repeating steps 1-4 $M$ times. We can check that the method gives an unbiased estimator for the target Born rule probability:
\begin{align}
    \mathbb{E}(\widehat{E}) &= \sum_{j,\vec{k}} \frac{|\alpha_j|}{\onenorm{\alpha}}p_{\vec{k}} \left(\onenorm{\alpha} \re\{ e^{i\theta'_{j,\vec{k}}} \bra{R'_{j,\vec{k}}}\Pi \ket{L'_{j,\vec{k}}}\}  \right)  \\
    & = \re\{ \sum_{j, \vec{k}} e^{i \theta_j}|\alpha_j| p_\vec{k} \Tr[\Pi U_{\vec{k}}\ket{L_j} \bra{R_j}U_{\vec{k}}^\dagger]\}\\
    & = \re\{ \Tr[\Pi \sum_\vec{k} p_{\vec{k}}U_\vec{k} (\sum_j \alpha_j\ket{L_j} \bra{R_j}) U_\vec{k}^\dagger )] \} \\
    & = \Tr[\Pi \chan(\rho)].
\end{align}
\jrs{We can therefore apply Hoeffding's inequality in the same way as for the standard quasiprobability technique, and using the fact that each $\widehat{E}$ is in the range $[-\onenorm{\alpha}, + \onenorm{\alpha}]$, we find that the total number of samples needed to achieve additive error $\epsilon$ and success probability $1-p_{\mathrm{pfail}}$ is:}
\begin{equation}
    M   \geq 2 \onenorm{\alpha}^2 \epsilon^{-2} \log (2 p_{\mathrm{fail}}^{-1}).
\end{equation}
\jrs{When the decomposition of $\rho$ is optimal with respect to dyadic negativity as per Definition \ref{def:dyadic}, we have that $\onenorm{\alpha} = \Lambda(\rho)$. When this holds, the worst-case runtime of the algorithm will be $\mathcal{O}(\Lambda(\rho)^2)$.}

\jrs{This simplified algorithm can be used only in the case where the stabilizer circuit is a convex mixture of unitary Clifford operations, so channels are restricted to be unital. Our main goal, however, is to admit more general stabilizer channels. In particular, extending to adaptive Clifford circuits with mixed magic state inputs allows for universal quantum computation \cite{bravyi2005universal}. We now sketch how the dyadic frame simulator can be extended to admit all completely stabilizer-preserving channels. Full pseudocode and technical proofs of validity and performance are given in Appendix \ref{app:dyadicframe}.}

\jrs{The simplicity of the restricted simulator derives from the fact that unitary operations preserve the norm of the state vector. This means that when each circuit element $O^{(t)}$ can be decomposed as a convex mixture of unitary gates, the probability of choosing a particular trajectory $\vec{k}$ through the circuit depends only on the coefficients $p_k^{(t)}$ and is independent of initial state.} 
\jrs{Conversely, Kraus decompositions of non-unital channels always include non-unitary operators. When these channels appear in a circuit, transition probabilities for selecting one of the non-unitary operators must be computed on the fly as we step through the circuit. These transition probabilities depend not only on the initial state, but on the Kraus operators selected in previous steps, so they cannot be pre-computed.}
\jrs{Note that in general a channel may be decomposed as a mixture of $N_U$ unitary and $N_K$ non-unitary Kraus operators:}
\begin{equation}
    O(\cdot) = \sum_k^{N_U} p_k U_k (\cdot) U^\dagger_k + \sum_{k'}^{N_K} q_{k'} K_{k'} (\cdot) K_{k'}^\dagger.
\end{equation}
\jrs{The probability of picking one of the $N_U$ operators $U_k$ can simply be read off from the coefficients $p_k$. We can infer the \emph{total} probability $1-\sum_k p_k$ that the trajectory chosen will be from among the $N_K$ \emph{non}-unitary operators, but individual transition probabilities for each $K_{k'}$ must be computed based on the initial state.}
\jrs{ In Appendix \ref{app:dyadicframe} we show how appropriate transition probabilities can be computed efficiently even when the input \jrs{operator} is not a state but a dyad, provided that we restrict to channel decompositions where $N_U, N_K \leq \mathrm{poly}(n)$. We call such a decomposition \textit{simulable}. This leads to the following theorem.}

\begin{theorem}\label{thm:simulationJS}
    \jrs{Let  $\rho = \sum_j \alpha_j \ket{L_j}\bra{R_j}$, be a known dyadic decomposition of an initial $n$-qubit state,  where $\alpha\in \mathbb{C}$} and the probability distribution $\{|\alpha_j|/\onenorm{\alpha}\}$ can be efficiently sampled. Let $\chan = O^{(T)} \circ \ldots \circ O^{(1)}$, where each $O^{(t)} \in \mathcal{O}_n$ is a completely stabilizer-preserving channel. 
    \jrs{Suppose that every $O^{(t)}$ has a known simulable decomposition.} Then, given a stabilizer projector $\Pi$, \jrs{we can estimate the Born rule probability $\mu = \Tr(  \Pi \chan [\rho ] )$ within additive error $\epsilon$, with success probability at least $1- p_{\mathrm{fail}}$ and  worst-case runtime}
	\begin{equation}
 \frac{\|\alpha\|_1^2}{\epsilon^2}  \log(p_{\mathrm{fail}}^{-1})T
{\rm poly}(n).
	\end{equation}	
	Furthermore, if the dyadic decomposition of $\rho$ is optimal then $\|\alpha\|_1$ can be replaced by $\Lambda(\rho)$.
\end{theorem}

By exploiting a dyadic frame, the negativity of the quasiprobability distribution and algorithm runtime is greatly reduced compared to previous work~\cite{Howard17robustness}, with an improved exponential scaling of the runtime (recall Thm.~\ref{ExpGapProp}).

\subsection{The density-operator stabilizer rank simulator}
\label{sec:simulate2}

\subsubsection{Prior art: the BBCCGH simulator}
\label{sec:BBCCGHrecap}
\jrs{Here we briefly review a previous stabilizer rank-based simulation method, which we will refer to as BBCCGH in what follows (after the authors' initials~\cite{bravyi2018simulation}). BBCCGH simulates sampling length $w$ bit strings $\vec{x}$ from measurements on pure magic states $\ket{\psi}$ with runtime linear in pure-state extent $\xi(\psi)$, and represents the prior state of the art in stabilizer rank techniques. In subsequent sections we improve on this algorithm and generalize to mixed states.BBCCGH can be decomposed into two main subroutines: \textsc{Sparsify}, which generates a sparse approximation of the target state, and \textsc{FastNorm}, which estimates Born rule probabilities $\norm{\Pi\ket{\psi}}^2$ up to multiplicative error. By calling \textsc{FastNorm} $\mathcal{O}(w)$ times, one estimates a chain of conditional probabilities so as to successively sample the outcome for each bit of $\vec{x}$ in turn. It is crucial that the error is multiplicative, as this ensures that the output distribution of the classical algorithm is close in $\ell_1$-norm to the quantum distribution $P(\vec{x}) = |\braket{\vec{x}}{\psi}|^2$.}

\jrs{In general, stabilizer rank simulators exploit the fact that a}ny pure quantum state $\ket{\psi}$ can be expressed as a linear combination of stabilizer states,
\begin{align}
\smash{\ket{\psi} = \sum_{j=1}^k c_j \ket{\phi_j} },     \label{eq:stab-decomposition}
\end{align}
where $\ket{\phi_j}$
are stabilizer states and $c_j$ are complex. \jrs{The exact stabilizer rank $\chi(\ket{\psi})$ is the smallest number of terms $k$ needed for a given state $\ket{\psi}$ ~\cite{Bravyi16stabRank,bravyi2016improved,bravyi2018simulation}. Computations can be performed in $\mathrm{poly}(n,k)$ time} by treating each stabilizer term in turn (albeit $k$ can grow exponentially with $n$). \jrs{In particular Bravyi et al. \cite{bravyi2018simulation} showed that $\textsc{FastNorm}$ can estimate $\norm{\psi}^2$ up to multiplicative error by repeatedly generating a random number $\eta_A = 2^n |\braket{\phi_A}{\psi}|^2$, where $\ket{\phi_A}$ is randomly drawn from a subset of stabilizer states known as equatorial states. Evaluating $\eta_A$ amounts to computing $k$ stabilizer inner products (one for each term of $\psi$), which can be done efficiently by exploiting a canonical representation of stabilizer states known as CH-form \cite{bravyi2018simulation}. We summarize this result of Bravyi et al. in the following theorem:}
\begin{theorem}\label{thm:fastnorm} \emph{\cite{bravyi2018simulation}}
    Given an un-normalized $n$-qubit vector $\ket{\psi} = \sum_{j=1}^\chi c_j \ket{\phi_j}$ with $\chi$ stabilizer terms in its decomposition, there exists a classical algorithm \textsc{FastNorm} that outputs a random variable $\eta$ such that:
    \begin{equation}
        (1-\epsilon)\norm{\psi}^2 \leq \eta \leq (1 + \epsilon) \norm{\psi}^2
    \end{equation}
    with probability greater than $1-p_{\mathrm{fail}}$ in worst-case runtime $\mathcal{O}(\chi n^3 \epsilon^{-2} \log(p_{\mathrm{fail}}^{-1}))$.
\end{theorem}
\jrs{By applying this algorithm to projected vectors $\Pi \ket{\psi}$, where $\Pi$ is a stabilizer projector, one can estimate Born Rule probabilities $\norm{\Pi \ket{\psi}}^2$.}

\jrs{ If one was to apply $\textsc{FastNorm}$ directly to the ideal state $\ket{\psi}$, the runtime would be $\mathcal{O}(\chi(\psi))$, where $\chi$ is the exact stabilizer rank. However, computing the exact stabilizer rank is intractable for many-qubit states. Instead the strategy of BBCCGH is to approximate  $\ket{\psi}$ with a sparsified $k$-term vector $\ket{\Omega}$ of smaller stabilizer rank, using the subroutine \textsc{Sparsify} (Figure \ref{fig:sparsify}).}
\jrs{Bravyi ~\textit{et al.} \cite[Lem.~6]{bravyi2018simulation})  showed that for any pure state $\ket{\psi}$ and any integer $k > 0$, one can use \textsc{Sparsify} to generate random (un-normalized) states $\ket{\Omega}$} with $k$ stabilizer terms such that:
\begin{equation}
\label{BBCCGHsparse}
  \mathbb{E}(  \| \ket{\psi}-\ket{\Omega} \|^2 ) \leq \frac{\|c\|_1^2 }{k} ,
\end{equation}
\jrs{In Appendix \ref{app:BBCCGHtracenorm} we present a simple corollary of \cite[Lem.~6]{bravyi2018simulation}), which} implies that
\begin{equation}
\label{BBCCGHsparse2}
 \mathbb{E}(   \| \kb{\psi}{\psi}-\kb{\Omega}{\Omega} \|_1 ) \leq 2 \frac{\|c\|_1 }{\sqrt{k}}  + \frac{\|c\|_1^2 }{k}  \approx  2 \frac{\|c\|_1 }{\sqrt{k}} .
\end{equation}
For any target precision $\delta_S>0$, choosing $k$ so that: 
\begin{equation}
\label{BBCCGHsparse2k}
     k \geq \frac{4 \| c \|^2_1}{\delta_S^2}  ,
\end{equation}
we get
\begin{align}
  \mathbb{E}(\| \kb{\psi}{\psi}-\kb{\Omega}{\Omega} \|_1) \leq \delta_S + \mathcal{O}(\delta_S^2) . \label{BBCCGH-sparsification-lemma}
\end{align}
Recall from Section \ref{sec:magicmonotones} that the minimal value of  $\| c \|^2_1$ is precisely the pure-state extent. We call Eq.~\eqref{BBCCGH-sparsification-lemma} \jrs{combined with the lower bound on $k$ in Eq.~ \eqref{BBCCGHsparse2k}} the BBCCGH sparsification lemma \cite{bravyi2018simulation}. 
\begin{figure}[tbp]
\centering
\includegraphics[trim={6cm 9cm 6cm 9.5cm},clip,width=0.32\textwidth]{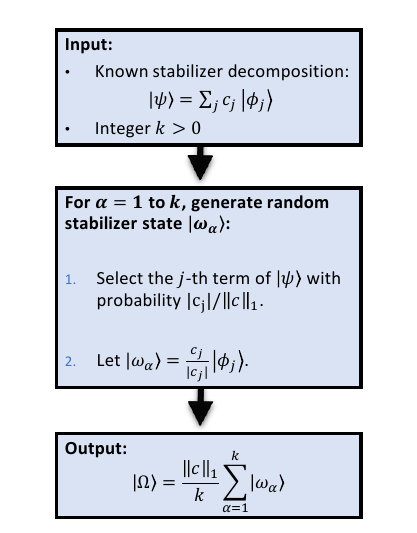}
\caption{The \textsc{Sparsify} procedure introduced by Bravyi \textit{et al.} \cite{bravyi2018simulation} The exact state $\ket{\psi}$ is approximated by an unnormalized random vector $\ket{\Omega}$ with $k$ stabilizer terms that is, on average, $\delta_S^2$-close in the Euclidean norm, where $\delta_S = \onenorm{c}^2/k$. }
\label{fig:sparsify}
\end{figure}
\jrs{Thus, with high probability and subject to some technical caveats discussed in Appendix \ref{app:whypostselect}, by combining the two subroutines BBCCGH simulates sampling from the quantum distribution $P(\vec{x}) = |\braket{\vec{x}}{\psi}|^2$ up to an error $\delta_S$ in runtime $\| c \|^2_1 \delta_S^{-4} \mathrm{poly}(n,w)$ . Assuming an optimal decomposition ($\xi(\psi) = \| c \|^2_1$), the runtime therefore scales linearly with pure-state extent. }

\jrs{Below we improve on this algorithm in three main respects: (i) we extend the simulator from pure to mixed magic state inputs, so that the average-case runtime is proportional to the mixed-state extent $\Xi$ defined in Section \ref{sec:magicmonotones}; (ii) we show that important cases admit decompositions such that $\Xi$ quantifies the worst-case runtime; and (iii) we derive a} new sparsification lemma that improves the runtime over that implied by Eq.~\eqref{BBCCGH-sparsification-lemma} by a factor of $1/\delta_S$ with minor caveats.  \jrs{Our new sparsification lemma also avoids some technical difficulties that arise when applying the BBCCGH sparsification lemma in a practical algorithm.   The runtime improvements originate from working in the density-operator picture even when the input magic state is pure.}
  
\jrs{We will first discuss the proof of our new lemma, before applying it} to classically simulate bit-string sampling.  
While our ideas naturally apply to estimating Born probabilities, and can be extended to propagate an initial state through a \jrs{noisy} stabilizer circuit prior to measurement, we omit this for brevity.

\subsubsection{Sparsification lemma}
\label{sec:sparsification}
 The input to \jrs{the subroutine \textsc{Sparsify}} is an integer $k$ and pure state $\ket{\psi}$ with known stabilizer decomposition \eqref{eq:stab-decomposition} with \jrs{coefficient} vector $c$. \jrs{The output is a randomly chosen $k$-term sparsification of $\ket{\psi}$,}
\begin{equation}
    \ket{\Omega}=\frac{\|c\|_1}{k} \sum_{\alpha=1}^k \ket{\omega_\alpha} 
    \label{eq:OmegaDef},
\end{equation}
where each $\ket{\omega_\alpha}$ is an i.i.d. sampled stabilizer state $(c_j/|c_j|)\ket{\phi_j}$ for some $j$, (see Figure \ref{fig:sparsify}), so that we have \cite{bravyi2018simulation} 
\begin{equation}
    \mathbb{E}(\ket{\omega_\alpha}) = \frac{\ket{\psi}}{\onenorm{c}} \Rightarrow \mathbb{E}(\ket{\Omega}) = \ket{\psi}.
\end{equation}
Since the output $\ket{\Omega}$ of \textsc{Sparsify} \jrs{is a random superposition of non-orthogonal terms, it} need not have unit norm. In \cite{bravyi2018simulation}, after obtaining a state $\ket{\Omega}$ from \textsc{Sparsify}, one estimates its Euclidean norm, and discards the state if its norm is not close to 1.  
\jrs{A state post-selected in this way will be close to the target state with high probability. See Appendix \ref{app:whypostselect} for a discussion of why this post-selection is necessary.}

Here, we instead consider a sampling strategy that avoids post-selecting $\ket{\Omega}$. After \textsc{Sparsify} gives a random $\ket{\Omega}$, we renormalize so that it has unit norm.  Furthermore, instead of bounding the error between an individual sample and the target state $\ket{\psi}$, we bound the error between $\ket{\psi}$ and the whole ensemble as captured by the density matrix
\begin{equation}
		\rho_1 := \mathbb{E}\left[  \frac{\kb{\Omega}{\Omega}}{\bk{\Omega}{\Omega}}   \right] = \sum_{\Omega} \Pr(\Omega) \frac{\kb{\Omega}{\Omega}}{\bk{\Omega}{\Omega}}.\label{eq:expectedSparse}
\end{equation}	
Intuitively, this is advantageous because coherent errors in each sample smooth out to a less harmful stochastic error.  Similarly, randomizing coherent errors improves error bounds in the setting of circuit compilation~\cite{wallman15,campbell2017shorter,HastingMixing,campbell2019random}.

Our \jrs{refinement} to the BBCCGH sparsification lemma is \jrs{summarized in} the following theorem. 
\begin{theorem}
	\label{NewSparse3} Let $\rho_1$ be the mixed state in Eq. \eqref{eq:expectedSparse}.
Let $\ket{\psi}$ be an input state with known decomposition $\ket{\psi} = \sum_j c_j \ket{\phi_j}$, where $\ket{\phi_j}$ are stabilizer states, and let $c$ be the vector whose elements are the coefficients $c_j$ and
\begin{align}
 C_{\psi,c} = \norm{c}_1 \sum_j |c_j| |\bk{\psi}{\phi_j}|^2,
 \label{defi:constant-C}
\end{align}
\jrs{Then }there is a critical precision $\delta_c = 8(C_{\psi,c}-1)/ \|c\|_1^2$
such that for every target precision $\delta_S$ for which $\delta_S \ge \delta_c$, we can sample \jrs{pure states} from \jrs{an ensemble} $\rho_1$\jrs{, where} every pure state drawn from $\rho_1$ has stabilizer rank at most $\lceil 4 \|c \|_1^2 /\delta_S \rceil$ and 
\begin{equation}	
\label{NewSparseEq1}
  	\|  \rho_1  - \kb{\psi}{\psi}\|_1 \leq \delta_S + \mathcal{O}(\delta_S^2).
\end{equation}	
When $\ket{\psi}$ is a Clifford magic state (Recall Sec.~\ref{sec:distill}), the critical precision becomes $\delta_c = 0$, and sampled pure states in $\rho_1$ have stabilizer rank at most $\lceil(2+\sqrt{2}) \|c \|_1^2 /\delta_S \rceil $.
\end{theorem}
Notice that the theorem sets a critical precision $\delta_c$ above which we can achieve the promised $1/\delta_S$ improvement in the runtime over BBCCGH \cite{bravyi2018simulation}. At higher precision, our runtime has the same leading order $\delta_S$-scaling as BBCCGH but with a much smaller constant prefactor, so still yields improved performance.  Furthermore, for the important case of noisy $T$ states, they are Clifford magic states so the improvement holds across all $\delta$.

The proof of Theorem \ref{NewSparse3} follows from two lemmata. 
\jrs{Here we sketch the proof strategy, deferring full technical proofs to Appendices \ref{app:tracenormlemma} and \ref{AppVarianceBound}. The first lemma captures the idea that the ensemble \eqref{eq:expectedSparse} can be made close in the trace norm to the target state $\op{\psi}$ by choosing sufficiently large $k$, up to a term that depends on the variance of $\bk{\Omega}{\Omega}$. The second lemma then bounds this variance in terms of $C_{\psi,c}$, $\onenorm{c}$ and $k$.}

\begin{lem}[\textbf{Ensemble sampling lemma}]
	\label{tracenormlemma}
Given a state $\ket{\psi}=\sum_j c_j \ket{\phi_j}$ where $\phi_j$ are stabilizer states, we can sample from \jrs{an ensemble} $\rho_1$ such that every sampled pure state has stabilizer rank $\leq k$ and 
\begin{equation}	
\label{RigorousSparseEq1}
	\|  \rho_1  - \kb{\psi}{\psi}\|_1 \leq \frac{2 \|c\|_1^2}{k}  + \sqrt{ \mathrm{Var}[\bk{\Omega}{\Omega}] }
\end{equation}	
where $\ket{\Omega}$ is the random sparsified vector defined in Eq.~\eqref{eq:OmegaDef}.
\end{lem}
\jrs{The first step in proving Lemma \ref{tracenormlemma} is to note that we can use the triangle inequality to split the problem into two parts:}
\begin{equation}	
	\|  \rho_1  - \kb{\psi}{\psi}\|_1 \leq	\|  \rho_1  -\rho_2  \|_1 + \|\rho_2  - \kb{\psi}{\psi}\|_1 ,
\end{equation}
\jrs{where $\rho_2 = \mathbb{E}(\op{\Omega})/\mathbb{E}(\bk{\Omega}{\Omega})$. The first term is upper bounded by $\mathrm{Var}[\bk{\Omega}{\Omega}]$. The second term can then be evaluated in terms of $\mathbb{\ket{\omega}}$, and turns out to be upper bounded by $2\onenorm{c}^2/k$. Full technical details are given in Appendix \ref{app:tracenormlemma}. It remains to bound the variance of $\bk{\Omega}{\Omega}$.}
\begin{lem}[\textbf{Sparsification variance bound}]
\label{lem:variancelemma}
Using the notation of Lemma \ref{tracenormlemma} the variance of $\braketself{\Omega}$ satisfies the bound
\begin{equation}
    \mathrm{Var}[ \bk{\Omega}{\Omega} ] \leq  \frac{4(C-1)}{k} +  \frac{2 \norm{c}_1^4}{k^2} + \mathcal{O}\left(\frac{C}{k^3}\right),
    \label{eq:variancegeneralbound}
\end{equation}
where $C=C_{\psi,c}$ is as given in Eq.~\eqref{defi:constant-C}.
When $\ket{\psi}$ is a Clifford magic state as defined in Ref.~\cite{bravyi2018simulation}, 
\begin{equation}
     \mathrm{Var}[ \bk{\Omega}{\Omega} ] \leq 
     \frac{ 2 \norm{c}_1^4}{k^2} + \mathcal{O}\left(\frac{1}{k^3}\right). \label{eq:CliffordmagicVar}
\end{equation}
\end{lem}
\begin{figure}
    \centering
    \includegraphics{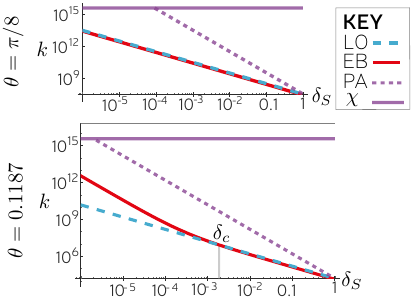}
    \caption{For the target state  $\ket{\psi}=(\cos(\theta) \ket{0}+\sin(\theta)\ket{1})^{\otimes 100}$ with two choices of $\theta$, we plot the trace norm error $\delta_S$ when using a sparsification with $k$ terms. EB (Exact Bound) refers to Eq.~\eqref{RigorousSparseEq1} and is valid for all $\delta_S$,  with the variance exactly bounded by Eq.~\eqref{varianceExact}. LO (Leading Order) refers to our Thm.~\ref{NewSparse3} expression $k=4\|c\|_1^2/\delta_S$, and is valid provided $\delta_S \geq \delta_c$ with $\delta_c$ highlighted by a vertical line.   
    Note $\theta=\pi/8$ corresponds to the Clifford magic state $\ket{H}$, for which $\delta_c=0$.  PA (Prior Art) shows the cost of Ref.~\cite{bravyi2018simulation}. The exact stabiliser rank is $\chi$ (see Thm. 2 of Ref.~\cite{bravyi2018simulation}) and this is an upper bound on PA.  When $C\neq 1$ and $\delta_S < \delta_c$, then EB shows that there is still a large saving even though LO is not valid in this regime. To better understand the deviations of EB from LO, we refer the reader to App.~\ref{AppVarianceBound} and in particular Fig.~\ref{fig:Cliffordness}, and to the discussion at the end of Section \ref{sec:simulate2}.}
    \label{fig:DeltaScaling}
\end{figure}

\jrs{In Appendix \ref{AppVarianceBound}, we prove Lemma \ref{lem:variancelemma} by expanding $\mathrm{Var}[ \bk{\Omega}{\Omega} ]$ as a series of terms of the form $\mathbb{E}(\bk{\omega_\alpha}{\omega_\beta}\bk{\omega_\lambda}{\omega_\mu})$, treating the cases where the indices $\alpha$, $\beta$, $\lambda$ and $\mu$ are all distinct (and therefore correspond to i.i.d. random variables), where $\alpha=\beta$, but $(\alpha, \lambda, \mu)$ are all distinct and so on. }

By combining Lemmas \ref{tracenormlemma} and \ref{lem:variancelemma} we can now prove Thm. \ref{NewSparse3}. Substituting $k = 4 \onenorm{c}^2/\delta_S$ and $\delta_S  \geq 8(C-1)/\onenorm{c}^2$ into Eq.~\eqref{eq:variancegeneralbound}, we obtain 
\begin{equation}
    \mathrm{Var}[ \bk{\Omega}{\Omega} ] \leq \frac{\delta_S ^2}{4}\left(1 + \mathcal{O}\left(\frac{\delta_S }{\onenorm{c}^4}\right)\right),
\end{equation}
and hence, using $\sqrt{1 +x} \leq 1 + x$ for $x\geq 0$:
\begin{equation}
    \sqrt{\mathrm{Var}[ \bk{\Omega}{\Omega} ]} \leq \frac{\delta_S }{2} + \mathcal{O}(\delta_S ^2).
    \label{eq:sqrtvarbound}
\end{equation}
Using \eqref{eq:sqrtvarbound} with the expression for $k$ and Lemma $\ref{tracenormlemma}$, we have
\begin{equation}
    \|  \rho_1  - \kb{\psi}{\psi}\|_1 \leq \delta_S  + \mathcal{O}(\delta_S ^2) .
\end{equation}
This proves the main result of Theorem \ref{NewSparse3}. 
When $\ket{\psi}$ is a Clifford magic state, Eq.~\eqref{eq:CliffordmagicVar} combined with Lemma \ref{tracenormlemma} gives
\begin{equation}
    \|  \rho_1  - \kb{\psi}{\psi}\|_1 \leq \frac{(2 +\sqrt{2})\onenorm{c}^2}{k} + \mathcal{O}\left( \frac{1}{k^2}\right).
\end{equation}
This allows us to obtain Eq.~\eqref{NewSparseEq1} by setting  ${k = \lceil(2+\sqrt{2}) \|c \|_1^2 /\delta_S\rceil  }$, completing the proof.

We have shown that whenever the constraint on the target precision $\delta_S $ is greater than a critical precision, one can sample from an ensemble of sparsified states $\rho_1$ that is $\delta_S $-close in the trace norm to $\braketself{\psi}$, where the number of stabilizer terms is $k=4 \onenorm{c}^2 / \delta_S $. Compared to the BBCCGH \cite{bravyi2018simulation} sparsification lemma where  $k=4 \onenorm{c}^2 /\delta_S ^2$, we see a factor $1/\delta_S $ improvement. If the target precision is smaller than the critical precision, one can compute $C$ and obtain a sharp bound on the trace norm error by using Lemmas \ref{tracenormlemma} and \ref{lem:variancelemma} directly. In this case, the $\delta_S ^{-2}$ scaling of $k$ is recovered, but with a prefactor often much smaller than in the original BBCCGH sparsification lemma. This is because one typically finds that $(C-1)/\onenorm{c}^2 \ll 1$ for many-qubit magic states. We illustrate this in Fig.~\ref{fig:DeltaScaling}, where we compare the sharpened trace-norm bound of our Lemmata with that of Ref. \cite{bravyi2018simulation} for states of the form $\ket{\psi_N} = \ket{\psi}^{\otimes N} $, where $\ket{\psi}$ are single-qubit magic states, and $N = 100$.   While $\delta_S  \geq 8(C-1)/\onenorm{c}^2$ we have a quadratic improvement over Eq.~\eqref{BBCCGHsparse2}, but even in the high-precision regime, we find a significant reduction in $k$. 

\subsubsection{Bit-string sampling \jrs{from mixed magic states}} 
\label{sec:simulate2:bitstringsampling}

\begin{figure}[htbp]
\centering
\includegraphics[trim={5.5cm 7.2cm 5.5cm 7.2cm},clip,width=0.35\textwidth]{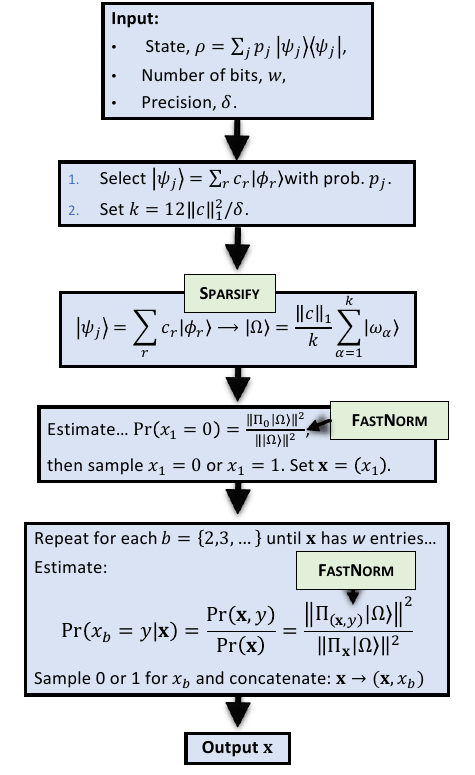}
\caption{\jrs{Our procedure for classically sampling a single length $w$ bit-string given an $n$-qubit state with known decomposition $\rho=\sum_j p_j \op{\psi_j}$, where each pure state $\ket{\psi_j}$ in turn has a known stabilizer decomposition $\ket{\psi_j} = \sum_r c_r^{(j)} \ket{\phi_j}$. We assume that $\delta$ is greater than the critical precision defined in Thm. \ref{NewSparse3}.  The procedure is a variant of that given in \cite{bravyi2018simulation} for pure state input, making use of two subroutines from that work, \textsc{Sparsify} and \textsc{FastNorm}, as described above. In the main text we describe how we have improved on the sparsification step, and extended the simulator to admit mixed states as input. The factor of 12 in the initial step arises from optimization of the error budget between the sparsification error $\delta_S$ and fast norm estimation error $\epsilon$; we set $\delta_S = \delta/3$ and $\epsilon = 2 \delta/3$ (See Appendix \ref{app:bitsampling}).}}
\label{fig:bitsampling}
\end{figure}
\jrs{Consider the setting where we have an $n$-qubit mixed magic state $\rho$, and we measure a subset of $w$ qubits in the computational basis (i.e. we measure Pauli $Z$ for each qubit), thereby generating a random bit string $\vec{x}$ of length $w$ representing the measurement outcomes. Without loss of generality we can assume we measure the first $w$ qubits. Let $\Pi_{\vec{x}} = \op{\vec{x}} \otimes \id_{n-w} $ be the projector representing the outcome where we obtain bit-string $\vec{x}$. Then the probability of obtaining the string $\vec{x}$ is given by the Born rule:}
\begin{equation}
    P(\vec{x}) = \Tr[\Pi_{\vec{x}} \rho]
\end{equation}

\jrs{We call $P$ the \textit{quantum} probability distribution. Here we deal with the simulation task of classically sampling from a probability distribution $\Psim(\vec{x})$ over $w$-bit strings $\vec{x}$ such that $\Psim$ is $\delta$-close in $\ell_1$-norm to $P$, with high probability. Our algorithm is closely related to the sampling algorithm given in Bravyi \textit{et al.} \cite{bravyi2018simulation}, differing in two key respects: (i) whereas the Bravyi \textit{et al.} simulator is defined only for pure states, our variant admits general mixed states; and (ii) we take advantage of our improved sparsification lemma to reduce runtime. We also avoid a post-selection step needed for the Bravyi \textit{et al.} algorithm (see Appendix \ref{app:whypostselect}). Figure \ref{fig:bitsampling} gives the key steps for our procedure for sampling a single bit-string. Full pseudocode is given in Appendix \ref{app:bitsampling}. The main steps in the algorithm are (1) the sampling of a random pure state $\ket{\psi_j}$ from the ensemble $\rho = \sum_j p_j \op{\psi_j}$, (2) a call to the subroutine \textsc{Sparsify} to generate the $k$-term approximation $\ket{\Omega}$, and (3) computation of a chain of conditional probabilities using at most $2w + 1$ calls to $\textsc{FastNorm}$. In Appendix \ref{app:bitsampling} we prove the validity of the algorithm, and give a full analysis of the runtime. Here we first sketch the proof before discussing the runtime improvement over Bravyi \textit{et al.} \cite{bravyi2018simulation}. In what follows we assume that $\delta > \delta_c$ as defined in Eq.~\eqref{NewSparse3}, returning to the case of arbitrary precision at the end of the section.}

\jrs{We want to show that the classical probability distribution $\Psim$ satisfies:}
\begin{equation}
\onenorm{\Psim - P} \leq \delta + \mathcal{O}(\delta^2)    
\end{equation}
\jrs{where $P$ is the quantum distribution. We split the proof into two parts. First, we consider an idealized algorithm $\textsc{Exact}$ where the calls to \textsc{FastNorm} are replaced by an oracle that can compute $\norm{\Pi_{\vec{x}}\ket{\Omega}}$ exactly given $k$-term sparsification $\ket{\Omega}$. Let $\Pex(\vec{x})$ be the probability of obtaining the string $\vec{x}$ as the output of $\textsc{Exact}$. We will first show that $\Pex$ is $\delta_S$-close to the quantum distribution $P$ in $\ell_1$-norm, and then show that $\Psim$ is $\epsilon$-close to $\Pex$. We then split the error budget so that $\delta = \delta_S + \epsilon$. In Appendix \ref{app:bitsampling} we show that the optimal strategy is to set $\delta_S = \delta/3$ and $\epsilon= 2 \delta/3$.}

\jrs{Let $\vec{x}_m =(x_1, \ldots, x_m)$ be the bit string comprised of the first $m$ bits of $\vec{x}$, and let 
$\ket{\Omega_m} = \Pi_{\vec{x}_m} \ket{\Omega} $ be the projection of the first $m$ qubits of $\ket{\Omega}$. By inspection of the last two steps of Figure \ref{fig:bitsampling}, we can multiply the chain of conditional probabilities and obtain the probability of sampling $\vec{x}$ from \textsc{Exact} given fixed sparsification $\ket{\Omega}$:}
\begin{align}
    \Pr(\vec{x}|\Omega) & = \mathrm{Pr}(x_1) \mathrm{Pr}(x_2|\vec{x}_1) \ldots \mathrm{Pr}(x_w|\vec{x}_{w-1}) \\
            & = \frac{\norm{\ket{\Omega_1}}^2}{\norm{\ket{\Omega}}^2}  \frac{\norm{\ket{\Omega_2}}^2}{\norm{\ket{\Omega_1}}^2}\ldots 
            \frac{\norm{\Pi_\vec{x}\ket{\Omega}}^2}{\norm{\ket{\Omega_{w-1}}}^2}  \label{eq:probchain}\\
            & = \frac{\norm{\Pi_\vec{x}\ket{\Omega}}^2}{\norm{\ket{\Omega}}^2}=\Tr\left[\Pi_\vec{x} \frac{\op{\Omega}}{\bk{\Omega}{\Omega}}\right]. 
\end{align}
\jrs{Thus \textsc{Exact} simulates sampling from the quantum state $\ket{\Omega}/\norm{\ket{\Omega}}$ exactly; any error arises solely from the sparsification procedure. Now consider that randomly choosing a pure state $\psi_j$ from $\rho = \sum_j p_j \op{\psi_j}$, generating a random approximation $\ket{\Omega}$ using $\textsc{Sparsify}$ and then normalizing is equivalent to sampling a pure state from the ensemble:}
\begin{equation}
    \sigma = \sum_j p_j \sum_\Omega \Pr(\Omega| \psi_j) \frac{\op{\Omega}}{\bk{\Omega}{\Omega}} = \sum_j p_j \rho_1^{(j)}
\end{equation}
\jrs{where $ \Pr(\Omega| \psi_j)$ is the probability of \textsc{Sparsify} outputting the vector $\ket{\Omega}$ , and $\rho_1^{(j)}$ is the expected projector $\mathbb{E}(\op{\Omega}/\bk{\Omega}{\Omega})$ as defined in Eq.~\eqref{eq:expectedSparse}, both conditioned on the input to \textsc{Sparsify} being $\ket{\psi_j}$. From our argument above it follows that $\Pex(\vec{x}) = \Tr[\Pi_\vec{x} \sigma]$. A key conceptual difference between our method and that of Bravyi \textit{et al.} \cite{bravyi2018simulation} is that while the BBCCGH sparsification results are concerned with the fidelity between the target state $\ket{\psi}$ and a single randomly chosen sparsification $\ket{\Omega}$, here we compare the target state $\rho$ with the full ensemble over sparsifications $\sigma$. From our sparsification lemma (Thm. \ref{NewSparse3}), for each pure state $\ket{\psi_j}$, we have that $\onenorm{\rho_1^{(j)}-\op{\psi_j}}\leq \delta_S + \mathcal{O}(\delta_S^2)$. It follows that $\onenorm{\sigma-\rho}\leq \delta_S + \mathcal{O}(\delta_S^2)$, and so:}
\begin{equation}
    \onenorm{\Pex - P} \leq \delta_S + \mathcal{O}(\delta_S^2). \label{eq:PexVsP}
\end{equation}

\jrs{Next we argue that $\Pex$ is $\epsilon$-close to $\Psim$, the distribution arising from our full classical algorithm. Recall from Thm.~\ref{thm:fastnorm} that \textsc{FastNorm} is able to output estimates for $\norm{\Omega_m}^2$ up to some relative error $\epsilon_{\mathrm{FN}}$, which we can set arbitrarily small (at the cost of increased runtime). One can show (see Appendix \ref{app:bitsampling}) that estimating the chain of $w$ conditional probabilities \eqref{eq:probchain} using \textsc{FastNorm} leads to a total relative error $3w \epsilon_{\mathrm{FN}}$ in the distribution sampled from, i.e.}
\begin{equation}
    (1 - 3 w\epsilon_{\rm FN}) \Pex(\vec{x}) \leq \Psim(\vec{x}) \leq (1 + 3 w\epsilon_{\rm FN}) \Pex(\vec{x}),
\end{equation}
\jrs{so to achieve relative error $\epsilon$ we must set $\epsilon_{\mathrm{FN}} = \epsilon/3w$. This governs the runtime of $\textsc{FastNorm}$. By combining this result with Eq. \eqref{eq:PexVsP} we have:}
\begin{equation}
    \onenorm{\Psim - P} \leq \delta + \mathcal{O}(\delta^2).
\end{equation}

\jrs{To analyze the runtime of our simulator, we define:}
\begin{equation}
    \widetilde{\Xi} = \sum_j p_j \onenorm{c^{(j)}}^2, \label{eq:Xitwiddle}
\end{equation}
\jrs{where $c^{(j)}$ is the vector of coefficients in the decomposition $\ket{\psi_j} = \sum_r c^{(j)} \ket{\phi_j}$. Recall from Thm.~\ref{thm:fastnorm} that for an $n$-qubit state vector with $k$ terms, the runtime of \textsc{FastNorm} is $\mathcal{O}(k n^3 \epsilon_{\mathrm{FN}}^{-2})$. From the previous discussion, if we selected the $j$-th pure state in the decomposition of $\rho$, we will have set $k\propto \onenorm{c^{(j)}}^2 \delta^{-1}$ and $\epsilon_{\mathrm{FN}} \propto \delta w^{-1}$. In a single run of the full algorithm, \textsc{FastNorm} is called $\mathcal{O}(w)$ times. Therefore the runtime to generate a single $w$-length bit string is $T = \mathcal{O}(\onenorm{c^{(j)}}^2 w^3 n^3 \delta^{-3}) $ with probability $p_j$. So from Eq. \eqref{eq:Xitwiddle}, the \textit{average-case} runtime is $\mathcal{O}(\widetilde{\Xi} w^3 n^3 \delta^{-3})$. Through $\widetilde{\Xi}$, this average-case runtime is sensitive to the particular decomposition of $\rho$ supplied to the simulator. In the case where the decomposition is optimal with respect to the mixed-state extent $\Xi$ (Def.~\ref{RoofExtension}), we have $\widetilde{\Xi}=\Xi(\rho)$, so that the average-case runtime is linear in $\Xi(\rho)$.  Recall from Section \ref{sec:single} that all single-qubit states admit an equimagical decomposition (Thm.~\ref{Sandwich2}) that naturally extends to all tensor products of single-qubit states. In that case $\onenorm{c^{(j)}}^2 = \Xi(\rho)$ for all $j$, so that we can give the \textit{worst-case} runtime as $\mathcal{O}(\Xi(\rho))$.}

\jrs{The runtime scaling of $\mathcal{O}(\delta^{-3})$ holds provided that the sparsification error $\delta_S$ is not smaller than the critical threshold $\delta_c = 8(C-1)/ \|c\|_1^2$, where $C$ is defined in Eq. \eqref{defi:constant-C}. However, the algorithm is still valid for the case of arbitrary precision, $\delta_S<\delta_c$. In this case we recover the same leading order scaling as Bravyi \textit{et al.}, namely $\mathcal{O}(\delta^{-4})$ \cite{bravyi2018simulation}, but typically with a prefactor improved by several orders of magnitude (see Fig. \ref{fig:DeltaScaling}). A detailed technical analysis is provided in Appendix \ref{app:bitsampling}, including proof of the following theorem, which captures the results discussed above.}

\begin{theorem}\label{thm:bit-string-sampling}
    Let $\rho = \sum_j p_j \kbself{\psi_j}$ be an $n$-qubit state where every pure state  has a known stabilizer decomposition $\smash{{\ket{\psi_j} = \sum_r c_r^{(j)} \ket{\phi_r} }}$. 
    For every $\ket{\psi_j}$, \jrs{let ${C_j = \norm{c^{(j)}}_1 \sum_r |c_r^{(j)}| |\bk{\psi}{\phi_r}|^2}$}.  Let ${\widetilde{\Xi}= \sum_j p_j \onenorm{c^{(j)}}^2}$, and let $D = \mathrm{max}\{(C_j - 1)/\onenorm{c^{(j)}}^2\}$. Then for any $p_{\mathrm{fail}} >0$, and $\delta \geq 24D$ there exists a classical algorithm that, with success probability $(1-p_{\mathrm{fail}})$, samples a bit-string $\vec{x}$ of length $w$ with probability $\Psim(\vec{x})$ such that:
    \begin{equation}
        \onenorm{\Psim - P} \leq \delta + \mathcal{O}(\delta^2),
    \end{equation}
    where $P(\vec{x})=\Tr(\Pi_{\vec{x}}\rho)$, and $\Pi_{\vec{x}}= \kbself{\vec{x}} \otimes \id_{n-w}$ is a projector.
    The algorithm returns $\vec{x}$ with random runtime $T$ where the \jrs{average} runtime is
    \begin{equation}
        \mathbb{E}(T) = \mathcal{O}(w^3 n^3\widetilde{\Xi} \delta^{-3}\log(w/p_{\mathrm{fail}})).
        \label{eq:expectedRuntime}
    \end{equation}
    If the decomposition of $\rho$ is optimal with respect to the definition \eqref{eq:XiDef}, then the expected runtime is $\mathcal{O}(\Xi(\rho))$. Moreover, if the state decomposition is equimagical, then the right side of \eqref{eq:expectedRuntime} also bounds the \jrs{worst-case} runtime.
    
    If arbitrary precision $\delta \leq 24D$ is required, this can be achieved at the cost of an increased runtime:
    \begin{equation}
        \mathbb{E}(T) = \mathcal{O}(w^3 n^3\widetilde{\Xi} (\delta^{-3} + 3D\delta^{-4})\log(w/p_{\mathrm{fail}})).
        \label{eq:expectedRuntimeSharpened}
    \end{equation} 

\end{theorem}

\subsection{Constrained path simulator}
\label{sec:simulate3}
\jrs{In a standard quasiprobability simulator, the target state $\rho = \sum_j q_j \sigma_j$ is decomposed as an affine combination of frame elements $\sigma_j$ that are in some sense easy to simulate. We can alternatively combine all the positive and negative contributions into convex combinations $\sigma_+$ and $\sigma_-$ respectively, so that the decomposition is rewritten: $\rho = \lambda \sigma_+ - (\lambda-1) \sigma_-$, for some $\lambda = \sum_{q_j \geq 0} q_j \geq 1$. The standard sampling procedure for estimating $\langle E \rangle$ for some observable $E$ can then be divided into two steps: (i) randomly sample the positive or negative path with probability $\lambda/\onenorm{q}$ or $(\lambda - 1)/\onenorm{q}$, where $\onenorm{q} = 2\lambda -1$; (ii) sample an individual frame element from the selected convex combination $\sigma_{\pm}$. As explained in Section \ref{sec:quasi}, the number of samples needed to achieve any accuracy $\epsilon>0$ is $\mathcal{O}(\onenorm{q}^2\epsilon^{-2})$. Viewed in this way, we see that any increased runtime for simulating magic states arises in step (i) rather than step (ii). In other words, sampling a frame element from the convex combination $\sigma_{\pm}$ does not incur additional overhead.}

\jrs{An alternative strategy is to constrain sampling to the positive path so that step (i) is avoided. This is equivalent to making the approximation $\rho \approx \lambda \sigma_+ $, and comes at the cost of an unavoidable systematic error of size $|(\lambda - 1)\Tr[E \sigma_-]|$. However, an advantage to this approach is that since $\Tr[E \sigma_-]$ is no longer evaluated explicitly, $\sigma_-$ need not be an efficiently simulable state. Therefore it is natural to connect this strategy with primal solutions to the generalized robustness problem:} 
\begin{equation}
    \rho = \Lambda^+ (\rho) \sigma - (\Lambda^+(\rho) - 1) \rho_-,
\end{equation}
\jrs{where $\sigma$ is a mixed stabilizer state, but $\rho_-$ can be any density operator. Moreover, since systematic error is unavoidable, it is unnecessary to evaluate the first term to high precision, so the runtime can be reduced. Pseudocode for our constrained path simulator is given in Algorithm \ref{algo:constrainedpath}. We will then place tight bounds on the systematic error. }

\begin{algorithm}[H]
\caption{constrained path simulator\label{algo:constrainedpath}}
\begin{algorithmic}[1]
\algrenewcommand\algorithmicrequire{\textbf{Input:}}
\algrenewcommand\algorithmicensure{\textbf{Output:}}
\Require \jrs{Target state} $\rho$; real numbers $\lambda, c,p_{\mathrm{fail}} > 0 $,  and stabilizer state $\sigma \in \bar{\mathrm{S}}_{n}$ s.t. $\rho \leq \lambda \sigma$ and $c, p_{\mathrm{fail}} \ll 1$; Pauli observable $E$ and simulable channel $O\in \mathcal{O}_n$.
\Ensure Estimate $\widehat{E}$ and error bound $\Delta$, s.t $|\widehat{E} - \mathrm{Tr}(E O[\rho])| \leq \Delta$.
\Statex
\State \jrs{$\epsilon \gets c\lambda$}

\State Let $E_\sigma$ be an estimate for $ \lambda \mathrm{Tr}(E O[ \sigma ]) $ obtained via sampling up to error $\epsilon$ and success probability $1- p_{\mathrm{fail}}$.
\State \jrs{ $E_\mathrm{max} \gets \min\{1, E_\sigma + \epsilon + \lambda - 1\}$.}
\State \jrs{$E_\mathrm{min} \gets \max\{-1, E_\sigma - \epsilon - \lambda + 1\}$.}
\State $\widehat{E} \gets (E_\mathrm{max} + E_\mathrm{min})/2$
\State $\Delta \gets (E_\mathrm{max} - E_\mathrm{min})/2 $
\end{algorithmic}
\end{algorithm}

Choosing $E_\mathrm{max}$ and $E_\mathrm{min}$ to be given in steps 3 and 4, we ensure that for all $\lambda$ and $E_\sigma$,
\begin{equation}
    |\widehat{E} - \mathrm{Tr}(E O[\rho])|\leq \Delta
\end{equation}
holds with probability $1-p_{\mathrm{fail}}$. The major caveat is that there are certain regimes (for large $\lambda$ and small $E_\sigma$) where the algorithm fails by trivially estimating the true expectation value to be anywhere in the range $[-1,1]$. Nevertheless, in some regimes we efficiently obtain a biased but non-trivial estimate. We first briefly explain steps 3 and 4, before analysing the error bound and runtime.

When $\lambda$ and $\sigma$ are such that $\rho \leq \lambda\sigma$, 
using \eqref{eq:LambdaPlusDefn3}, 
there is some density matrix $\rho_-$ such that $\rho$ can be written as 
\begin{equation}
    \rho = \lambda \sigma - (\lambda - 1) \rho_-.
\end{equation}
Step 2 estimates $E_\sigma$ such that \jrs{$|E_\sigma - \lambda \Tr(EO[\sigma])| \leq \epsilon$ }with probability $1-p_{\mathrm{fail}}$. We use this to bound possible values of $\Tr(E O[\rho])$:
\begin{align}
    \Tr(E O[\rho]) & = \lambda\Tr(E O[\sigma]) - (\lambda - 1) \Tr(E O[\rho_-]) \\
                   & \leq E_\sigma + \epsilon + (\lambda - 1)
\end{align}
Similarly one obtains \jrs{$\Tr(E O[\rho])\geq  E_\sigma - \epsilon - (\lambda - 1) $.} Trivially we know that $|\Tr(E O[\rho])| \leq 1$, so in case either expression exceeds this (for example if $E_\sigma$ is close to $\pm 1$) we simply take either $E_\mathrm{max} = 1$ or $E_\mathrm{min} = -1$ as necessary. We now consider the regimes where the bounds are trivial, and give the size of the error otherwise. 

\noindent\textbf{Case 1 (failure):} Trivial bounds are obtained when both these conditions hold:
\begin{align}
    E_\sigma + \epsilon + (\lambda - 1)  & \geq 1  \label{eq:failurebound1}\\
    E_\sigma - \epsilon - (\lambda - 1) & \leq -1,  \label{eq:failurebound2}
\end{align}
that is, when $E_\sigma$ satisfies:
\begin{equation}
    2- \lambda(1+c)\leq E_\sigma \leq \lambda (1+c) - 2. 
\end{equation}
This holds only if $\lambda \geq 2/(1 + c) \approx 2$, as otherwise at most one of the inequalities \eqref{eq:failurebound1} and \eqref{eq:failurebound2} can be true.

\noindent\textbf{Case 2 (constant error):} When  $\lambda < 2/(1 + c) \approx 2$, there is a range of values of $E_\sigma$ where inequalities \eqref{eq:failurebound1} and \eqref{eq:failurebound2} are both violated:
\begin{equation}
     \lambda (1+c) - 2\leq E_\sigma \leq 2- \lambda(1+c). 
\end{equation}
In this case, we have
\begin{align}
    \widehat{E} & = E_\sigma, \\
    \Delta & = \lambda(1 + c) - 1.
\end{align}

\noindent\textbf{Case 3 (error decreases with $\abs{E_\sigma}$):} The remaining case occurs when \jrs{$|E_\sigma|$ is sufficiently large}, so that either $E_\mathrm{max} = 1 $ or $E_\mathrm{min} = -1 $. This limits the range of possible values of $\Tr(E O[\rho])$, so that
\begin{equation}
    \Delta = \frac{\lambda(1 + c) - |E_\sigma|}{2}.
\end{equation}
This occurs when exactly one of the inequalities \eqref{eq:failurebound1} and \eqref{eq:failurebound2} is satisfied, while the other is violated.  \jrs{Note that this can happen even when $\lambda \gg 2 $, as it depends on the value of $E_\sigma$ returned. For example if $E_\sigma = \pm\lambda$, we obtain $\widehat{E} = 1 \mp \epsilon/2$ and $\Delta = \epsilon/2$.}

\jrs{Estimating $\lambda \mathrm{Tr}(E O[ \sigma ]) $ using any Clifford simulator (e.g. the dyadic frame simulator)} takes up the most time in the algorithm, as the other steps are trivial to evaluate. 	
Since $\sigma$ is a \jrs{convex combination of stabilizer states, there is no additional sampling overhead due to negativity. The prefactor $\lambda$ increases the variance of the estimator, but we compensate for this by setting the precision to $\epsilon = c \lambda$, where $c$ is a small constant. The rationale for this is that the systematic error due to our ignorance of $\rho_-$ is unavoidable, and this error is of size $\lambda - 1$. Therefore there is a limit to the precision we can achieve by increasing the runtime of the sampling step, and we should set the precision commensurate with the size of $\lambda$. Using the standard arguments (see Section \ref{sec:quasi}), the smallest number of samples $T$ sufficient to achieve this precision is:}
\begin{equation}
    T = \lceil  2 \lambda^2 \epsilon^{-2} \log (2 p_{\mathrm{fail}}^{-1}) \rceil = \lceil  2 c^{-2} \log (2 p_{\mathrm{fail}}^{-1}) \rceil.
\end{equation}
\jrs{The runtime for our constrained path simulator is therefore constant with respect to $\lambda$ (i.e. the generalized robustness $\Lambda^{+}(\rho)$ when the decomposition is optimal), depending only on the parameters $c$ and $p_{\mathrm{fail}}$. In this sense, we achieve efficient runtime by trading off against precision in the estimate; it is the error $\Delta$ which scales with the magic monotone rather than the runtime.}

Our explicit algorithm for estimating Pauli expectation values easily adapts to estimate Born rule probabilities for stabilizer projectors $\Pi$ by replacing the assumption $|\Tr(E \rho )|\leq 1$ for any $\rho$ with $0\leq \Tr(\Pi \rho) \leq 1$.


\section{Applications to other resources}
\label{sec:otherTheories}

\prx{%
Although we focused on the simulation of quantum circuits within the stabilizer formalism in Sec.\ref{sec:simulate}, our methods can be extended beyond magic-state quantum computation. The crucial idea here is to identify an efficiently simulable quantum sub-theory, consisting of a set of $n$-qubit pure-state vectors $\mathrm{S}_n$ and a set of operators $\mathcal{T}_n$ such that any $K \in \mathcal{T}_n$ acts on a state $\ket\phi \in \mathrm{S}_n$ as $K \ket\phi \propto \ket{\phi'} \in \mathrm{S}_n$, and this update can be efficiently tracked. Any such sub-theory can then be extended to mixed states and to the dyadic setting, allowing for the adaptation of our classical simulators. In particular, we can show that --- as long as the sub-theory itself satisfies some basic criteria regarding its simulability --- we can always efficiently simulate quantum circuits built from operators in $\mathcal{T}_n$ when acting on states composed of convex mixtures of projectors in $\mathrm{S}_n$. Just as before, when the algorithms work outside the given quantum sub-theory, they incur an additional resource cost which can be measured using our monotones.

This formalism very naturally fits into the framework of quantum resource theories~\cite{chitambar_2019}, which study the quantification and manipulation of resources in physically restricted settings. Here, a set of states and a set of operations are considered ``free'', while states and operations outside of these sets are costly to use and implement. The connection we build between classical simulators and resource theories then connects the quantitative value of such resources with the performance of the classical simulators, thus giving an explicit operational meaning to important resource monotones. Indeed, the monotones $\mathcal{R}$, $\Lambda$, and $\Lambda^+$ can be defined in general resource theories~\cite{regula2017convex}, but their operational meaning is not always known. For instance, although the robustness monotones $\mathcal{R}$ and $\Lambda^+$ have found general use in tasks such as channel discrimination~\cite{takagi_2019-2,takagi_2019} and resource conversion~\cite{liu_2019,regula_2019-1}, the dyadic negativity $\Lambda$ has not been shown to have any direct operational applications in general resource theories, nor has a connection between monotones such as the generalized robustness $\Lambda^+$ and classical simulation been established.

In the following, we will refer to the pure states $\mathrm{S}_n$ as free. Similarly, we define the set of free operations $\mathcal{O}_n$ as all quantum channels whose Kraus operators belong to $\mathcal{T}_n$, and thus cannot generate any resource from a free state. Analogously, the set of free observables $\mathcal{M}_n$ can be defined to be all observables which always result in a free post-measurement state. 

From our discussion of the dyadic frame simulator in Sec.~\ref{sec:simulate1} and Appendix~\ref{app:dyadicframe}, the proof clearly requires only three crucial assumptions about the classical simulability of the underlying sub-theory. We formalize them as follows:
\begin{enumerate}[({S}1)]
  \item Only $\mathcal{O}(\mathrm{poly}(n))$ bits of information are necessary to index all $n$-qubit pure free states in the set $\mathrm{S}_n$.
  \item Given a free operator $K \in \mathcal{T}_n$ and any free state $\ket{\phi} \in \mathrm{S}_n$, we can compute the update $K \ket{\phi}$ as well as the norm $\| K \ket{\phi} \|$ in $\mathcal{O}(\mathrm{poly}(n))$ time.
  \item Given a free observable $\Pi \in \mathcal{M}_n$ and any free states $\ket{L}, \ket{R} \in \mathrm{S}_n$, we can compute $\bra{R}\Pi\ket{L}$ in $\mathcal{O}(\mathrm{poly}(n))$ time.
\end{enumerate}
As before, we will be interested in the composition of free operations $O \in \mathcal{O}_n$ which admit a simulable decomposition, i.e. can be written as $O(\cdot) = \sum_{i=1}^{N_K} K_i \cdot K_i^\dagger$ with each $K_i \in \mathcal{T}_n$ and $N_K \leq \mathrm{poly}(n)$. With this, our proofs of Sec.~\ref{sec:simulate1} and Appendix \ref{app:dyadicframe} can be immediately applied to generalize the dyadic frame simulator.

\begingroup
\renewcommand{\thetheorem}{\ref{thm:simulationJS}'}
\begin{theorem}
Consider a resource theory with free pure states $\mathrm{S}_n$, free operations $\mathcal{O}_n$, and free observables $\mathcal{M}_n$ satisfying criteria (S1)-(S3) above.

Let  $\rho = \sum_j \alpha_j \ket{L_j}\bra{R_j}$, be a known dyadic decomposition of an initial $n$-qubit state,  where $\alpha\in \mathbb{C}$ and the probability distribution $\{|\alpha_j|/\onenorm{\alpha}\}$ can be efficiently sampled. Let $\chan = O^{(T)} \circ \ldots \circ O^{(1)}$, where each $O^{(t)} \in \mathcal{O}_n$.
    Suppose that every $O^{(t)}$ has a known simulable decomposition. Then, given $\Pi \in \mathcal{M}_n$, we can estimate the Born rule probability $\mu = \Tr(  \Pi \chan [\rho ] )$ within additive error $\epsilon$, with success probability at least $1- p_{\mathrm{fail}}$ and  worst-case runtime
  \begin{equation}
 \frac{\|\alpha\|_1^2}{\epsilon^2}  \log(p_{\mathrm{fail}}^{-1})T
{\rm poly}(n).
  \end{equation}  
  Furthermore, if the dyadic decomposition of $\rho$ is optimal then $\|\alpha\|_1$ can be replaced by $\Lambda(\rho)$.
\end{theorem}
\endgroup


Theorem \ref{thm:simulationJS}' establishes an efficient simulation algorithm which can be employed in any resource theory that satisfies the requirements.
The theorem also connects the monotone $\Lambda$ with the sampling overhead of the algorithm, thus endowing $\Lambda$ with an exact operational interpretation in the context of resource theories beyond magic. The result of Theorem \ref{thm:simulationJS}' additionally allows us to employ the constrained path simulator of Sec.~\ref{sec:simulate3} to define a related simulation algorithm which depends on another monotone --- the generalized robustness $\Lambda^+$. Once again, the reasoning of Sec.~\ref{sec:simulate3} can be applied verbatim under the assumptions (S1)-(S3).
}

As an example where the result can be immediately applied, consider the resource theory of quantum coherence~\cite{baumgratz_2014,streltsov_2017}, where the free states $\mathrm{S}_n$ are the vectors of the computational basis $\{\ket{i}\}$. The free measurements are in the computational basis, for instance projectors of the form $\Pi=\kb{{\bf x}}{{\bf x}} \otimes \id$ where ${\bf x}$ is a fixed bit-string, which can be efficiently computed.  The corresponding dyadic negativity $\Lambda$ is then the (element-wise) $\ell_1$-norm, $\|\rho\|_{\ell_1} = \sum_{i,j} \left| \langle i \vert \rho \vert j \rangle \right|$.  We remark that $\|\cdot \|_{\ell_1}$ is trivially a multiplicative monotone in any dimension. Although $\|\cdot \|_{\ell_1}$ is one of the most commonly employed measures in the resource theory of coherence~\cite{baumgratz_2014,rana_2017,streltsov_2017}, it has lacked an explicit operational interpretation thus far.
Since the resource theory of coherence is not known to admit a unique, physically motivated choice of free operations~\cite{chitambar_2016,streltsov_2017}, 
we briefly discuss the possible choices of $\mathcal{O}_n$ and their classical simulability. From this, we use Thm.~\ref{thm:simulationJS}' to give $\|\cdot \|_{\ell_1}$ an operational interpretation.

 The most fundamental class of free operations within the resource theory of coherence are the \emph{incoherent operations} (IOs)~\cite{baumgratz_2014}, defined to be maps which admit a decomposition into Kraus operators which preserve the set of incoherent states. Such Kraus operators can be expressed as~\cite{winter_2016} $K=\sum_{x \in S} c_x \kb{f(x)}{x}$ for some set of bit-strings $S$, coefficients $c_x$, and an arbitrary function $f$. Given such a $K$ acting on no more than $b$ qubits, where $b$ is constant, we can efficiently compute $(K \otimes \id_{n-b}) \ket{i}$. Since any Boolean $f$ can be implemented by composing a set of universal classical logic gates (with $b=2$) such as AND and XOR, such gates can generate an IO realising any Boolean function.  Furthermore, IOs can simulate any quantum channel with sufficiently many coherent states~\cite{baumgratz_2014,bendana_2017}.
 
 One family of useful IOs in practice are the \emph{strictly incoherent operations} (SIOs)~\cite{winter_2016,yadin_2016}, which can be efficiently implemented by quantum circuits using only incoherent ancillae~\cite{yadin_2016}. As a subtheory of IO, all the updates are still efficiently computable.   Furthermore, $b=3$ suffices to provide the Toffoli gate which is universal for classical reversible logic, and so can generate any Kraus operator of the required form.   The biggest difference between SIOs and IOs is that, while SIOs are better understood from the perspective of their practical implementation, they cannot be promoted to universal quantum operations through the use of ancillary resource states~\cite{bendana_2017}.

We conclude that the resource theory of coherence that uses either IOs or SIOs as free operations satisfies the conditions of Thm.~\ref{thm:simulationJS}'. Thus, the theorem endows the $\ell_1$-norm of coherence with an operational interpretation as the sampling overhead in \br{the simulation of either of the classes of operations} using the dyadic frame simulator, \br{and similarly the constrained path simulator gives another meaning to the robustness of coherence $\Lambda^+$~\cite{napoli_2016}.}

Our dyadic frame simulator is especially useful in resource theories where other simulation algorithms such as the Howard--Campbell simulator for magic states~\cite{Howard17robustness} cannot be readily adapted.
For instance, in the resource theory of coherence, the free states $\bar{\mathrm{S}}_{n}$ form a zero-measure subset of all states, which means that no resourceful state $\rho$ can be decomposed as $\rho= \sum_j p_j \kb{\phi_j}{\phi_j}$  with $\ket{\phi_j}  \in \mathrm{S}_n$ and so the corresponding robustness quantifier $\mathcal{R}(\rho)$ diverges.

Note, however, that the dyadic frame simulator does not work for all resource theories. While the dyadic frames for stabilizer and incoherent operations meet the conditions of Thm.~\ref{thm:simulationJS}, the requirements cannot hold for the theory of separable states under local operations and classical communication (LOCC). This is because the free states consist of an infinite number of inequivalent pure product states, which cannot be described using $\mathrm{poly}(n)$ bits. However, one can accurately compute local unitaries acting on product states efficiently. Indeed, our framework could encompass entanglement and similar theories using a suitable $\epsilon$-net over the set of separable states, and we leave the precise statement of the relevant conditions for future work.

\section{Conclusions}

We have introduced three resource monotones into the setting of magic state quantum computation: the dyadic negativity $\Lambda$, the generalized robustness $\Lambda^{+}$, and the \jrs{mixed-state extent} $\Xi$.  The first part of the paper focuses on resource-theoretic results, including that: (\textit{i}) for pure states, the monotones all equal the extent monotone $\xi$; (\textit{ii}) for tensor products of single-qubit mixed states, they all coincide; and (\textit{iii}) the monotones act multiplicatively on tensor products of single-qubit mixed states. The results significantly simplify the computation of the monotones for multiple copies of a single-qubit state, and allow us to completely understand the asymptotic behavior of our magic quantifiers, which contrasts with previously used monotones.  Furthermore, our magic monotones often tighten previously known bounds on distillation rates.   

For each monotone, we introduce a related classical simulation algorithm. Our dyadic negativity simulator has a runtime proportional to $\Lambda(\rho)^2$, which is similar to \br{--- but significantly faster than ---} the Howard--Campbell simulator with runtime $\mathcal{R}(\rho)^2$ where $\mathcal{R}$ is the robustness of magic. Additionally, we show that the dyadic negativity simulator works for circuits which use completely stabilizer-preserving operations. This class includes all the conventional stabilizer operations (Clifford unitaries, Pauli projections, etc.) and we believe it is likely to be strictly larger. If true, the situation would mirror entanglement theory in the separation between LOCC and separable operations~\cite{LOCC_SEP}. We found that for tensor products of $n$ single-qubit states, both $\Lambda$ and $\mathcal{R}$ scale exponentially with $n$, but $\mathcal{R}$ is exponentially larger than $\Lambda$. This establishes our dyadic negativity simulator as the fastest known quasiprobability simulator for qubit magic states.

However, not all classical simulation algorithms are based on quasiprobability distributions, with the stabilizer rank methods representing a distinct paradigm. There are several crucial differences, including that stabilizer rank methods enable a stronger notion of classical simulation, as they allow us to sample outputs of a quantum computation, not just estimate Born rule probabilities.  Prior work on stabilizer rank simulations considered only pure states, but our simulator extends this to mixed states and demonstrates an expected runtime proportional to $\Xi(\rho)$.  Note the linear dependence on $\Xi$ (largely due to fast norm estimation~\cite{bravyi2016improved}), in contrast to the quadratic dependence encountered with quasiprobability simulators.  Since in general $\Lambda[\rho] \leq \Xi[\rho]$, it is theoretically possible that $\Lambda[\rho] \ll \Xi[\rho]$ so that $\Lambda[\rho]^2 \leq \Xi[\rho]$, which would mean a runtime advantage for the quasiprobability methods.  However, for products of single-qubit states the monontones are equal, so for such states our resource theory results show that the advantage clearly falls to the stabilizer rank simulators.  Furthermore, we improve stabilizer rank bounds, with the runtime for sampling Clifford magic states (e.g. $T$ states) improved to $\mathcal{O}(1/\delta^3)$ from the prior $\mathcal{O}(1/\delta^4)$ bound
where $\delta$ is the sampling precision. For other magic states, the advantage is not as simple to describe using big-$\mathcal O$ notation, but Fig.~\ref{fig:DeltaScaling} shows it to be considerable in practice.

Finally, by ensuring that our simulation algorithms can be easily generalized and providing a recipe to adapt the simulators to  resource theories beyond magic states, 
we shed light on the simulation of quantum circuits using very general resources under suitable assumptions. 
This not only provides new insight into the practical uses of resource quantifiers in well-studied theories such as quantum coherence, but also opens an avenue for a further study of the connections between the theoretical frameworks of quantum resources and their operational applications in quantum computation.

A clear direction for further research is to extend our results to the channel picture, which would enable a more direct route to simulate circuits with no need to replace non-free operations with state injection gadgets.  This is especially important in the context of stabilizer theory for the simulation of circuits with gates outside the Clifford hierarchy, as the gadgets then become more complex.

\begin{acknowledgments}
We would like to thank the anonymous reviewers
for their insightful comments. 
This work was supported by the Engineering and Physical Sciences Research Council [grant numbers EP/P510270/1 (J.R.S.) and EP/M024261/1 (E.T.C.\ and Y.O.)]. B.R. was supported by the Presidential Postdoctoral Fellowship from Nanyang Technological University, Singapore. Research at Perimeter Institute is supported in part by the Government of Canada through the Department of Innovation, Science and Economic Development Canada and by the Province of Ontario through the Ministry of Colleges and Universities. H.P.\ also acknowledges the support of the Natural Sciences and Engineering Research Council of Canada (NSERC) discovery grants [RGPIN-2019-04198] and [RGPIN-2018-05188] and the ARC via the Centre of Excellence in Engineered Quantum Systems (EQUS) Project No. CE170100009. This work was completed
while E.T.C.\ was at the University of Sheffield.
\end{acknowledgments}

\bibliographystyle{apsrmp4-2}
 \bibliography{ms}
 
\appendix

\section{Alternative proofs for previous results}

\subsection{Monotone equivalence proof}
\label{RegulaProof}
 Here we prove Lem.~\ref{RegulaLemma1}.  Consider an optimal decomposition for the extent, such that
 \begin{equation}
	 \ket{\Psi} = \sum_i c_i \ket{\psi_i} , \; \ket{\psi_i} \in \mathrm{S}_n 
 \end{equation}
 with $\xi(\Psi)=\|c\|_1^2$. Then 
  \begin{equation}
 \kb{\Psi}{\Psi} = \sum_{i,j} c_i c_j^* \kb{\psi_i}{\psi_j}  
 \end{equation}
 is a valid decomposition into the dyadic frame leading to
  \begin{align}
      \Lambda(  \kb{\Psi}{\Psi}  ) & \leq \sum_{i,j} | c_i c_j | \\   \nonumber
      & =  (\sum_i  | c_i| )( \sum_j |c_j | ) \\  \nonumber
      & = \|c\|_1^2 = \xi(\Psi) .
 \end{align}
 Next, we prove the converse inequality.  The dual convex problem to the minimization of $\xi$ is 
 \begin{equation}
		 \xi ( \Psi ) = \mathrm{max}_\omega \{  |\bk{\omega}{\Psi}|^2   :     \mbox{ s.t. } \forall \phi \in  \mathrm{S}_n ,  |\bk{\omega}{\phi}|^{2} \leq 1 \}.
 \end{equation}
 Note that the $\omega$ need not be properly normalized.   Let us label $\omega^\star$ as a vector achieving this  maximum so that  $\xi ( \Psi ) =  |\bk{\omega^\star}{\Psi}|^2$.   We further recall that $\Lambda$ also has a dual formulation 
  \begin{align}
  \label{DualFormulation}
	 \Lambda ( \rho) = \mathrm{max}_W \{  \mathrm{Tr}[ W \rho ]  :  \mbox{ s.t. } &    \forall \,\ket{\phi}, \ket\psi \in \mathrm{S}_n ,\notag\\
	 &|\bra{\phi}W \ket{\psi}| \leq 1 \} .
 \end{align}
 In particular for feasible $W$ we have  $\Lambda ( \rho) \geq  \mathrm{Tr}[ W \rho ] $.  We notice that the extent witness $\ket{\omega^\star}$ can be used to build an operator $W=\kb{\omega^\star}{\omega^\star}$ that is a valid witness for $\Lambda$.  Therefore,
  \begin{align}
\Lambda ( \kb{\Psi}{\Psi}) & \geq \mathrm{Tr}[\kb{\omega^\star}{\omega^\star} \Psi \rangle \langle \Psi \vert ] \\ \nonumber
& \geq |\bra{\omega^\star} \Psi \rangle |^2 = \xi ( \Psi).
\end{align}
Having proved both directions, we conclude an equality.  This proves Lem.~\ref{RegulaLemma1}.  Since the witness $W$ was a positive operator, an identical proof also shows that $\Lambda^+ ( \kb{\Psi}{\Psi})= \xi ( \kb{\Psi}{\Psi} )$.  For the $\Xi$ monotone, there is only one convex decomposition of $\kb{\Psi}{\Psi}$. Hence $\Xi ( \kb{\Psi}{\Psi})= \xi ( \kb{\Psi}{\Psi} )$.

\subsection{Sandwich theorem}

Here we present a proof of Thm.~\ref{Sandwich1}.  Recall that $\Lambda$ is the result of maximizing over all $W$-witnesses, whereas $\Lambda^+$ is limited to all $W^+$-witnesses, which immediately leads to $\Lambda^+( \rho )	\leq \Lambda( \rho )$.  To show $\Lambda( \rho ) \leq \Xi ( \rho )$, one simply takes the optimal decomposition w.r.t to $\Xi$, as follows
\begin{align}
    \rho & = \sum_j p_j \kb{\Psi_j }{ \Psi_j } , \\ \nonumber
	\Xi[\rho] & =  \sum_j p_j \xi( \Psi_j) .
\end{align}	
Next, we insert this decomposition into $\Lambda$ and use convexity
\begin{equation}
	\Lambda[ \rho ] \leq \sum_j p_j \Lambda( \kb{\Psi_j}{\Psi_j})  .
\end{equation}	
Using Lem.~\ref{RegulaLemma1} we have
\begin{equation}
	\Lambda[ \rho ] \leq \sum_j p_j \xi( \Psi_j)  = \Xi[\rho] ,
\end{equation}	
which completes the proof of Thm.~\ref{Sandwich1}.

\section{Geometry of $\Phi^{\pm}_{\rho}$}
\label{AppGeometry}

Here we show that the states $\Phi^{\pm}_{\rho}$ introduced in Eq.~\eqref{PhiDefined} are contained in the set $P_Y$ introduced in Def.~\ref{OctantDefs}.  Note that Eq.~\eqref{PhiDefined} is defined in terms of $\rho \in P_{Y}$. 

The result $\Phi^{\pm}_{\rho} \in P_Y$ was used in the proof of Lem.~\ref{1qubitLemma}.  In that proof, we appealed to geometry presented in Fig.~\ref{SetsSlice} and here we instead provide an algebraic argument.

Using the definition of the set $P_Y$, to prove that $\Phi^{\pm}_{\rho} \in P_Y$,
it suffices to show that 
\begin{align}
   \langle \Phi_{\rho}^\pm \vert Y \vert \Phi_{  \rho}^\pm \rangle & \geq  0, \\ \nonumber
    \langle \Phi_{\rho}^\pm \vert (X-Y)\vert \Phi_{\rho}^\pm \rangle &  \geq     0 , \\ \nonumber
    \langle \Phi_{\rho}^\pm \vert (Z-Y) \vert \Phi_{\rho}^\pm \rangle &  \geq    0 . 
\end{align}
Now note that $\langle \Phi_{\rho}^\pm \vert Y \vert \Phi_{\rho}^\pm \rangle = \mathrm{Tr}[Y \rho]$. 
Minimizing $\mathrm{Tr}[Y \rho]$ over all feasible Bloch vectors $(r_A,r_B,r_F)$ in the decomposition of $\rho$ given by \eqref{WonkyDecomp}, we find that $\mathrm{Tr}[Y \rho] \ge 0$, which proves the first inequality.

Next, we tackle the second inequality (with the third inequality following in a similar fashion).  One computes that 
\begin{align}
   2 \langle \Phi_{\rho}^\pm \vert (X-Y) \vert \Phi_{\rho}^\pm \rangle & =   \langle \Phi_{\rho}^\pm \vert (\sqrt{6}\sigma_A + \sqrt{2}\sigma_B) \vert \Phi_{\rho}^\pm \rangle  \\ \nonumber
    &= \sqrt{6} r_A \pm \sqrt{2}\sqrt{1- r_A^2-f^2},
\end{align}
which is positive whenever
\begin{align}
  \sqrt{6} r_A \geq \sqrt{2(  1- r_A^2 -f^2 )}
\end{align}
or more concisely
\begin{align}
\label{Required}
   r_A \geq  \sqrt{(1 -f^2)}/2 .
\end{align}
For mixtures of $\Psi_{f}^X$ and $\Psi_{f}^Z$, we find this holds with equality.  For mixtures in the convex hull of $\{\Psi_{f}^X, \Psi_{f}^Y, \Psi_{f}^Z \}$ we find $r_A \leq  \sqrt{(1 -f^2)}/2$. However, we are currently considering $\rho$ outside this set, just outside the facet spanned by $\Psi_{f}^X$ and $\Psi_{f}^Z$. Therefore, Eq.~\eqref{Required} indeed holds.

\prx{%
\section{Dyadic frame simulator technical details}
\label{app:dyadicframe}
\jrs{Here we prove Theorem \ref{thm:simulationJS} which assures the validity and runtime of our dyadic frame simulator. Recall that the goal of the dyadic frame simulator is to estimate the Born rule probability $\mu=\Tr[\Pi \chan(\rho)]$, where $\rho$ is an $n$-qubit mixed magic state, and $\chan = O^{(T)} \circ \ldots \circ O^{(1)}$ is a sequence of completely stabilizer-preserving channels $O^{(t)} \in \mathcal{O}_n$. Before proving the theorem, we first restate and discuss the restrictions we impose on $O^{(t)}$. Any completely stabilizer-preserving channel $O \in \mathcal{O}_n$ has at least one Kraus decomposition of the form}
\begin{equation}
    O = \sum_r^{N_U} p_r \mathcal{U}_r + \sum_s^{N_K} q_s \mathcal{K}_s,
    \label{eq:simdecomp}
\end{equation}
\jrs{where all $\mathcal{U}_r = U_r(\cdot)U_R^\dagger$ are unitary Clifford operations, and $\mathcal{K}_s = K_s(\cdot)K_s^\dagger$ correspond to completely stabilizer-preserving non-unitary Kraus operators \cite{Seddon19}. Let $P_U = \sum_r p_r$ be the total weight of the unitary part of the decomposition. We say that a channel decomposition is \emph{simulable} if the number of Kraus operators is bounded as $N_U, N_K \leq \poly(n)$. In Theorem \ref{thm:simulationJS} we assume that the channels $O^{(t)}$ provided as input to the algorithm all have a known simulable decomposition. We use the restriction on $N_U$ for simplicity in the proof, but provided one can efficiently sample from the distribution $\{p_r/P_U\}$ and compute any corresponding $U_r$, then this restriction can be removed. Note however that the restriction on $N_K$ cannot be similarly relaxed.} For concreteness, \jrs{we assume that non-unitary Kraus operators are} given as a length \jrs{$\poly(n)$} list with each entry being a pair. The pair encodes a stabilizer-preserving Kraus operator and its associated weight factor. The Clifford part of the decomposition takes the same format. We use $\mathcal{L}^{(t)}_U$ and $\mathcal{L}_K^{(t)}$ to denote the respective lists for the unitary and non-unitary part of the decomposition of $O^{(t)}$. Each stabilizer-preserving Kraus operator $K$ is described by giving an efficient description of the stabilizer state corresponding to the Choi state $\Phi_{K}$ as 
\begin{equation}
\Phi_{K} = (\mathcal{K} \otimes \id)\ket{\Phi}\bra{\Phi} \label{eq:choi state} ,
\end{equation}
where $\ket{\Phi} \propto \sum_j \ket{j}\ket{j}$.  Note that since $K$ acts by conjugation, $\Phi_K$ is a pure stabilizer state, so can be specified by $\mathcal{O}(n^2)$ classical bits.

\jrs{The class of simulable channel decompositions encompasses a wide range of practically important stabilizer operations. First, any convex combination of $n$-qubit Clifford gates $O = \sum_j p_j U_j$ is included, provided $\{p_j\}$ can be efficiently sampled from. Another subset of simulable channels are those  of the form $O=O'\otimes \id_{n-b}$, where $O'\in\mathcal{O}_b$ and $b$ is a small constant. Any $O'\in \mathcal{O}_b$ has a $2b$-qubit Choi state $\Phi_{O'}$ that lies inside the stabilizer polytope (i.e. it can be written as a convex combination of pure stabilizer states, each corresponding to a Kraus operator) \cite{Seddon19}. Although the number of stabilizer states grows super-exponentially with $b$,} the real vector space inhabited by $2b$-qubit density matrices is $(4^b -1)$-dimensional.  We can therefore completely partition the stabilizer polytope into simplices with $4^b$ vertices, where any mixed stabilizer state inhabits at least one simplex. Hence, by Carath\'eodory's theorem, $\Phi_O$ can be written as a convex combination of at most $4^b$ pure stabilizer states. Thus for families of circuits where $b$ has a fixed upper bound, \jrs{$N_U + N_K$ does not grow with $n$.}  This restriction is not too onerous, since practical quantum algorithms are typically synthesized in terms of one-, two- and three-qubit gates, and noise channels are often assumed to act locally. Moreover, we often already know the stabilizer decompositions of interesting channels. For instance, we can express $T$-gate injection gadgets and the single-qubit depolarizing channel with only two and four Kraus operators respectively.

\jrs{We now }present the algorithm and prove its validity. Our algorithm has two subroutines: (i) Algorithm \ref{algo:KrausUpdate}, which is an extended Gottesman-Knill-type subroutine that probabilistically updates an input stabilizer dyad given a set of Kraus operators; and (ii) Algorithm \ref{algo:dyadic}, which is an outer quasiprobability sampling routine that samples an initial dyad from the initial non-stabilizer state and propagates the dyad through the circuit, randomly selecting a single Kraus operator \jrs{from each decomposition $O^{(t)}$.}

\jrs{Note that in Algorithm \ref{algo:KrausUpdate}, we use the trace norm (i.e. the Schatten 1-norm, $\onenorm{A}=\Tr[\sqrt{A^\dagger A}]$), rather than the usual trace (as in the Born rule) to calculate the transition probabilities for propagating with a particular Kraus operator. While $\Tr(\Pi\rho) = \onenorm{\Pi\rho\Pi^\dagger}$ for physical states $\rho$, this does not hold for general dyads $\kb{L}{R}$. We illustrate that the Schatten 1-norm is the appropriate choice with a toy example.} 
\jrs{Consider the scenario where the penultimate dyad is $\sigma^{(T-1)} = \kb{\!+\!0}{-0}$, the final stabilizer channel $O^{(T)}$ is defined by Kraus operators $K_1 = \id \otimes \kbself{0}$ and $K_2 = U \otimes \kb{1}{1}$ for some Clifford $U$, and the final measurement operator to be evaluated is $\Pi = \kb{1}{1} \otimes \id$. Now, the channel $O^{(T)}$ leaves $\sigma^{(T-1)}$ unchanged, $O^{(T)}(\kb{\!+\!0}{-0})=\kb{\!+\!0}{-0}$. It is therefore clear that the correct contribution to the expectation value estimate (line \ref{step:EmEstimate} in Alg. \ref{algo:dyadic}) should be:}
\begin{align}
\mu_m &= \onenorm{\alpha} \re\{\Tr[\Pi \sigma^{(T)}] \}\\
& = \onenorm{\alpha} \re\{\Tr[(\kbself{1}\otimes \id )\kb{\!+\!0}{-0}]\} \\
            & = \onenorm{\alpha} \re\{\braket{-}{1}\braket{1}{+}\} = -\onenorm{\alpha}/2,
\end{align} 
\jrs{where we used cyclicity of the trace, and neglect the phase $e^{i\theta_r}$ for brevity. We need to ensure that the transition probabilities we compute (in line \ref{step:transitionProb} of Alg. \ref{algo:KrausUpdate}) produce statistics that converge to this contribution. Suppose we were to naively use the trace to compute transition probabilities, $P_{\mathrm{Tr},j} = \Tr[K_j \sigma^{(T-1)} K_j^\dagger]$. Then we would obtain:}
\begin{align}
P_{\mathrm{Tr},1} & = \Tr[(\id \otimes \kbself{0}) \kb{\!+\!0}{-0}(\id \otimes \kbself{0})]  \\
                  & = \bk{-}{+} = 0,\\
P_{\mathrm{Tr},2} & = \Tr[(U \otimes \kbself{1}) \kb{\!+\!0}{-0}(U^\dagger \otimes \kbself{1})] \\
                  & = \bra{-}U^\dagger U \ket{+} |\bk{1}{0}|^2 = 0.
\end{align}
\jrs{Here we have a problem, because both paths evaluate to zero, preventing $\mu_m$ from making any non-zero contribution to our estimate. By contrast, in our algorithm we use the Schatten 1-norm to compute transition probabilities:}
\begin{align}
P_1 & = \onenorm{(\id \otimes \kbself{0}) \kb{\!+\!0}{-0}(\id \otimes \kbself{0})} \\
    & = \onenorm{\kb{\!+\!0}{-0}} = 1 \\
P_2 & = \onenorm{(U \otimes \kbself{1}) \kb{\!+\!0}{-0}(U^\dagger \otimes \kbself{1})}\\
    & = |\bk{1}{0}|^2 \onenorm{(U\ket{+})\kb{1}{1}(\bra{-})U^\dagger } = 0.
\end{align}
\jrs{This method correctly tells us that we should select Kraus operator $K_1$ with certainty, resulting in the correct contribution $\mu_m =  -\onenorm{\alpha}/2$.} 

\jrs{Below we prove that this strategy leads to an unbiased estimator for $\mu$, where each individual sample is bounded as $|\mu_m|\leq \onenorm{\alpha}$. As per standard quasiprobability simulators (see Section \ref{sec:quasi}), to estimate an observable} within additive error of $\epsilon$ with success probability $p_{\mathrm{suc}} \geq 1-p_{\mathrm{fail}}$, we require at least $M$ samples from our algorithm \jrs{\cite{pashayan15,Howard17robustness}} , where
	\begin{equation}
M\geq 2 \frac{\|\alpha\|_1^2}{\epsilon^2}  \log\left(\frac{2} {p_{\mathrm{fail}}}\right).\label{eq: num_samples}
	\end{equation}

\begin{figure}[t]
\begin{algorithm}[H]
\caption{Stabilizer Kraus update subroutine \label{algo:KrausUpdate}}
\begin{algorithmic}[1]
\algrenewcommand\algorithmicrequire{\textbf{Input:}}
\algrenewcommand\algorithmicensure{\textbf{Output:}}
\Require Initial stabilizer dyad $\sigma = \kb{L}{R}$; length-$N$ list of pairs $ \mathcal{L} = \{(q_1, K_1),\ldots,(q_N, K_{N})\}$, where $q_j>0$ are weights and $K_j$ are stabilizer Kraus operators; normalization factor $P_X$.
\Ensure Updated dyad $\sigma' = \ket{L'}\bra{R'}$.
\Function{stabilizerUpdate}{$\sigma$,$\mathcal{L}$,$P_X$}
\For{$r \gets 1$ to $N$} \label{step:Krausdist1}
\State $q'_r \gets q_r / P_X$ \label{step:normalizeq}
\State $P_r \gets \onenorm{q'_r K_r \sigma K_r^\dagger}$ \label{step:transitionProb}
\EndFor
\State $P_0 \gets 1 - \sum_{r=1}^{N_K} P_r$ \label{step:Krausdist2}
\State Sample $s$ from $\{0,\ldots,N\}$ with probability $P_{s}$
\If{s = 0}
\State $ \sigma' \gets 0$
\Else
\State $\sigma' \gets \ket{L'}\bra{R'} = q'_{s} K_{s} \ket{L} \bra{R}K_{s}^\dagger / P_{s} $
\EndIf
\State
\Return $\sigma'$
\EndFunction
\end{algorithmic}
\end{algorithm}
\vspace*{-\baselineskip}
\end{figure}

\begin{figure}[t]
\begin{algorithm}[H]
\caption{Dyadic frame simulator\label{algo:dyadic}}
\begin{algorithmic}[1]
\algrenewcommand\algorithmicrequire{\textbf{Input:}}
\algrenewcommand\algorithmicensure{\textbf{Output:}}
\Require Initial state $\rho$ with known dyadic stabilizer decomposition $\rho = \sum_{j} \alpha_j\ket{L_j}\bra{R_j}$; number of samples $M$; list of channels $\{O^{(1)},\ldots,O^{(T)} \}$ of length $T$ where $O^{(t)} \in \mathcal{O}_n$ for all $t$\jrs{, and each is \emph{simulable} as described in this appendix, and described by a list of weighted Clifford gates $\mathcal{L}^{(t)}_U$ and non-unitary Kraus operators $\mathcal{L}^{(t)}_K$}; stabilizer projector $\Pi$.
\Ensure Estimate $\hat{\mu}$ for $\Tr[\Pi \chan(\rho)]$, where $\chan = O^{(T)} \circ \ldots \circ O^{(1)}$.
\Statex
\State Let $P_j = |\alpha_j|/\onenorm{\alpha}$ define a probability distribution.
\For{$m \gets 1$ to $M$}
\State Sample $r$ with probability $P_{r}$. \label{step:sampledyad}
\State $\sigma^{(0)} \gets \ket{L^{(0)}}\bra{R^{(0)}} = \ket{L_{r}}\bra{R_{r}}$, \quad $e^{i\theta_r} \gets \alpha_r/|\alpha_r|$.
\For{$t \gets 1$ to $T$}
\State $P_U \gets \sum_j p_j$, where $p_j$ are the weights in list $\mathcal{L}^{(t)}_U$.
\State $P_K \gets 1 - P_U$ 
\State With probability $P_U$, $X\gets{``U\text{''}}$, else $X\gets{``K\text{''}}$.
\State $\sigma^{(t)} \gets$ \Call{stabilizerUpdate}{$\sigma^{(t-1)}$,$\mathcal{L}^{(t)}_X$,$P_X$}
\If{$\sigma^{(t)} = 0$ }
\State $\sigma^{(T)} \gets 0$
\BREAK
\EndIf
\EndFor
\State  $\mu_m \gets \re\{\onenorm{\alpha}e^{i\theta_{r}}\Tr[\Pi\sigma^{(T)}]\}$. \label{step:EmEstimate}
\EndFor
\State \Return $\hat{\mu} \gets \sum_m \mu_m/M$ 
\end{algorithmic}
\end{algorithm}
\vspace*{-\baselineskip}
\end{figure}

To prove the validity of our algorithm we must: (i) explain how the stabilizer update $q_r K_r \sigma K_r^\dagger$ can be carried out efficiently, (ii) show that the values $P_r$ in steps \ref{step:Krausdist1}-\ref{step:Krausdist2} of Algorithm \ref{algo:KrausUpdate} form a proper probability distribution,  and (iii) show that $\hat{\mu}$ returned by Algorithm \ref{algo:dyadic} is an unbiased estimator for $\Tr[\Pi \chan(\rho)]$. The total runtime given in Theorem \ref{thm:simulationJS} \jrs{is the product of the number of samples $M$ and the runtime to compute each sample.}

\textbf{(i) Efficient stabilizer update with Kraus operators.} In Algorithm \ref{algo:KrausUpdate} we must \jrs{compute the trace norm  $\onenorm{q_r K_r \kb{L}{R} K_r^\dagger}$ for all $n$ pairs $(q_r,K_r)\in \mathcal{L}$, and then perform the update $\kb{L}{R} \to \kb{L'}{R'}$. Note that we track any accumulated phase through the update, but here we absorb this factor in $\kb{L'}{R'}$ for brevity. There are two cases. Either $\mathcal{L}$ is a list of unitary operators, or it is a list of non-unitary Kraus operators. In the unitary case, computation of the norm is trivial. Since the initial dyad $\sigma$ is normalized with respect to the trace norm, and the norm is unitarily invariant, we have $ \onenorm{p_r U_r \sigma U_r^\dagger} = p_r$. By assumption, ${p_r/\sum_r p_r}$ can be efficiently sampled from. In the non-unitary case, we must compute:}
\begin{equation}
    \onenorm{q'_r K_r \kb{L}{R}K_r^\dagger} = q'_r \| K_r \ket{L}\| \cdot \| K_r \ket{R}\|.
\end{equation}
\jrs{Note that unlike the trace, the trace norm does not depend on the overlap between $K_r \ket{L}$ and $K_r \ket{R}$, and their vector norms are calculated separately. To see how this is done, we note that a Kraus operator whose Choi state is a normalized pure stabilizer state can always be written in the form $K_r = 2^{h/2} U_r \Pi_r$, where $U_r$ is a Clifford gate and $\Pi_r$ is a stabilizer projector of rank $2^{n-h}$ \cite{GrossHeinrichComm}, for some $h\leq n$. Since $U_r$ leaves the norm invariant, for the purpose of computing transition probabilities, only the projector and normalization constant matter:}
\begin{equation}
     \onenorm{q'_r K_r \kb{L}{R}K_r^\dagger} = q'_r 2^h \norm{\Pi_r\ket{L}}\cdot \norm{\Pi_r \ket{R}}.
\end{equation}
\jrs{The projection of each pure state onto a stabilizer subspace can be computed using standard stabilizer simulation techniques \cite{aaronson04improved,bravyi2018simulation}, in time $\mathcal{O}(hn^2)$. In the non-unitary case, we must compute the norm for $2N_K$ projected stabilizer states, so the total runtime for computing all transition probabilities for a single step $t$ is $\mathcal{O}(h N_K n^2)$.  }

\jrs{Once all $N$ transition probabilities are computed, a single Kraus operator $K_s \propto U_s \Pi_s$ is randomly selected for the update, so we must compute $\kb{L'}{R'} \propto U_s \Pi_s \kb{L}{R} \Pi_s U_s^\dagger$. As discussed in Section \ref{sec:simulate1}, it is vital that we track any acquired phase throughout each update, as this will affect our final estimate when we average over all $M$ samples $\mu_m$. We can do this using } the phase-sensitive Clifford simulator described in Ref. \cite{bravyi2018simulation}\jrs{. There} it was shown that the update corresponding to the projection $(\id + Q)/2$, where $Q$ is a Pauli operator, can be carried out in $\mathcal{O}(n^2)$ steps. \jrs{A rank $2^{n-h}$ stabilizer projector can be decomposed as a product of $h$ Pauli projections, so the projective part of the update takes time $\mathcal{O}(hn^2)$. As for the Clifford update, any $n$-qubit Clifford operation can be written in canonical form comprised of $\mathcal{O}(n^2/\log(n))$ gates from the standard gate set $\{CNOT,H,S\}$ \cite{aaronson04improved}. Bravyi \textit{et al.} \cite{bravyi2018simulation} showed that their phase-sensitive Clifford simulator can perform $CNOT$ and $S$ updates in time $\mathcal{O}
(n)$, and $H$ in time $\mathcal{O}(n^2)$, so the Clifford update for $U_r$ can be completed in time $\mathcal{O}(n^4/\log(n))$. Since $h\leq n$ and for simulable decompositions $N_K\leq \poly(n)$, the time taken is $\poly(n)$.}

\jrs{Combining all steps, the total time for a single call to \textsc{StabilizerUpdate} is $\mathcal{O}(h (N_K  + 1) n^2)+ \mathcal{O}(n^4/\log(n))$. 
Since $h\leq n$ and for simulable decompositions $N_K\leq \poly(n)$, the call is completed in $\poly(n)$ time in the general case. We note that in the special case where we restrict each $O^{(t)}$ to act on at most $b$ qubits, for some fixed $b$, the runtime for a single call can be improved considerably. In that case, $N_K\leq 4^b $ (i.e. a constant with respect to $n$) and the runtime will be $\mathcal{O}( b (4^b + 1) n^2) + \mathcal{O}(b^2 n^2/\log(b))$. }


\textbf{(ii) Valid probability distribution.} From the definition of $P_0$ in step \ref{step:Krausdist2} of Algorithm \ref{algo:KrausUpdate}, it is clear that $\sum_{r=0}^{N} P_r = 1$ and $P_1,\dots, P_{N} \ge 0$. Hence, to show that $\{P_r\}$ is a probability distribution, it suffices to show that $P_0 \geq 0$. \jrs{This is trivially true for the unitary path, as $\{p_r/(\sum_{s=1}^{N_U} p_s)\}$ is clearly a properly normalized distribution, with $P_0 = 0$. It remains to show $P_0 \geq 0$ for the non-unitary case. }

It is given that \jrs{the channel $O^{(t)}$  is a CPTP map. Let $O_K(\cdot) = \sum_{r=1}^{N_K} q_r K_r (\cdot) K_r^\dagger$ denote the non-unitary part of the decomposition \eqref{eq:simdecomp}, and let $O_U(\cdot) = \sum_{r=1}^{N_U} p_r U_r (\cdot) U_r^\dagger$, so that $O^{(t)} = O_U + O_K$. Each Clifford is a unitary operator, so $\sum_{r=1}^{N_U} p_r U_r^\dagger U_r = P_U\id$, recalling that $P_U =\sum_{r=1}^{N_U} p_r$ and $P_K = 1 - P_U$. Since $O^{(t)}$ is CPTP, its Kraus representation must be complete:}
\begin{equation}
    \id = P_U \id + \sum_{r=1}^{N_K} (\sqrt{q_r}K_r)^\dagger (\sqrt{q_r}K_r).
\end{equation}
\jrs{It follows that $\sum_{r=1}^{N_K} (\sqrt{q_r}K_r)^\dagger (\sqrt{q_r}K_r) = P_K \id$, meaning that $O_K$ is a CPTP map up to normalization by $1/P_K$. This normalization is achieved by setting $q'_r = q_r/P_K$ in step \ref{step:normalizeq}. Then} for any pure state $\ket{\psi}$, we have
\begin{align}
   1\geq \Tr[O(\kb{\psi}{\psi})/P_K] 
   &= \sum_{r=1}^{N_K} \Tr[\sqrt{q'_r}K_r \kb{\psi}{\psi}K_r^\dagger \sqrt{q'_r}] \notag\\
    & = \sum_{r=1}^{N_K} \norm{\sqrt{q'_r} K_r \ket{\psi}}^2\text{.} \label{eq:normalizedQ}
\end{align}
Let $Q^{(\psi)}$ be the $N_K$-element real vector where the $r$-th entry is $Q^{(\psi)}_r =  \norm{\sqrt{q'_r} K_r \ket{\psi}}$. From Eq.~\eqref{eq:normalizedQ}, we have that $\norm{Q^{(\psi)}} \leq 1$. Then for any normalized dyad $\kb{L}{R}$ we can express the sum of $P_r$ for $r\geq 1$ as a dot product between $Q^{(L)}$ and $Q^{(R)}$: 
\begin{align}
    \sum_{r=1}^{N_K} P_{r} &= \sum_{r=1} \onenorm{q'_r K_r \kb{L}{R}K_r^\dagger}\\
               &= \sum_{r=1} \| \sqrt{q_r} K_r \ket{L}\| \cdot \| \sqrt{q_r} K_r \ket{R}\| \\
               &= \sum_{r=1} Q^{(L)}_r  Q^{(R)}_r \\
               & = Q^{(L)} \cdot Q^{(R)} \leq \norm{Q^{(L)}} \cdot \norm{Q^{(R)}} \leq 1,
\end{align}
where in the last line we used the Cauchy-Schwarz inequality to show that $\sum_{r\geq 1} P_{r} \leq 1$, as promised. We note that the strategy of using an `abort' outcome $P_0$ was deployed in the appendix of Ref. \cite{Rall2019} to simulate post-selective channels. In our case the fact that $P_r$ for $r\geq 1$ can sum to less than 1 instead arises from the non-Hermiticity of the initial dyad $\sigma$.

\textbf{(iii) Unbiased estimator.} Finally we show that the expected value of $\hat{\mu}$ in Algorithm \ref{algo:dyadic} is $\Tr[\Pi \chan(\rho)]$. First, let us recombine the unitary and non-unitary part of $O^{(t)}$ into a single Kraus representation $O^{(t)} = \sum_{a=1}^{N_{\mathrm{Tot}}} q^{(t)}_r K^{(t)}_r (\cdot) K_r^{(t)\dagger} $, where $N_{\mathrm{Tot}} = N_U + N_K$, and consider the probability of sampling the $r$-th pair $(q_r, K_r)$ at step $t$. By inspection of Algorithms \ref{algo:KrausUpdate} and \ref{algo:dyadic} we see that the probability of taking the unitary path and then selecting the $r$-th pair from $\mathcal{L}_U$ is:
\begin{equation}
    \Pr(``U",r) = P_U \cdot P_r = P_U \onenorm{p'_r U_r \sigma U^\dagger_r} = \onenorm{p_r U_r \sigma U^\dagger_r}.
\end{equation}
Similarly the probability of choosing the non-unitary path followed by the $r$-th element of $\mathcal{L}_K$ is $\Pr(``K",r) = P_K P_r = \onenorm{q_r K_r \sigma K^\dagger_r}$. Thus at step $t$, the probability of sampling any Kraus operator from the decomposition, whether unitary or non-unitary, is given by $\Pr(q_r,K_r) = \onenorm{q_rK_r \sigma K_r^\dagger}$, and we can drop the distinction between the two.

Now, let the $(T+1)$-element vector $\vec{r} = (r_0,r_1,\ldots,r_T)$ label a particular trajectory through the circuit, in the following sense. The first entry $r_0$ labels the initial dyad $\sigma_{\vec{r}}^{(0)} = \kb{L_{r_0}}{R_{r_0}}$ sampled in step \ref{step:sampledyad}. For $t\geq 1$, the entry $r_t$ gives the index of the Kraus operator chosen at the $t$-th circuit element and we write $\mathcal{K}_{\vec{r}}^{(t)}(\cdot)=K_{r_t}^{(t)}(\cdot)K_{r_t}^{(t)\dagger}$, and use $q_{\vec{r}}^{(t)}$ to denote the corresponding prefactor. Let $\sigma_{\vec{r}}^{(t)}$ denote the current dyad updated up to the $t$-th Kraus operator along the trajectory $\vec{r}$, so that we have the recursive relation $\sigma_{\vec{r}}^{(t)}=q_{\vec{r}}^{(t)} \mathcal{K}_{\vec{r}}^{(t)}(\sigma_{\vec{r}}^{(t-1)})/P_{\vec{r}}^{(t)}$, where $P_{\vec{r}}^{(t)}$ is the probability of obtaining the outcome corresponding to the map $\mathcal{K}_{\vec{r}}^{(t)}$. The probability $P_{\vec{r}}$ of choosing the trajectory $\vec{r}$ is given by 
$
    P_{\vec{r}} = \prod_{t=0}^T P_{\vec{r}}^{(t)},
$
where $P_{\vec{r}}^{(0)} = |\alpha_{r_0}|/\onenorm{\alpha}$ is the probability of sampling the initial dyad $\sigma_{\vec{r}}^{(0)}$. For $t\geq 1$,  $P_{\vec{r}}^{(t)} = \onenorm{q_{\vec{r}}^{(t)} \mathcal{K}_{\vec{r}}^{(t)}(\sigma_{\vec{r}}^{(t-1)})} $ is calculated in the $t$-th call to Algorithm \ref{algo:KrausUpdate}.
Then the final dyad $\sigma_{\vec{r}}^{(T)}$ that we obtain from sampling the trajectory $\vec{r}$ is 

\begin{equation}
    \sigma_{\vec{r}}^{(T)} = \frac{q_\vec{r}\mathcal{K}_\vec{r}(\sigma_{\vec{r}}^{(0)})}{P_{\vec{r}}/P_{\vec{r}}^{(0)}},\label{eq:finaldyadKrausExpansion}
\end{equation}
where $ \mathcal{K}_\vec{r}(\cdot) = \mathcal{K}_{\vec{r}}^{(T)} \circ \ldots\circ \mathcal{K}_{\vec{r}}^{(1)}(\cdot)$ and $q_{\vec{r}}=\prod_{t=1}^T q_{\vec{r}}^{(t)}$. This dyad is properly normalized according to the trace norm, but is only defined for those trajectories with non-zero probabilities \jrs{$P_{\vec{r}}^{(t)} > 0 $ for all $t$}. We write $\mathcal{R}$ to denote the set of all such non-zero probability trajectories.

Now, there are two mutually exclusive possibilities for a given iteration of Algorithm \ref{algo:dyadic}: either we pick $r_t > 0 $ at each circuit element, choose some $\vec{r} \in \mathcal{R}$, and thus obtain a normalized dyad $ \sigma_{\vec{r}}^{(T)}$, or at some step we choose $r_t = 0$, and the iteration terminates with $\sigma_{\vec{r}}^{(T)} = 0$. Since these are the only possible outcomes, the total probability of \jrs{terminating} must be $P_\mathrm{term} = 1 - \sum_{\vec{r}\in \mathcal{R}} P_{\vec{r}}$. We can now write down an explicit expression for the expectation value of the random variable $\mu_m$ in step \ref{step:EmEstimate}:
\begin{align}
    \langle \mu_m \rangle & = P_\mathrm{term}\cdot 0 + \sum_{\vec{r}\in \mathcal{R}} P_{\vec{r}} \re\{\onenorm{\alpha}e^{i\theta_{r_0}}\Tr[\Pi\sigma_{\vec{r}}^{(T)}]\}  \\
                        & = \sum_{\vec{r}\in \mathcal{R}} P_{\vec{r}}^{(0)} \re\{\onenorm{\alpha}e^{i\theta_{\vec{r}_0}}\Tr[\Pi q_{\vec{r}} \mathcal{K}_{\vec{r}}(\sigma_{\vec{r}}^{(0)})]\}  \label{eq:nonzerotraj},
\end{align}
where in the second line we have cancelled the factors $P_{\vec{r}}^{(t)}$ for $t\geq 1$ with those in the denominator of Eq.~\eqref{eq:finaldyadKrausExpansion}. The real vectors $\vec{r} \notin \mathcal{R}$ are never chosen when running the algorithm, since they correspond to paths where $P_{\vec{r}}^{(t)} = 0$ for some $t$, and hence $\mathcal{K}_{\vec{r}}^{(t)} (\sigma_{\vec{r}}^{(t-1)})= 0$. Since $\mathcal{K}_{\vec{r}}(\sigma_{\vec{r}}^{(0)}) = 0$ for all $\vec{r} \notin \mathcal{R} $, we can add these zero-probability trajectories to the summation ~\eqref{eq:nonzerotraj} without affecting the total. Thus
\begin{align}
    \langle \mu_m \rangle & = \sum_{\vec{r}} P_{\vec{r}}^{(0)} \re\{\onenorm{\alpha}e^{i\theta_{r_0}}\Tr[\Pi q_{\vec{r}} \mathcal{K}_{\vec{r}}(\sigma_{\vec{r}}^{(0)})]\} \\
    & = \sum_{r_0} P_{\vec{r}}^{(0)} \re\{\onenorm{\alpha}e^{i\theta_{r_0}}\Tr[\Pi\sum_{r_1,\ldots,r_T} q_{\vec{r}} \mathcal{K}_{\vec{r}}(\kb{L_{r_0}}{R_{r_0}})]\},\nonumber
\end{align}
where in the second line we have written $P_{\vec{r}}^{(0)}$ outside of the inner sum since this probability is independent of $r_t$ for $t\geq 1$. The inner expression sums over all Kraus trajectories, and by linearity we have
\begin{align}
    \sum_{r_1,\ldots,r_T} q_{\vec{r}} \mathcal{K}_{\vec{r}} &= \sum_{r_T} q_{\vec{r}}^{(T)}\mathcal{K}_{\vec{r}}^{(T)} \circ \ldots\circ \sum_{r_1} q_{\vec{r}}^{(1)}\mathcal{K}_{\vec{r}}^{(1)} \\
    & = O^{(T)} \circ \ldots \circ O^{(1)} = \chan.
\end{align}
Hence
\begin{align}
    \langle \mu_m \rangle 
    & = \sum_{r_0} P_{\vec{r}}^{(0)} \re\{\onenorm{\alpha}e^{i\theta_{r_0}}\Tr[\Pi\chan(\kb{L_{r_0}}{R_{r_0}})]\}\\
    & =  \re\{\Tr[\Pi\chan(\sum_{r_0}\alpha_{r_0} \kb{L_{r_0}}{R_{r_0}})]\}\\
    & = \Tr[\Pi\chan(\rho)],
\end{align}
where in the second line we used the definition $P_{\vec{r}}^{(0)}e^{i \theta_{r_0}} = \alpha_{r_0}/\onenorm{\alpha}$. Hence \jrs{we have proved} that $\langle \hat{\mu} \rangle = \Tr[\Pi \chan(\rho)]$, so $\hat{\mu}$ is an unbiased estimator, \jrs{with each individual sample satisfying $|\mu_m|\leq \onenorm{\alpha}$. We argued above this implies we need $2\|\alpha\|_1^2\epsilon^{-2}  \log(2p_{\mathrm{fail}}^{-1})$ samples ( Eq.~\eqref{eq: num_samples}). To generate each sample, we need to make $T$ calls to \textsc{StabilizerUpdate}, and we showed in part (i) that each call is computed in $\poly(n)$ time. Therefore the total runtime is $ \|\alpha\|_1^2\epsilon^{-2}  \log(p_{\mathrm{fail}}^{-1})T
{\rm poly}(n)$ , as stated in Theorem \ref{thm:simulationJS}.}

\section{\jrs{Trace norm error for BBCCGH sparsification}}
\label{app:BBCCGHtracenorm}

\jrs{As discussed in Section \ref{sec:BBCCGHrecap}, the BBCCGH sparsification lemma \cite[Lem.~6]{bravyi2018simulation} entails that, given a pure state with exact stabilizer decomposition $\ket{\psi} = \sum_j c_j \ket{\phi_j}$, one can randomly generate a $k$-term sparsification $\ket{\Omega}$, such that:}
\begin{equation}
  \mathbb{E}(  \| \ket{\psi}-\ket{\Omega} \|^2 ) \leq \frac{\|c\|_1^2 }{k} ,
\end{equation}
\jrs{where $\norm{\cdot}$ is the standard vector norm. In order to compare with our new sparsification result, which deals with density operators, we need to translate this in terms of the trace norm. Here we prove the following simple corollary to the BBCCGH sparsification lemma.}

\jrs{\begin{corollary}
Given a normalized state $\ket{\psi} = \sum_j c_j \ket{\phi_j}$, for any $k>0$, one can sample from a distribution of sparsified vectors $\ket{\Omega} = (\onenorm{c}/k)\sum_{\alpha=1}^k \ket{\omega_\alpha}$, where $\ket{\omega_\alpha}$ are stabilizer states, such that:
\begin{equation}
    \mathbb{E}(\onenorm{\op{\psi}-\op{\Omega}}) \leq 2 \frac{\onenorm{c}}{\sqrt{k}} + \frac{\onenorm{c}^2}{k}.
\end{equation}
\end{corollary}}
\begin{proof}
\jrs{Let $\ket{\Delta} = \ket{\psi}-\ket{\Omega}$. Then for any particular $\ket{\Omega}$ we have:}
\begin{align}
    \op{\psi} - \op{\Omega} & = \op{\psi} - (\op{\psi} + \op{\Delta} \nonumber\\
                            & \quad \quad \quad  \quad \quad - \kb{\Delta}{\psi} - \kb{\psi}{\Delta})  \\
                            & = \kb{\Delta}{\psi} + \kb{\psi}{\Delta} - \op{\Delta}.
\end{align}
\jrs{Using the triangle inequality:}
\begin{align}
    \onenorm{\op{\psi} - \op{\Omega}} & \leq 2 \onenorm{\kb{\Delta}{\psi}} + \onenorm{\op{\Delta}} \\
                                     & = 2 \norm{\ket{\Delta}} \cdot \norm{\ket{\psi}} + \norm{\ket{\Delta}}^2 \\
                                     & = 2 \norm{\ket{\Delta}}  + \norm{\ket{\Delta}}^2,
\end{align}
\jrs{where the last line follows because $\ket{\psi}$ is normalized. Since the above is true for any $\ket{\Omega}$ taken from the distribution, it follows that:}
\begin{equation}
    \mathbb{E}(\onenorm{\op{\psi} - \op{\Omega}}) \leq 2 \mathbb{E}(\norm{\ket{\Delta}})  + \mathbb{E}(\norm{\ket{\Delta}}^2). \label{eq:tracenormDiff}
\end{equation}
\jrs{For the second term, the BBCCGH sparsification lemma \cite[Lem.~6]{bravyi2018simulation} tells us that we have $\mathbb{E}(\norm{\ket{\Delta}}^2)\leq \|c\|_1^2 /k$.}

\jrs{This leaves the first term. From Jensen's inequality, for any random variable $X$, we have that $\mathbb{E}(X) \leq \sqrt{\mathbb{E}(X^2)}$. So:}
\begin{align}
    \mathbb{E}(\norm{\Delta}) & \leq \sqrt{\mathbb{E}(\norm{\Delta}^2)} \\
                                & \leq \frac{\onenorm{c}}{\sqrt{k}},
\end{align}
\jrs{where the second line again follows from Ref.\cite[Lem.~6]{bravyi2018simulation}. Substituting into the inequality \eqref{eq:tracenormDiff}, we obtain the result. }
\end{proof}
\section{\jrs{Post-selection, stabilizer fidelity and the sparsification tail bound in BBCCGH}}
\label{app:whypostselect}
\jrs{In this section we discuss technical difficulties that arise when applying the sparsification results of Bravyi \textit{et al.} \cite{bravyi2018simulation} in the context of the bit-string sampling algorithm. Recall that the main BBCCGH sparsification lemma \cite[Lem.~6]{bravyi2018simulation} only tells us that randomly chosen sparsifications $k$-term $\ket{\Omega}$ will be close to the target state $\ket{\psi} = \sum_j c_j \ket{\phi_j}$ \textit{on average}, specifically $\mathbb{E}(\norm{\ket{\psi}-\ket{\Omega}}^2) \leq \onenorm{c}^2/k$. In itself, this does not preclude the possibility of occasionally obtaining $\ket{\Omega}$ that are very poor estimates for $\ket{\psi}$. One can check numerically that this is a rare occurrence, but it is preferable to put rigorous bounds on the probability of obtaining such outliers. Bravyi \textit{et al.} addressed this with the sparsification tail bound \cite[Lem.~7]{bravyi2018simulation}. This states that if we set $k \geq \onenorm{c}^2/\delta^2$, the probability of obtaining $\ket{\Omega}$ close to $\ket{\psi}$ is lower bounded as follows.}
\begin{align}
    \Pr \Big\{ \norm{\ket{\psi}-\ket{\Omega}}^2 &\leq \bk{\Omega}{\Omega} - 1 + \delta^2 \Big\}   \nonumber\\
    & \quad \quad \geq 1 - 2 \exp\left(- \frac{\delta^2}{8 F(\psi)}\right) \label{eq:tailbound}
\end{align}
\jrs{where $F(\psi)$ is the stabilizer fidelity, defined $F(\psi) = \max_\phi \abs{\bk{\phi}{\psi}}^2$ where $\phi$ are stabilizer states. However, there are two subtleties involved in applying this result in practice.}

\jrs{First, note that the usefulness of the bound depends on the norm of $\ket{\Omega}$ being close to (or smaller than) 1. But in general $\braketself{\Omega}$ can be larger. In principle, it is possible for it to be as large $\onenorm{c}^2$, though this is rather unlikely. In any case, this can be solved by using a post-selection step where we estimate $\bkself{\Omega}$ (e.g. using \textsc{FastNorm}) and then discard if we find $\bkself{\Omega} - 1 \gg \delta^2$. Note that normalizing $\ket{\Omega}$ does not solve this problem, as the BBCCGH sparsification results do not tell us about the closeness of $\ket{\psi}$ with $\ket{\Omega}/\norm{\ket{\Omega}}$, only the unnormalized vector $\ket{\Omega}$.}

\jrs{Assuming we have successfully obtained $\ket{\Omega}$ with sufficiently small norm, a second difficulty arises from the right-hand side of \eqref{eq:tailbound}. The probability of success is larger when the stabilizer fidelity is small. In Bravyi \textit{et al.} \cite{bravyi2018simulation} it is argued that the failure probability is negligible for cases of interest where stabilizer fidelity is exponentially small in the number of qubits $n$. Let us unpack this argument by considering a specific case. Assume for the sake of argument that $\bkself{\Omega}$ is very close to 1, so that the expression in square brackets in \eqref{eq:tailbound} is  $\norm{\ket{\psi}-\ket{\Omega}}^2 \lesssim \delta^2$. Now suppose we fix target precision $\delta$, and we want to achieve success probability at least $p$. By rearranging \eqref{eq:tailbound} we see that this is possible only if the stabilizer fidelity satisfies:}
\begin{equation}
    F(\psi) \leq \frac{\delta^2}{8 \log\left(\frac{2}{1-p}\right)}. \label{eq:stabFidLimit}
\end{equation}
\jrs{To make this concrete, let us use the modest assumptions that we want trace norm error $\delta$ no larger than $10\%$, and success probability better than $1/2$. This can be achieved only when $F(\psi) \lesssim 0.0009$. Now consider the case where $\ket{\psi} = \ket{T}^{\otimes t}$, where $\ket{T} = (\ket{0} + e^{i \pi/4} \ket{1})/\sqrt{2}$. For $t$-fold tensor products of $m$-qubit states where $m\leq 3$, stabilizer fidelity is multiplicative \cite{bravyi2018simulation}, so that we have $F(\ket{T}^{\otimes t}) \approx (0.854)^t$. It follows that \eqref{eq:stabFidLimit} is satisfied for these parameters only when we have at least 45 copies of $\ket{T}$. If we want improved accuracy and success probability, the minimum value of $t$ needed to satisfy \eqref{eq:stabFidLimit} increases. Furthermore, the sparsification tail bound has the curious property that it seems to suggest worse performance for states containing less magic, as quantified by the stabilizer fidelity. For example, if instead of the $\pi/8$-state $\ket{T}$ we consider $t$-fold tensor products of the $\pi/32$-state $\ket{\pi/32}= (\ket{0} + e^{i \pi/16}\ket{1})/\sqrt{2}$, we must have at least $t \approx 1200$ before \eqref{eq:tailbound} gives a non-trivial lower bound on success probability. Therefore there is a large class of interesting intermediate-sized quantum circuits for which the BBCCGH sparsification tail bound cannot be applied.}

\jrs{Our improved sparsification results in Section \ref{sec:sparsification} sidestep these difficulties by considering the difference in the trace norm between $\op{\psi}$ and the \textit{ensemble} $\rho_1 = \sum_\Omega \Pr(\Omega) \frac{\op{\Omega}}{\bkself{\Omega}}$ from which sparsified vectors are drawn, rather than fidelity with any particular $\ket{\Omega}$. This allows us to implement classical bit-string sampling with a distribution $\delta$-close to the quantum distribution, even though any particular $\ket{\Omega}$ may not be a good approximation to $\ket{\psi}$. The key idea is that the measurement statistics on the ensemble $\rho_1$ mimic those on $\op{\psi}$; the comparison with any individual $\ket{\Omega}$ is unimportant in the context of bit-string sampling.}

\section{\jrs{Proof of ensemble sampling lemma}}
\label{app:tracenormlemma}
\jrs{Here we prove Lemma \ref{tracenormlemma},  the first of the two lemmata leading to our sparsification result. Given target state ${\ket{\psi} = \sum_j c_j \ket{\phi_j}}$, we need to prove that, for randomly generated sparse vectors $\ket{\Omega}=(\onenorm{c}/k)\sum_\alpha \ket{\omega_\alpha}$ output from \textsc{Sparsify} (Figure \ref{fig:sparsify}), where $\ket{\omega_\alpha}$ are stabilizer states randomly drawn from $\{(c_j/|c_j|)\ket{\phi_j}\}$, the following holds:}
\begin{equation}	
\label{RigorousSparseEq1App}
	\delta_S = \|  \rho_1  - \kb{\psi}{\psi}\|_1 \leq \frac{2 \|c\|_1^2}{k}  + \sqrt{ \mathrm{Var}[\bk{\Omega}{\Omega}] }.
\end{equation}
\jrs{Here $k$ is the number of terms in the sparsified vector $\ket{\Omega}$ and $\rho_1$ is the ensemble over all possible normalized $\ket{\Omega}$:}
\begin{equation}
		\rho_1 := \mathbb{E}\left[  \frac{\kb{\Omega}{\Omega}}{\bk{\Omega}{\Omega}}   \right] = \sum_{\Omega} \Pr(\Omega) \frac{\kb{\Omega}{\Omega}}{\bk{\Omega}{\Omega}}.\label{eq:expectedSparseAPP}
\end{equation}

First we introduce the operator
\begin{equation}
	\rho_2 =  \frac{1}{\mu} \mathbb{E} \left[  \kb{\Omega}{\Omega}   \right] ,
\end{equation}	
where $\mu=\mathbb{E}[\bk{\Omega}{\Omega}]$.   \jrs{Then using the triangle inequality,}
\begin{align}	
	\delta_S &=	\|  \rho_1  + \rho_2-\rho_2 - \kb{\psi}{\psi}\|_1 \\
    &\leq	\|  \rho_1  -\rho_2  \|_1 + \|\rho_2  - \kb{\psi}{\psi}\|_1 .\label{triangle}
\end{align}
Now,
\begin{align}
		\|  \rho_1  -\rho_2 \|_1 & = \|  \mathbb{E}\left[  \frac{\kb{\Omega}{\Omega}}{\bk{\Omega}{\Omega}}   \right] -\frac{ \mathbb{E} \left[  \kb{\Omega}{\Omega}   \right]}{\mu} \|_1 \\ 
		& = \|  \mathbb{E}\left[ \kb{\Omega}{\Omega} \left( \frac{1}{\bk{\Omega}{\Omega}} - \frac{1}{\mu}   \right)    \right]  \|_1 .	
\end{align}	
Using Jensen's inequality we can bring the expectation value outside the norm so that
\begin{align}
\nonumber
\|  \rho_1  -\rho_2 \|_1 & \leq \mathbb{E} \|  \left[ \kb{\Omega}{\Omega} \left( \frac{1}{\bk{\Omega}{\Omega}} - \frac{1}{\mu}   \right)    \right]  \|_1 ,	\\ 
\nonumber
	& =  \mathbb{E}  \left| \bk{\Omega}{\Omega} \left( \frac{1}{\bk{\Omega}{\Omega}} - \frac{1}{\mu}   \right)    \right|   \\ 
	& = \frac{1}{\mu}  \mathbb{E} |  \mu - \bk{\Omega}{\Omega}   |.\label{eq:jensen-outside-bound}
\end{align}	
That $\mu= \mathbb{E}[\bk{\Omega}{\Omega}]=1+(\|c\|_1^2-1)/k$ comes from Ref.~\cite{bravyi2018simulation}.  
Loosening \eqref{eq:jensen-outside-bound} with $\mu^{-1}\leq 1$ gives
\begin{align}
	\|  \rho_1  -\rho_2  \|_1 & \leq  \mathbb{E} |  \mu - \bk{\Omega}{\Omega}   | ,
\end{align}	
which is simply the average deviation of $\bk{\Omega}{\Omega}$ from the mean.  Using Jensen's inequality we get
\begin{align}
    \mathbb{E} |  \mu - \bk{\Omega}{\Omega}   | & \leq \sqrt{ \mathbb{E} |  \mu - \bk{\Omega}{\Omega}   |^2 } \\ \nonumber
    & = \sqrt{ \mathrm{Var}[ \bk{\Omega}{\Omega} ]},
\end{align}
and so
\begin{align} \label{Bound1}
    \|  \rho_1  -\rho_2  \|_1 & \leq   \sqrt{ \mathrm{Var}[ \bk{\Omega}{\Omega} ]} .
\end{align}

Next, we consider the term $\| \rho_2  - \kb{\psi}{\psi} \|_1$, by first finding an explicit form for $\rho_2$. Observe that
\begin{equation}
		\kb{\Omega}{\Omega} = \frac{\| c \|_1^2}{k^2} \sum_{\alpha, \beta}  \kb{\omega_\alpha}{\omega_\beta} ,
\end{equation}	
\jrs{recalling that $\ket{\omega_\alpha} = c_j\ket{\phi_j}/|c_j|$ with probability $|c_j|/\onenorm{c}$, so that $\mathbb{E}(\ket{\omega_\alpha})=\ket{\psi}/\onenorm{c}$ \cite{bravyi2018simulation}. T}aking the expectation value we have
\begin{equation}
\mathbb{E}(\kbself{\Omega}) = \mu \rho_2	 = \frac{\| c \|_1^2}{k^2} \sum_{\alpha, \beta} \mathbb{E}[ \kb{\omega_\alpha}{\omega_\beta} ], \label{eq:splitme}
\end{equation}
Let $\sigma:= \mathbb{E}[ \kb{\omega_\alpha}{\omega_\alpha} ]$.
We split \eqref{eq:splitme} into two summations as follows:
\begin{align}
	\mu \rho_2	& = \frac{\| c \|_1^2}{k^2} \left[\left(\sum_{\alpha \neq \beta} \mathbb{E}[ \kb{\omega_\alpha}{\omega_\beta} ]\right) + \left(\sum_{\alpha } \mathbb{E}[ \kb{\omega_\alpha}{\omega_\alpha} ]\right)\right], \\ 
		& = \frac{\| c \|_1^2}{k^2} \left[\left(\sum_{\alpha \neq \beta} \frac{\kb{\psi}{\psi}}{\|c\|_1^2} ]\right) + \left(\sum_{\alpha } \sigma\right)\right] , \\ \nonumber
		& = \frac{1}{k^2} \left(\sum_{\alpha \neq \beta} \kb{\psi}{\psi}\right) + \frac{\|c \|_1^2}{k^2} \sum_{\alpha } \sigma.	
\end{align}
In the first contribution, we used the independence of $\omega_\alpha$ and $\omega_\beta$ when $\alpha \neq \beta$ and $\mathbb{E}[\ket{\omega_\alpha}]=\ket{\psi}/ \|c\|_1$.  Next, there are $k(k-1)$ terms and $k$ terms in the first and second summations respectively, so that
\begin{align}
	\mu \rho_2	 	& = \left( 1 - k^{-1} \right) \kb{\psi}{\psi} + \frac{\|c \|_1^2}{k} \sigma .
\end{align}
Using this form for $\rho_2$, we have that 
\begin{align}
	\| \rho_2 - \kb{\psi}{\psi} \|_1	 	& =  \mu^{-1}	\| \mu \rho_2 - \mu \kb{\psi}{\psi} \|_1 , \\ \nonumber
 & = \mu^{-1}	 \|  ( 1 - k^{-1} - \mu) \kb{\psi}{\psi} + \frac{\|c \|_1^2}{k} \sigma   \|_1 .
\end{align}
Substituting in the value of $\mu$ we find $1 - k^{-1} - \mu = - \|c\|_1^2 / k$ and so
\begin{align} \label{Bound2}
	\| \rho_2 - \kb{\psi}{\psi} \|_1 & = \frac{\|c \|_1^2}{k \mu}	 \|    \sigma  -  \kb{\psi}{\psi}  \|_1 \leq 2\frac{\|c \|_1^2}{k} ,
\end{align}
where we have used the triangle inequality, $\| \sigma \|_1$ and ${\mu^{-1}\leq 1}$. Substituting Eq.~\eqref{Bound1} and Eq.~\eqref{Bound2} into Eq.~\eqref{triangle}, completes the proof of the lemma.

}
\section{Sparsification variance bound}
\label{AppVarianceBound}

We now prove Lemma \ref{lem:variancelemma}, \jrs{the second lemma leading to} Theorem \ref{NewSparse3}. Recall that given a state $\ket{\psi} = \sum_j c_j \ket{\phi_j}$, where $\ket{\phi_j}$ are stabilizer states, we can obtain a sparsified $k$-term approximation given by:
\begin{equation}
    \ket{\Omega} = \frac{\onenorm{c}}{k} \sum^k_{\alpha=1} \ket{\omega_\alpha} \label{eq:OmegaAgain}
\end{equation}
where each $\ket{\omega_\alpha}$ is chosen randomly so that $\ket{\omega_\alpha} = (c_j/|c_j|)\ket{\phi_j}$ with probability $p_j = |c_j|/\onenorm{c}$. In general $\ket{\Omega}$ may not be conventionally normalized, but Lemma \ref{lem:variancelemma} upper bounds the variance of $\braketself{\Omega}$. We now prove Lem.~\ref{lem:variancelemma}. 

\begin{proof}[Proof of Lem.~\ref{lem:variancelemma}]
In Ref. ~\cite{bravyi2018simulation} it was shown that 
\begin{equation}
  \mu =  \mathbb{E}[\langle \Omega | \Omega \rangle] = \frac{\norm{c}_1^2}{k}  + \frac{\norm{c}_1^2}{k^2} \mathbb{E}\left(B\right), \label{eq:meanOmega}
\end{equation}
where $B = \sum_\alpha \sum_{\beta \neq \alpha} \langle\omega_\alpha|\omega_\beta\rangle$. 
Since $\ket{\omega_\alpha}$ and $\ket{\omega_\beta}$ are independently sampled for distinct $\alpha$ and $\beta$, we get
\begin{equation}
\mathbb{E}\left(\langle \omega_\alpha \vert \omega_\beta \rangle\right) = \mathbb{E}(\bra{\omega_\alpha}) \mathbb{E}(\ket{\omega_\beta}) = \frac{ \langle \psi \vert \psi \rangle}{\onenorm{c}^2}. \label{eq:alphabetaExpect}
\end{equation}
We use similar proof techniques to bound $\mathbb{E}[\langle \Omega \vert \Omega \rangle^2]$, and in turn bound the variance.
We begin with
\begin{align}
    \braketself{\Omega}^2 & = \frac{\norm{c}_1^4}{k^4}\left(\sum_{\alpha,\beta}\braket{\omega_\alpha}{\omega_\beta} \right)^2 , \\
                    & = \frac{\norm{c}_1^4}{k^4} \left(\sum_\alpha\left( \braketself{\omega_\alpha} + \sum_{\beta \neq \alpha}\braket{\omega_\alpha}{\omega_\beta}   \right)\right)^2 , \\
                    & = \frac{\norm{c}_1^4}{k^4} ( k + B)^2 , \\
                    & = \frac{\norm{c}_1^4}{k^4} ( k^2 + 2kB + B^2),
\end{align}
where in the second line we note that there are $k$ terms in the summation. Whereas from Eq.~\eqref{eq:meanOmega} we have
\begin{equation}
    \expected{\braketself{\Omega}}^2 = \frac{\norm{c}_1^4}{k^4} (k^2  + 2 k\mathbb{E}(B) +\mathbb{E}(B)^2) \label{eq:Squareofmean}.
\end{equation}
Comparing these expressions, for the variance we obtain
\begin{align}
    \VarOmega & = \expected{\braketself{\Omega}^2} -  \expected{\braketself{\Omega}}^2 , \\
    &= \frac{\norm{c}_1^4}{k^4} (\mathbb{E}(B^2) - \mathbb{E}(B)^2). \label{eq:twoVariances}
\end{align}
By counting terms in the summation $B$, and using the relation \eqref{eq:alphabetaExpect}, we find
\begin{equation}
    \mathbb{E}(B)^2 = \frac{k^2(k-1)^2}{\norm{c}_1^4} \label{eq:WSquareOfMean}.
\end{equation}
Expanding $B^2$, we get
\begin{align}
    B^2 &= \left(\sum_\alpha \sum_{\beta \neq \alpha} \braket{\omega_\alpha}{\omega_\beta}\right)\left(\sum_\lambda \sum_{\mu \neq \lambda} \braket{\omega_\lambda}{\omega_\mu}\right) \label{eq:Wexpansion}\\
        &= \sum_{(\alpha,\beta,\lambda,\mu) \in \mathcal{A}} \braket{\omega_\alpha}{\omega_\beta} \braket{\omega_\lambda}{\omega_\mu} + B'
\end{align}
where $\mathcal{A}$ denotes the set of all possible combinations $(\alpha,\beta,\lambda,\mu)$ where all four indices are distinct, and $B'$ denotes the remaining terms where at least two of the indices are the same. Now, if $(\alpha,\beta,\lambda,\mu)$ are all distinct, then $\braket{\omega_\alpha}{\omega_\beta}$ and $\braket{\omega_\lambda}{\omega_\mu}$ are independent random variables, so $\mathbb{E}(\braket{\omega_\alpha}{\omega_\beta} \braket{\omega_\lambda}{\omega_\mu} ) = \mathbb{E}(\braket{\omega_\alpha}{\omega_\beta})\mathbb{E}(\braket{\omega_\lambda}{\omega_\mu})$. This yields
\begin{equation}
    \mathbb{E}(B^2) = \frac{k(k-1)(k-2)(k-3)}{\norm{c}_1^4} + \mathbb{E}(B'). \label{eq:WmeanOfSquare}
\end{equation}
Substituting the expressions \eqref{eq:WSquareOfMean} and \eqref{eq:WmeanOfSquare} back into \eqref{eq:twoVariances}, we obtain
\begin{align}
    \VarOmega & = \frac{\norm{c}_1^4}{k^4}\mathbb{E}(B') - \frac{k(k-1)(4k-6)}{k^4}. \label{eq:nearlyVariance}
\end{align}
We must now consider terms $\braket{\omega_\alpha}{\omega_\beta} \braket{\omega_\lambda}{\omega_\mu}$ in the expansion of $B^2$ where $(\alpha,\beta,\lambda,\mu)$ are not all distinct. We use the notation $B_{j=k}$ to indicate the sum of all terms where indices $j$ and $k$ are equal but all others are distinct, e.g. $B_{\lambda=\alpha}  = \sum_{\alpha,\beta,\mu}  \braket{\omega_\alpha}{\omega_\beta} \braket{\omega_\alpha}{\omega_\mu} $, where the summation is over terms such that $\alpha$, $\beta$ and $\mu$ are all distinct, and so on. There are $k(k-1)(k-2)$ terms in each summation of this type. Similarly for the terms sharing two pairs of indices, we use the notation $B_{\lambda=\alpha;\mu = \beta} = \sum_{\alpha\neq\beta}  \braket{\omega_\alpha}{\omega_\beta}\braket{\omega_\alpha}{\omega_\beta}$. These summations comprise of $k(k-1)$ terms. 
From Eq.~\eqref{eq:Wexpansion}, we never have terms where $\alpha=\beta$ or $\lambda = \mu$. We can therefore write
\begin{align}
    B' &= B_{\lambda=\alpha} +  B_{\mu=\alpha} + B_{\lambda =\beta} + B_{\mu = \beta} \nonumber \\
    & \quad \quad + B_{\lambda=\alpha;\mu=\beta} + B_{\mu=\alpha;\lambda = \beta}.
\end{align}
One can check that $\mathbb{E}[B_{\lambda=\alpha}^*]=\mathbb{E}[B_{\mu=\beta}] $ and $\mathbb{E}[B_{\mu=\alpha}^*]=\mathbb{E}[B_{\lambda=\beta}] $.  Therefore 
\begin{align}
\expected{B'} &= 2 \re\{\expected{B_{\lambda = \beta}} +  \expected{B_{\mu=\beta}}\} \nonumber \\
    & \quad \quad \quad + \expected{B_{\lambda = \alpha;\mu=\beta} + B_{\mu=\alpha;\lambda=\beta}}. \label{eq:Bprime}
\end{align}
Next we note that
\begin{align}
		\mathbb{E}[B_{\lambda=\beta}] & = \sum_\alpha \sum_{\beta \neq \alpha } \sum_{\alpha \neq \mu \neq \beta} \mathbb{E}[  \bk{\omega_{\alpha}}{\omega_{\beta}}  \bk{\omega_{\beta}}{\omega_{\mu}} ] \\ \nonumber
			& = k(k-1)(k-2) \mathbb{E}[  \bra{\omega_{\alpha}}  ]   \mathbb{E}[  \kb{\omega_{\beta}}{\omega_{\beta}}  ] \mathbb{E}[  \ket{\omega_{\mu}}  ]  \\ 
				& = \frac{k(k-1)(k-2)}{\|c\|_1^2}  \bra{\psi}   \sigma \ket{\psi}, \label{eq:Bbetalambda}
\end{align}	 
where $\sigma = \mathbb{E}[  \kb{\omega_{\beta}}{\omega_{\beta}}] = \sum_j (|c_j|/\onenorm{c}) \kb{\phi_j}{\phi_j}$, since the probability of sampling $\kb{\omega_\beta}{\omega_\beta} = \kb{\phi_j}{\phi_j}$ is defined as $p_j = |c_j|/\onenorm{c}$. Next we consider $\expected{B_{\mu=\beta}}$. Taking the modulus and using the triangle inequality we obtain
\begin{align}
    |\expected{B_{\mu=\beta}}| &\leq \sum_\alpha \sum_{\beta \neq \alpha } \sum_{\alpha \neq \lambda \neq \beta} \expected{|\bk{\omega_\alpha}{\omega_\beta}\bk{\omega_\lambda}{\omega_\beta}|} \\
    	& = k(k-1)(k-2)  \mathbb{E}[  \bk{\omega_{\alpha}}{\omega_{\beta}}  \bk{\omega_{\beta}}{\omega_{\lambda}} ] \\
    	& = \frac{k(k-1)(k-2)}{\|c\|_1^2}  \bra{\psi}   \sigma \ket{\psi}. \label{eq:Bbetamu}
\end{align}
Similarly, for the last two terms ${B'' = B_{\lambda = \alpha;\mu=\beta} + B_{\mu=\alpha;\lambda=\beta}}$, we obtain
\begin{align}
    |\expected{B''}| & \leq \sum_{\alpha} \sum_{\beta \neq \alpha} \expected{| \bk{\omega_\alpha}{\omega_\beta}\bk{\omega_\alpha}{\omega_\beta}|} \\
    & \quad \quad +\sum_{\alpha} \sum_{\beta \neq \alpha} \expected{ \bk{\omega_\alpha}{\omega_\beta}\bk{\omega_\beta}{\omega_\alpha}} \\
    & = 2 \sum_{\alpha} \sum_{\beta \neq \alpha} \expected{ \bk{\omega_\alpha}{\omega_\beta}\bk{\omega_\beta}{\omega_\alpha}}.
\end{align}
Using cyclicity of the trace get
\begin{align}
    \expected{ \bk{\omega_\alpha}{\omega_\beta}\bk{\omega_\beta}{\omega_\alpha}} 
            & = \expected{\Tr[\bk{\omega_\alpha}{\omega_\beta}\bk{\omega_\beta}{\omega_\alpha}]} \\
            & = \Tr[\expected{\ket{\omega_\alpha}\bk{\omega_\alpha}{\omega_\beta}\bra{\omega_\beta}}] \\
            & = \Tr[\expected{\kbself{\omega_\alpha}}\expected{\kbself{\omega_\beta}}] \\
            & = \Tr[\sigma^2],
\end{align}
so that
\begin{equation}
    |\expected{B''}| \leq 2k(k-1) \Tr[\sigma^2] \leq 2k(k-1) \label{eq:Bprimeprime}.
\end{equation}
Combining the results \eqref{eq:Bprime}, \eqref{eq:Bbetalambda}, \eqref{eq:Bbetamu} and \eqref{eq:Bprimeprime} gives us
\begin{align}
    \expected{B'} & \leq 4\frac{k(k-1)(k-2)}{\|c\|_1^2}  \bra{\psi}   \sigma \ket{\psi} + 2k(k-1)  .
\end{align}
Writing 
\begin{equation}
\label{CliffordnessEquation}
    C = \onenorm{c}^2  \bra{\psi}   \sigma \ket{\psi} = \norm{c}_1 \sum_j |c_j| |\bk{\psi}{\phi_j}|^2
\end{equation}
and substituting the expression for $\mathbb{E}(B')$ into Eq.~\eqref{eq:nearlyVariance} we obtain
\begin{align} \label{varianceExact}
\VarOmega  \leq & 4\frac{k^3 - 3k^2 +2k}{k^4} C + 2\frac{\onenorm{c}^4}{k^2}\left(1 - \frac{1}{k} \right)  \\
                        & - \frac{4k^3 - 10 k^2 + 6k}{k^4} \nonumber,
\end{align}
which to leading order in $1/k$ is
\begin{align}
\VarOmega & \leq  \frac{4(C-1)}{k} + 2\left(\frac{\norm{c}_1^2}{k}\right)^2+ \mathcal{O}\left(\frac{C}{k^3}\right),
\end{align}
which gives us the general bound appearing in Lemma \ref{lem:variancelemma}. 

Clifford magic states were defined in Ref. \cite{bravyi2018simulation} as those pure states $\ket{\psi}$ that are stabilized by a group $\mathcal{Q}$ of Clifford unitary operators whose generators take the form $UX_jU^\dagger$, where $X_j$ is the Pauli $X$ operator that acts on the $j$-th qubit. For such states, there exists~\cite{bravyi2018simulation} an optimal decomposition 
\begin{equation}
    \ket{\psi} = \sum_{q \in \mathcal{Q}} c_q \ket{\phi_q} = \frac{1}{|\mathcal{Q}| \bk{\psi}{\phi_0}} \sum_{q \in \mathcal{Q}} q \ket{\phi_0},
\end{equation}
where $\ket{\phi_0}$ is some stabilizer state that achieves the maximize possible value for $|\bk{\psi}{\phi_0}|$. If we take this decomposition as the basis for our sparsification, then we have
\begin{equation}
\onenorm{c} = |\mathcal{Q}| \cdot (|\mathcal{Q}| |\bk{\psi}{\phi_0}|)^{-1} = |\bk{\psi}{\phi_0}|^{-1}
\end{equation}
and 
\begin{equation}
    \sigma = \sum_{q \in \mathcal{Q}} p_q q \kbself{\phi_0} q^\dagger,
\end{equation}
where $p_q = |\mathcal{Q}|^{-1}$. This yields 
\begin{align}
    \bra{\psi}   \sigma \ket{\psi} & = \sum_{q \in \mathcal{Q}} p_q \bra{\psi}q \kbself{\phi_0} q^\dagger \ket{\psi}  \\
                            & = \sum_{q \in \mathcal{Q}} p_q \bk{\psi}{\phi_0} \bk{\phi_0}{\psi}  \\
                            & = |\bk{\psi}{\phi_0}|^2 = \frac{1}{\onenorm{c}^2},
\end{align}
where in the second line we used the Hermiticity of $q$ and $q\ket{\psi} = \ket{\psi}$. This shows that for optimal decompositions of Clifford magic states, $C = 1$, and leads to the simplified bound
\begin{align}
\VarOmega & \leq  2 \left(\frac{\onenorm{c}^2}{k}\right)^2 + \frac{2}{k^3}.
\end{align}
\end{proof}
\jrs{Finally, we comment on the effect of the constant $C$ when $\ket{\psi}$ is not a Clifford magic state. Recall that $C$ can be written in terms of this expected overlap, $C = \onenorm{c}^2 \mathbb{E}\left[|\bk{\psi}{\omega}|^2\right]$, and enters into Thm. \ref{NewSparse3} via the critical precision $\delta_c = 8(C-1)/\onenorm{c}^2$.  Consider $\ket{\psi} = \ket{\psi'}^{\otimes N}$ where $\ket{\psi'}$ are pure states. 
When $\ket{\psi}$ is a product of $N$ pure states, we can write each randomly sampled stabilizer state as $\ket{\omega} = \bigotimes_{\alpha=1}^N \ket{\omega_\alpha}$, where $\ket{\omega_\alpha}$ are i.i.d. random vectors. It follows that $\mathbb{E}\left[|\bk{\psi}{\omega}|^2\right] = (\mathbb{E}\left[|\bk{\psi'}{\omega_\alpha}|^2\right])^N$.Since $\ket{\omega_\alpha}$ are always stabilizer states, when $\ket{\psi'}$ are non-stabilizer states, we have $|\bk{\psi'}{\omega_\alpha}|^2 < 1$. Therefore the threshold precision $\delta_c< 8C/\|c\|_1^2 = 8(\mathbb{E}\left[|\bk{\psi'}{\omega_\alpha}|^2\right])^N $ vanishes for large $N$ when $\ket{\psi}$ is a tensor product of $N$ pure states. Moreover, in Figure \ref{fig:Cliffordness} we plot values of $C$ for a class of single-qubit states, showing that $C-1$ is close to zero even when $N$ is not large.}

\begin{figure}[t]
    \centering
    \includegraphics{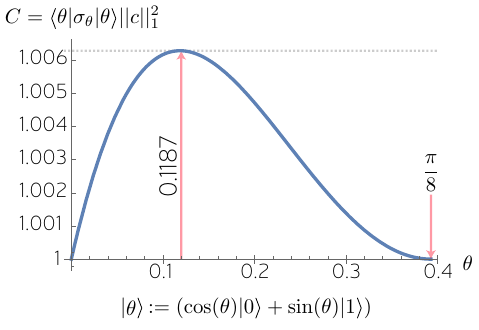}
    \caption{The variable $C$ as introduced in Eq.~\eqref{CliffordnessEquation} as a function of the angle $\theta$ for a class of single-qubit states.  This is the $C$ value for one copy of the state, for $n$ copies we must raise to the $n^{\mathrm{th}}$ power.  The prefactor $C-1$ appears in Eq.~\eqref{varianceExact} and is important because when $C=1$, the variance scales asymptotically as $\mathcal{O}(1/k^2)$.  We highlight two specific angles $\theta=\{ \pi / 8 , 0.1187 \}$ that correspond to angles used in Fig.~\ref{fig:DeltaScaling}.  For $\theta = \pi/ 8$, we have $C-1=0$ and so the $\mathcal{O}(1/k^2)$ is exact as can be seen in Fig.~\ref{fig:DeltaScaling}.  For $\theta = 0.1187$, we have the maximal possible value of $C$ and Fig.~\ref{fig:DeltaScaling} shows the maximal deviation from  $\mathcal{O}(1/k^2)$ scaling.}
    \label{fig:Cliffordness}
\end{figure}

\prx{%
\section{Bit-string sampling simulator technical details}
\label{app:bitsampling} 
\begin{figure}[t]
\begin{algorithm}[H]
\caption{Bit-string sampling for mixed states\label{algo:bitstring}}
\begin{algorithmic}[1]
  \algrenewcommand\algorithmicrequire{\textbf{Input:}}
    \algrenewcommand\algorithmicensure{\textbf{Output:}}
    \Require Decomposition $\rho = \sum_j p_j \op{\psi_j}$, where for each $\psi_j$ we have a known stabilizer state decomposition $\ket{\psi_j} = \sum_r c^{(j)}_r \ket{\phi^{(j)}_r}$. Real numbers $\delta_S$, $p_{\rm FN}$ and $\epsilon_{\rm FN}$. Number of bits $w$.
    \Ensure String $\vec{x}$ of length $w$, sampled from a distribution $P''(\vec{x}) = \sum_j p_j P_j(\vec{x})$, which approximates $P(\vec{x}) = \Tr(\Pi_{\vec{x}} \rho)$, where $\Pi_{\vec{x}} = \kb{\vec{x}}{\vec{x}} \otimes \id_{n-w} $.
    \State Select index $j$ with probability $p_j$. \label{step:sampleFromEnsemble}
    \State $k \gets \lceil 4 \norm{c^{(j)}}^2_1\delta_S^{-1}\rceil$ \label{step:kselect}
    \State $\ket{\Omega'} \gets$ \Call{Sparsify}{$\ket{\psi_j}$,$k$}
    \State $W \gets$ \Call{FastNorm}{$\ket{\Omega'}$,$p_{\rm FN}$,$\epsilon_{\rm FN}$} \label{step:FN1st}
    \State $\ket{\Omega} \gets \ket{\Omega'}\!/\sqrt{W}$
    \State $\vec{x} \gets ()$ (initialize empty string)
    \State $P_{\vec{x}} \gets 1 $
    \For{$b \gets 1$ \textbf{to} $w$}
        \State $P_{{(\vec{x},0)}}\gets$ \label{step:FN2nd} \Call{FastNorm}{$\Pi_{(\vec{x},0)}\ket{\Omega}$,$p_{\rm FN}$,$\epsilon_{\rm FN}$}
        \State $P(x_b = 0|\vec{x})  \gets P_{{(\vec{x},0)}}/P_{x}$ 
        \If{$P_{x_b = 0} < 1/2$} 
            \State $P(x_b = 1|\vec{x})  \gets 1 - P_{x_b = 0}$ 
            \Else
            \State $P_{{(\vec{x},1)}}\gets$ \Call{FastNorm}{$\Pi_{(\vec{x},1)}\ket{\Omega}$,$p_{\rm FN}$,$\epsilon_{\rm FN}$}\label{step:FN3rd}
            \State $P(x_b = 1|\vec{x})  \gets P_{{(\vec{x},1)}}/P_{\vec{x}}$
            \State $P(x_b = 0|\vec{x}) \gets 1 - P_{x_b = 1}$ 
        \EndIf
        \State Select $y \in \{0,1\}$ with probability $P(x_b = y|\vec{x})$, then $x_b \gets y$. \label{step:bitsampling}
        \State $P_{(\vec{x},x_b)} \gets P_{\vec{x}} \times P(x_b|\vec{x})$
        \State $\vec{x} \gets (\vec{x},x_b)$
    \EndFor
    \State \Return $\vec{x}$
\end{algorithmic}
\end{algorithm}
\vspace*{-\baselineskip}
\end{figure}

\jrs{In this appendix, we give full pseudocode for our bit-string sampling simulator (Algorithm \ref{algo:bitstring}), prove its validity as a method to classically emulate sampling from the quantum distribution $P(\vec{x})=\Tr[\Pi_{\vec{x}}\rho]$, and analyze its runtime. This constitutes a proof of Thm. \ref{thm:bit-string-sampling}.  As described in the main text, Algo.~\ref{algo:bitstring} draws bit strings $\vec{x}$ from a classical distribution $\Psim(\vec{x})$, using two subroutines from Ref. \cite{bravyi2018simulation}, \textsc{Sparsify} and \textsc{FastNorm}.  As sketched in the main text, our strategy is to define an idealized algorithm $\textsc{Exact}$ where calls to \textsc{FastNorm} are replaced by an oracle which can compute $\norm{\Pi_\vec{y} \ket{\Omega}}$ exactly for any un-normalized $\ket{\Omega}$ and bit string $\vec{y}$. The algorithm \textsc{Exact} draws from a distribution $\Pex(\vec{x})$. We first show that $\Pex$ is $\delta_S$-close to the quantum distribution $P$. We then argue that the distribution $\Psim$ that Algo. \ref{algo:bitstring} draws from is $\epsilon$-close to $\Pex$. Finally we optimize the choice of $\delta_S$ and $\epsilon$ and analyze the runtime.}

\jrs{\textsc{Exact} is identical to our Algo. \ref{algo:bitstring}, except where our algorithm estimates probabilities $\norm{\Pi_{\vec{y}}\ket{\Omega}}^2$ using \textsc{FastNorm}, \textsc{Exact} computes them exactly. Therefore \textsc{Exact} first samples a state $\ket{\psi_j}$ from the ensemble with probability $p_j$,
and chooses a sparsification $\ket{\Omega} = \ket{\Omega_{j,l}}$ with probability $q_{j,l} = \Pr(\Omega_{j,l}|\psi_j)$. Given the selected $\ket{\Omega}$, a bit string is sampled by choosing each bit in turn via a series of conditional probabilities:}
\begin{align}
    \Pr(\vec{x}|\Omega) & = \mathrm{Pr}(x_1) \mathrm{Pr}(x_2|\vec{x}_1) \ldots \mathrm{Pr}(x_w|\vec{x}_{w-1}) \label{eq:chainconditional} \\
            & = \frac{\norm{\Pi_\vec{x}\ket{\Omega}}^2}{\norm{\ket{\Omega}}^2}=\Tr\left[\Pi_\vec{x} \frac{\op{\Omega}}{\bk{\Omega}{\Omega}}\right]. 
\end{align}
\jrs{Here we use the notation $\vec{x}_m$ to denote the string comprised of the first $m$ bits of $\vec{x}$, so that $\Pi_{\vec{x}_m} = \bigotimes_{j=1}^m \kbself{x_j} \otimes \id_{n-m}$. We take $\vec{x}_0$ to be the empty string, so that $\Pi_{\vec{x}_0} = \id$. The probability of choosing $y\in{0,1}$ for the $m$-th bit, given $m-1$ bits already sampled, is computed as:}
\begin{equation}
    \Pr(y|\vec{x}_{m-1})=\norm{\Pi_{(x_1,\ldots,x_{m-1},y)}\ket{\Omega}}^2/\norm{\Pi_{\vec{x}_{m-1}}\ket{\Omega}}.
\end{equation} 
\jrs{Thus \textsc{Exact}
outputs bit strings \textsc{x} sampled from a distribution:}
 \begin{align}
    \Pex(\vec{x}) & = \sum_j \sum_l p_j q_{j,l}\frac{\norm{\Pi_{\vec{x}}\ket{\Omega_{j,l}}}^2}{\norm{\ket{\Omega_{j,l}}}^2} \label{eq:PprimeDecomp}\\ \nonumber
                & = \sum_j \sum_l p_j q_{j,l} \frac{\Tr[\Pi_{\vec{x}}\op{\Omega_{j,l}}]}{\braketself{\Omega_{j,l}}}  \\ \nonumber
                & = \Tr\left[ \Pi_{\vec{x}}\sum_j p_j \sum_l q_{j,l} \frac{\op{\Omega_{j,l}}}{\braketself{\Omega_{j,l}}} \right] \\ \nonumber
                & = \Tr\left[\Pi_{\vec{x}} \sum_j p_j \mathbb{E}\left(\frac{\op{\Omega_j}}{\braketself{\Omega_j}}\right)\right] = \Tr[\Pi_{\vec{x}} \rho'],
\end{align}
where $\rho' = \sum_j p_j \rho_1^{(j)}$, and each $\rho_1^{(j)}$ given by:
\begin{equation}
		\rho_1^{(j)} :=  \sum_{\Omega} \Pr(\Omega|\psi_j) \frac{\kb{\Omega}{\Omega}}{\bk{\Omega}{\Omega}}.
\end{equation}
\jrs{In other words $\rho_1^{(j)}$ is the expected sparsification given target pure state $\ket{\psi_j}$, as defined in Eq.~\eqref{eq:expectedSparse}. In step \ref{step:kselect}, $k$ is chosen so that by Thm. \ref{NewSparse3}, we 
have $\norm{\rho_1^{(j)}- \op{\psi_j}}_1 \leq \delta_S + \mathcal{O}(\delta_S^2)$, provided $\delta_S \geq \delta_c$, where $\delta_c$ is the critical precision. We will return to the $\delta_S < \delta_c $ case at the end of this appendix.} By the triangle inequality we have
\begin{align}
    \norm{\rho' - \rho }_1 & = \norm{\sum_j p_j \rho_1^{(j)} - \sum_j p_j \op{\psi_j}}_1 \\
                        & \leq \sum_j p_j \norm{\rho_1^{(j)}- \op{\psi_j}}_1 \\
                        & \leq \sum_j p_j [\delta_S +  \mathcal{O}(\delta_S^2)] = \delta_S  + \mathcal{O}(\delta_S^2).
\end{align}
\jrs{Since $\Pex(\vec{x}) = \Tr[\Pi_{\vec{x}} \rho']$ and for the quantum distribution we have $P(\vec{x}) = \Tr[\Pi_{\vec{x}} \rho']$, It follows that $\onenorm{\Pex - P} \leq  \delta_S  + \mathcal{O}(\delta_S^2)$.}

It remains to show that \jrs{ using a sequence of calls to \textsc{FastNorm}, Algo. \ref{algo:bitstring} generates probability distributions $\Psim(x)$ that well approximate $\Pex(\vec{x})$, where}
\begin{equation}
    \Psim(\vec{x}) = \sum_j p_j q_{j,l} Q_{j,l}(\vec{x}) . \label{eq:Pdoubleprime}
\end{equation}
Here each $Q_{j,l}(\vec{x})$ is the probability of Algo.~\ref{algo:bitstring} returning $\vec{x}$ given the sparsification $\ket{\Omega_{j,l}}$. We now drop the subscript as we consider a single sparsification $\ket{\Omega}$. \jrs{Recall that \textsc{FastNorm} takes as input error parameters $p_{\rm FN}$ and $\epsilon_{\rm FN}$, and  un-normalized vectors $\Pi_ \vec{y}\ket{\Omega}$ with known $k$-term stabilizer decomposition. Then with probability $(1-p_{\rm FN})$ it outputs a random variable $\eta$ that approximates $\norm{\Pi_{\vec{y}}\ket{\Omega}}^2$ to within a multiplicative error of $\epsilon_{\rm FN}$:}
\begin{equation}
        (1 - \epsilon_{\rm FN}) \norm{\Pi_{\vec{y}}\ket{\Omega}}^2 \leq \eta \leq(1 + \epsilon_{\rm FN}) \norm{\Pi_{\vec{y}}\ket{\Omega}}^2.
\end{equation}
\jrs{Algo. \ref{algo:bitstring} approximates the chain of conditional probabilities \ref{eq:chainconditional} by calls to \textsc{FastNorm}. The probability of choosing $y\in\{0,1\}$ for the $m$-th bit of \vec{x}, conditioned on the first $m-1$ bits being $\vec{x}_{m-1}$ is therefore bounded as:}
\begin{equation}
     \epsilon_-\frac{\norm{\Pi_{(\vec{x}_{m-1},y)} \ket{\Omega}}^2}{\norm{\Pi_{\vec{x}_{m-1}}\ket{\Omega}}^2} \leq  \Pr(y|\vec{x}_{m-1}) \leq \epsilon_+\frac{\norm{\Pi_{(\vec{x}_{m-1},y)} \ket{\Omega}}^2}{\norm{\Pi_{\vec{x}_{m-1}}\ket{\Omega}}^2},
     \nonumber
\end{equation}  
with probability $(1 - p_{\rm FN})^2$, where
\begin{equation}
\epsilon_{\pm} = \frac{1\pm\epsilon_{\rm FN}}{1\mp \epsilon_{\rm FN}}.
\end{equation}
So, given a particular sparsification $\ket{\Omega}$, the $w$-bit string $\vec{x}$ is sampled from a distribution $Q(\vec{x})$ which satisfies
\begin{equation}
    \prod_{m=1}^w \epsilon_-\frac{\norm{\Pi_{\vec{x}_{m}} \ket{\Omega}}^2}{\norm{\Pi_{\vec{x}_{m-1}}\ket{\Omega}}^2} \leq Q(\vec{x}) \leq \prod_{m=1}^w \epsilon_+\frac{\norm{\Pi_{\vec{x}_{m}} \ket{\Omega}}^2}{\norm{\Pi_{\vec{x}_{m-1}}\ket{\Omega}}^2}
    \nonumber
\end{equation}
with probability at least $(1 - p_{\rm FN})^{2w}$. This simplifies to
\begin{equation}
   \frac{(1-\epsilon_{\rm FN})^w\norm{\Pi_{\vec{x}} \ket{\Omega}}^2}{(1+\epsilon_{\rm FN})^w \norm{ \ket{\Omega}}^2} \leq Q(\vec{x}) \leq \frac{(1+\epsilon_{\rm FN})^w\norm{\Pi_{\vec{x}} \ket{\Omega}}^2}{(1-\epsilon_{\rm FN})^w \norm{ \ket{\Omega}}^2}.\label{eq:algoProbBounds}
\end{equation}
One can check that $(1+\epsilon_{\rm FN})^w/(1-\epsilon_{\rm FN})^w \leq 1 + 3w\epsilon_{\rm FN}$, whenever $\epsilon_{\rm FN} \leq 1/5$, and the analogous result holds for the lower bound. Therefore $Q_{j,l}(\vec{x})$ approximates ${\norm{\Pi_{\vec{x}}\ket{\Omega_{j,l}}}^2}/{\norm{\ket{\Omega_{j,l}}}^2}$ up to multiplicative error $3w\epsilon_{\rm FN}$. Comparing \eqref{eq:PprimeDecomp} with \eqref{eq:Pdoubleprime}, we therefore obtain:
\begin{equation}
    (1 - 3 w\epsilon_{\rm FN})\Pex(\vec{x}) \leq \Psim(\vec{x}) \leq (1 + 3 w\epsilon_{\rm FN}) \Pex(\vec{x})
\end{equation}

If we want to bound the total multiplicative error due to the sequence of calls to \textsc{FastNorm} to $\epsilon$, then we must set $\epsilon_{\rm FN} = \epsilon/(3w)$. It then follows that 
\begin{equation}
    \onenorm{\Psim - \Pex} \leq \epsilon.
    \label{eq:L1normcloseFN}
\end{equation}
In the first part of the proof we showed that $\onenorm{\Pex - P}\leq \delta_S + O(\delta_S^2)$ (provided we are above the critical precision threshold $\delta_c$). Combined with Eq. \eqref{eq:L1normcloseFN}, we obtain 
\begin{equation}
    \onenorm{\Psim - P} \leq \epsilon + \delta_S +  \mathcal{O}(\delta_S^2),
\end{equation}
where $P(\vec{x}) = \Tr[\Pi_{\vec{x}}\rho]$.

Similarly the error bound given above is only obtained with probability $(1-p_{\rm FN})^{2w} \approx 1 - 2wp_{\rm FN}$, so to obtain the above closeness in $\ell_1$-norm, with failure probability at most $p_\mathrm{fail}$, we must set $p_{\rm FN} = p_\mathrm{fail} / (2w)$. 
If we select the state $\ket{\psi_j}$ in step \ref{step:sampleFromEnsemble}, then $k\leq 4 \norm{c^{(j)}}_1^2\delta_S^{-1} + 1 $. To return a single bit-string $\vec{x}$ there are at most $2w$ calls to \textsc{FastNorm}, so the runtime is $\mathcal{O}(wkn^3\epsilon_{\rm FN}^{-2} \log p_{\rm FN}^{-1})= \mathcal{O}(w^3n^3\norm{c^{(j)}}_1^2 \delta_S^{-1}\epsilon^2\log(w/p_{\mathrm{fail}})  )$. 
Recall that the statement of the theorem defined the quantity ${\widetilde{\Xi}= \sum_j p_j \onenorm{c^{(j)}}^2}$\jrs{, so that the time $T$ to obtain a single bit string is non-deterministic. The \emph{expected} (average-case) runtime is $\mathcal{O}(w^3n^3\widetilde{\Xi} \delta_S^{-1}\epsilon^2\log(w/p_{\mathrm{fail}})  )$. }
If the decomposition is optimal with respect to the monotone $\Xi$, then we have $\widetilde{\Xi} = \Xi(\rho)$ and the \jrs{average-case} runtime is $\mathcal{O}(\Xi(\rho))$. For equimagical states, $\Xi(\rho) = \xi(\psi_j)$ for all $j$, and this expression becomes the \jrs{worst-case} runtime.

We now optimize the choice of $\delta_S$ and $\epsilon$. Setting the total error budget $\delta = \delta_S + \epsilon$, by inspecting the runtime we find that the best constant is obtained by setting $\delta_S = \delta/3$ and $\epsilon = 2\delta/3$. The constraint $\delta_S \geq 8D$ therefore becomes $\delta \geq 24D$. Substituting the optimal choice of $\delta_S $ and $\epsilon$ into the expected runtime, we obtain
    \begin{equation}
        \mathbb{E}(T) = \mathcal{O}(w^3 n^3\widetilde{\Xi} \delta^{-3}\log(w/p_{\mathrm{fail}})).
    \end{equation}

\jrs{The above holds for the case where the sparsification error $\delta_S $ is no smaller than a critical value $\delta_c = 8(C_j-1)/ \|c^{(j)}\|_1^2$, where $C_j =\norm{c^{(j)}}_1 \sum_r |c_r| |\bk{\psi_j}{\phi_r}|^2$ is defined for the randomly chosen pure state $\ket{\psi_j}$. Therefore, to ensure we are above the critical error regime for any $\ket{\psi_j}$, we can require that $\delta_S \geq 8D$, where $D = \mathrm{max}\{(C_j - 1)/\onenorm{c^{(j)}}^2\}$. This entails $\delta \geq 24D$ for the overall precision. }

Now suppose that we want to achieve arbitrary precision, \jrs{$\delta < 24D$}. In this regime, one can amend the expression for $k$ in step \ref{step:kselect} to achieve any desired precision, at the cost of slightly poorer scaling in the runtime.  We first use lemmata \ref{tracenormlemma} and \ref{lem:variancelemma} to obtain a sharpened bound on the sparsification error:
\begin{equation}
    \delta_S \leq 2 \frac{\onenorm{c^{(j)}}^2}{k} + \sqrt{\frac{\|c^{(j)}\|_1^2}{k}}\sqrt{4D + 2 \frac{\|c^{(j)}\|_1^2}{k} + \mathcal{O}\left(\frac{1}{k^2}\right)}.
\end{equation}
When $\delta_S \ll 8D$, we can achieve a precision of $\delta_S$ by choosing
\begin{equation}
    k \approx 4 \|c^{(j)}\|_1^2 \left(\frac{D}{\delta_S^2} + \frac{1}{\delta_S} \right)+ \mathcal{O}(1).
\end{equation}
Substituting the revised expression for $k$ into the expected runtime, with $\delta_S = \delta/3$ and $\epsilon = 2\delta/3$, we obtain:
    \begin{equation}
        \mathbb{E}(T) = \mathcal{O}(w^3 n^3\widetilde{\Xi} (\delta^{-3} + 3D\delta^{-4})\log(w/p_{\mathrm{fail}})).
    \end{equation}
    
Here we recover the same asymptotic $\delta^{-3}$ scaling as derived from the original BBCCGH sparsification lemma \cite{bravyi2018simulation}.
However, the prefactor from this prior work was two, whereas our prefactor $D$ is typically exponentially small in the number of qubits (see Appendix \ref{AppVarianceBound}). Therefore, at intermediate precision, the $\delta^{-4}$ term may still dominate.
When the target precision $\delta$ is too small, our bound on the required $k$ exceeds the number of terms in the exact decomposition of $\ket{\psi}$ (i.e. the decomposition achieving the stabilizer rank $\chi(\psi)$). In this scenario, using a sparsified approximation in both our approach and in \cite{bravyi2018simulation} has no benefit, 
and one should instead use an exact decomposition without any sparsification.

}

\end{document}